\pdfoutput=1
\documentclass[preprint]{elsarticle}

\usepackage{amsmath}
\usepackage{amssymb}
\usepackage{amsthm}
\usepackage{graphicx}

\usepackage{algorithm}
\usepackage{algpseudocode}

\newtheorem{proposition}{Proposition}
\newtheorem{remark}{Remark}

\usepackage{xcolor}
\AtBeginDocument{\definecolor{red}{rgb}{0,0,0}}

\usepackage[colorlinks=true, linkcolor=blue, citecolor=blue,
            urlcolor=blue]{hyperref}

\journal{Journal of Computational Physics}

\makeatletter
\def\ps@pprintTitle{%
     \let\@oddhead\@empty
     \let\@evenhead\@empty
     \def\@oddfoot{\ifnopreprintline\relax\else\footnotesize\itshape
       Under consideration and revision at \ifx\@journal\@empty Elsevier
       \else\@journal\fi\hfill\@date\fi}%
     \let\@evenfoot\@oddfoot}
\makeatother

\begin{document}

\begin{frontmatter}

\title{AMPERe: an exactly divergence-free quadrature-by-expansion evaluator for
near-surface Biot–Savart and virtual-casing fields in fusion geometries}

\author[cea]{Homam Betar}
\ead{homam.betar@cea.fr}

\author[cea]{Yanick Sarazin}

\author[cea]{Virginie Grandgirard}

\author[M2P2]{David Zarzoso}

\address[cea]{CEA, IRFM, 13108 Saint-Paul-lez-Durance Cedex, France}
\address[M2P2]{Aix-Marseille Univ, CNRS, Centrale Med, M2P2 UMR 7340, Marseille, France}

\begin{abstract}
Accurate magnetic-field evaluation on and near current-carrying surfaces is a
recurring computational bottleneck in magnetic-confinement fusion. As targets
approach the source surface, Biot--Savart and virtual-casing kernels become
singular, causing standard Gauss--Legendre and upsampled trapezoidal
quadratures to lose accuracy. We introduce AMPERe, a quadrature-by-expansion
(QBX) evaluator for vector magnetostatic kernels on general three-dimensional,
non-axisymmetric tokamak and stellarator geometries. The method develops a
spherical-harmonic QBX formulation for the vector Biot--Savart and
virtual-casing (Stratton--Chu) fields and expresses the resulting magnetic
field as the analytic curl of a regular-solid-harmonic vector potential. This
construction yields a divergence-free discrete field at every expansion order,
provides analytic derivatives, and enables an a-priori separation of
truncation, quadrature, and geometric errors. On the rotating-ellipse benchmark
of Malhotra et al.~\cite{malhotra2020}, QBX maintains a relative error of
$2$--$5\times10^{-6}$ throughout the near-surface regime at
distance-independent cost, whereas upsampled trapezoidal quadrature exhibits
$O(1)$ errors and requires node counts scaling as $(h/d)^2$. AMPERe is
implemented as a header-only, performance-portable Kokkos library supporting
CPU and GPU backends, with target-specific and local-correction accelerations
validated against the baseline. The method is further demonstrated on
realistic Wendelstein~7-X configurations, including sheet-current and
virtual-casing field evaluations.
\end{abstract}


\end{frontmatter}

{\small\tableofcontents}
\bigskip

\section{Introduction}
\label{sec:intro}

Magnetically confined fusion devices --- tokamaks and stellarators ---
aim to hold a hot plasma inside a toroidal magnetic cage long enough for
fusion reactions to release net energy, and they are the central platform
of the international effort toward fusion power~\cite{freidberg2007,hazeltine2003,boozer2004}.
In both concepts the quality of the confinement is set almost entirely by the
geometry of the magnetic field, so that nearly every step of design,
optimization, and analysis rests on the ability to compute that field
accurately.

The magnetic field everywhere in such a device --- in the vacuum region
surrounding the plasma, in the tenuous scrape-off layer, and at the
first wall --- is produced by two current systems: the external field coils and the currents
flowing inside the plasma. Modern stellarator coil sets are computed with dedicated design
tools such as \textsc{regcoil}~\cite{landreman2017regcoil} (with analogous tools for
tokamaks), while plasma equilibria are obtained from equilibrium solvers such
as \textsc{vmec}~\cite{hirshman1983vmec}, \textsc{desc}~\cite{dudt2020desc}, and
\textsc{spec}~\cite{hudson2012spec}. Once these current-carrying sources are
available, the magnetic field is, in principle, given everywhere by the
Biot--Savart law.

Evaluating the plasma contribution directly is expensive because it requires
integrating the Biot--Savart kernel over the plasma volume. The
\emph{virtual-casing principle} avoids this cost by replacing the volumetric
current with equivalent surface densities on the last closed flux surface
(LCFS), whose exterior field is identical to that of the plasma~%
\cite{shafranov1972vc,Lazerson2012,hanson2015vc}. The underlying
electromagnetic foundation is the Stratton--Chu representation, which
reconstructs a source-free field from its tangential and normal traces on the
bounding surface~\cite{strattonchu1939,stratton1941}. The evaluation is thus
reduced from a three-dimensional volume integral to a two-dimensional surface
integral.

This reduction introduces the central numerical challenge addressed in this
work. The surface integrals of the virtual-casing principle---and, more
generally, of layer-potential formulations---involve the Biot--Savart kernel,
or equivalently the Laplace Green's function
$1/|\mathbf{r}-\mathbf{r}'|$ and its gradient, which become singular as the
target point $\mathbf{r}$ approaches the source surface. On the surface the
integrand is genuinely singular; nearby it is only \emph{nearly} singular,
remaining smooth but developing a sharply localized peak that fixed
quadrature rules cannot resolve. Consequently, standard
Gauss--Legendre quadrature, although highly accurate for well-separated
targets, rapidly loses accuracy in the near-surface regime.

The \emph{on-surface} case has long been addressed in magnetic-fusion
computations, where the virtual-casing field is evaluated on the source
surface itself. The classical approach, dating back to
\textsc{nestor}~\cite{merkel1987}, is \emph{singularity subtraction}: the
singular part of the kernel is integrated analytically, leaving a smooth
remainder for standard quadrature. More recently, Malhotra
et al.~\cite{malhotra2019,malhotra2020} developed high-order boundary-integral
methods for Taylor states in stellarators, including a singularity-subtraction scheme based on
a partition of unity and a polar coordinate transformation. These methods
represent the current state of the art for on-surface evaluation on general
three-dimensional toroidal surfaces.

The singularity-subtraction methods used in fusion---including the recent
high-order scheme of~\cite{malhotra2020}---are designed for \emph{on-surface}
evaluation. As noted in~\cite{malhotra2020}, accurate and efficient
off-surface quadratures for targets near arbitrary three-dimensional surfaces
remain an open problem in applied mathematics. In current fusion workflows,
off-surface field points are typically evaluated using upsampled trapezoidal
quadrature on refined surface meshes~\cite{malhotra2019}, i.e., by
locally increasing the number of quadrature nodes on the surface to reduce
quadrature errors. While effective away
from the surface, this approach becomes increasingly costly as the target
distance decreases, since the mesh resolution must track the shrinking
near-singular region. Developing an accurate and efficient near-surface
evaluation method for general three-dimensional, non-axisymmetric geometries
is the focus of this work. The proposed approach also applies directly to axisymmetric
configurations and is not restricted to closed three-dimensional surfaces.

The practical stakes of this gap are high. In magnetic-confinement studies,
the field is routinely evaluated on cylindrical $(R,\varphi,Z)$ grids and
interpolated for magnetic-field-line tracing and guiding-center orbit
calculations of energetic particles --- central to heat-load, fast-ion
confinement, and divertor analyses, as performed by codes such as
\textsc{beams3d}~\cite{beams3d}, \textsc{ascot5}~\cite{hirvijoki2014ascot},
\textsc{simple}~\cite{albert2020simple}, and \textsc{tapas}~\cite{tapas2024}.
These calculations often extend into the edge and scrape-off-layer region
just outside the LCFS, where the virtual-casing sources are located and
near-surface quadrature errors are most severe. Errors in this region
therefore directly propagate into particle trajectories and wall-load
predictions.

A second requirement is structural: the magnetic field must preserve the
solenoidal constraint $\nabla\!\cdot\!\mathbf{B}=0$ imposed by Maxwell's
equations. Violating this property introduces artificial magnetic sources
and can accumulate errors in field-line and guiding-center integrations.
Field representations assembled through interpolation or finite differences
are generally not exactly divergence-free, with the largest errors often
appearing near the singular and nearly singular regions where the field
evaluation is already most challenging. A near-surface evaluator that is
both high-order accurate and divergence-free by construction therefore
addresses both accuracy and physical consistency, which is a central aim of
this work.

These two requirements --- high-order accuracy up to the source surface and
exact preservation of the solenoidal structure --- motivate the
boundary-integral viewpoint adopted here. Many of the field-evaluation
problems above can be formulated through Poisson-type equations, for which
boundary-integral methods replace a three-dimensional volume problem by a
two-dimensional surface representation: the field at any point follows from
an explicit layer-potential evaluation. The price is a dense discretization
with naive $O(N^2)$ complexity, reduced to $O(N\log N)$ or $O(N)$ by fast
multipole methods (FMM)~\cite{greengard1997fmm}, and the fundamental
challenge of singular and nearly singular layer potentials. On-surface
targets require singular quadrature, while nearby off-surface targets produce
sharply localized kernels that defeat standard quadrature. Higher-order
rules, refinement, and oversampling can alleviate the problem but cannot
remove this kernel-driven limitation as the target approaches the surface.

The literature on singular and nearly singular layer-potential evaluation is
extensive. Common approaches combine product integration with singularity
subtraction, where the known kernel singularity is integrated analytically
against the local basis~\cite{atkinson1997,brebbia1984}. Other strategies use
kernel-specific quadratures, singularity-canceling transformations, or kernel
regularization to recover the accuracy of smooth quadrature rules
(see~\cite{kbgo2013qbx} and references therein). High-order near-boundary
evaluation has also been achieved through complex-analysis techniques
~\cite{helsing2008,bwv2015}. As emphasized in~\cite{kbgo2013qbx}, once a
potential representation is available, target derivatives follow
analytically; this is particularly relevant here since the magnetic field is
obtained from derivatives of the vector potential.

In this work, we adopt and extend the Quadrature by Expansion (QBX) method,
introduced by Kl\"ockner et al.~\cite{kbgo2013qbx} and further developed for
high-order boundary evaluation~\cite{barnett2014}. QBX overcomes the loss of
quadrature accuracy near a source surface by representing the layer potential
through local expansions about off-surface centers. The expansion coefficients
depend only on the sources, allowing the potential to be evaluated for many
targets within the convergence region. This source--target separation makes
QBX particularly compatible with fast summation methods such as the FMM.
The framework also provides a rigorous treatment of boundary limits and jump
relations~\cite{coltonkress1998}, enabling accurate evaluation on both sides of
the surface.

The principle underlying QBX is that, although a layer potential is singular
on the source surface, it is smooth at any point away from it. QBX exploits
this by placing an \emph{expansion center} at a controlled distance from the
surface, where the layer-potential integral can be evaluated accurately using
standard quadrature, and constructing a local expansion --- such as a Taylor,
Bessel--Hankel, or spherical-harmonic series --- about that center. Targets
lying within the expansion's region of convergence are then evaluated by
summing this local representation rather than by directly resolving the
near-singular kernel. Because these local expansions share the same structure
as those used in the FMM, QBX can be naturally combined with fast summation
methods~\cite{rachh2017}.

Several variants of QBX have been developed. In \emph{global} QBX, the
expansion coefficients at each center include contributions from the entire
surface. The resulting source--target separation allows the coefficients to be
reused for multiple targets and preserves compatibility with FMM acceleration
~\cite{klinteberg2014ewald}. In \emph{local} QBX, each expansion is instead
formed from a restricted neighborhood of the center, reducing the computational
cost while introducing a controllable approximation. The main challenge of
local QBX is that preserving the same compatibility with FMM requires a more
involved organization of the near and far contributions.

A related approach is the \emph{target-specific expansion} introduced by
Siegel and Tornberg~\cite{siegel2018tsqbx}, in which the expansion is constructed
directly for each target, eliminating the need for a reusable coefficient
array. A different realization was proposed by Wala and Kl\"ockner
~\cite{wala2020tse}, where the complete potential --- including both near and
far contributions --- is routed through target-specific expansions as part of
the specialized \textsc{GIGAQBX} FMM framework
~\cite{wala2018gigaqbx,wala2019gigaqbx3d}. This approach achieves high
performance but requires a substantially more complex FMM architecture, since
the expansion is no longer a simple post-processing step. In contrast, the
local target-specific construction of Siegel and Tornberg
~\cite{siegel2018tsqbx} retains a simpler organization by avoiding this global
coupling.

These ideas motivate the two acceleration strategies developed here: a local
correction that reduces the amount of near-field work while preserving the
separated source--target structure, and a target-specific contraction that
eliminates the $(P+1)^2$ coefficient storage and evaluation cost.

Against this background, we present AMPERe (Accelerated Magnetic-field
Provider via Expansion Regularization), a QBX-based evaluator for magnetostatic
Biot--Savart and virtual-casing kernels on stellarator and tokamak geometries.
The formulation is not restricted to these applications, however, and applies
to general source surfaces carrying known current or charge distributions.
AMPERe addresses a complementary problem to the boundary-integral approaches
discussed above. The methods of Malhotra et al.~\cite{malhotra2019} and Wala and Kl\"ockner~\cite{wala2018gigaqbx} are
primarily \emph{solvers}: they determine unknown surface densities by inverting
boundary integral equations. AMPERe is an \emph{evaluator}: given known surface
currents and charges --- from coil models, equilibrium calculations, or
Stratton--Chu traces of a reference field --- it computes the resulting field
at and near the source surface with high-order accuracy. These two tasks meet
at the level of near-surface evaluation, which is a fundamental component of
such boundary-integral solvers and is precisely the capability provided here
for fusion-relevant geometries and operating regimes.

The method rests on two core elements. First, it develops a spherical-harmonic
QBX formulation for the \emph{vector} magnetostatic kernel --- the Biot--Savart
field of surface currents and the virtual-casing field obtained from
Stratton--Chu traces --- with uniform evaluation across the on-, near-, and
far-surface regimes on general three-dimensional non-axisymmetric surfaces.
This addresses the off-surface near-boundary regime identified as an open
challenge in magnetic fusion by Malhotra et al.~\cite{malhotra2020}, and is
validated on the rotating-ellipse and W7-X vacuum configurations,
where $\nabla\times \mathbf{B} = 0$\cite{malhotra2019}.
Second, the field is constructed as the exact curl of
a vector potential represented by regular solid harmonics. Consequently, the
discrete QBX field is divergence-free at every expansion order and provides a
closed-form analytic gradient, eliminating the need for finite-difference
derivatives. Solenoidality is therefore enforced structurally rather than
obtained as an accuracy-dependent property; the objective is to preserve this
constraint throughout the discrete evaluation on realistic stellarator and tokamak
geometries.

These contributions extend QBX in a direction complementary to the extensive
work on its acceleration. The acceleration framework developed here --- a
target-specific near-field expansion, bounding-volume-hierarchy classification
of near and far interactions, and a header-only performance-portable
CPU/GPU implementation --- builds on established QBX acceleration techniques.
In particular, the target-specific expansion follows the approaches of
Wala--Kl\"ockner and Siegel--Tornberg~\cite{wala2020tse,siegel2018tsqbx}, while
the fast-algorithm foundations are closely related to
\textsc{gigaqbx}~\cite{wala2018gigaqbx,wala2019gigaqbx3d}. Whereas these works
focus on accelerating QBX for general layer potentials and fast summation,
AMPERe extends the framework to the vector magnetostatic setting: the
Biot--Savart and virtual-casing kernels, structure-preserving
divergence-free fields, and near-surface evaluation on realistic
three-dimensional fusion geometries. The two developments are therefore
complementary: the local-expansion interface introduced here is compatible
with future \textsc{gigaqbx}-style acceleration without modification. The
a priori error analysis used throughout --- separating truncation, quadrature,
and geometric representation errors --- builds on the foundational QBX
analyses of Epstein, Greengard, and Kl\"ockner, and of af Klinteberg and
Tornberg~\cite{egk2013,kt2017,kt2018}.

In summary, the contributions of this work are:

\begin{enumerate}
\item A spherical-harmonic QBX formulation for vector magnetostatic
Biot--Savart and virtual-casing fields on general three-dimensional surfaces.
The method resolves the near-surface off-surface regime identified as an open
problem in magnetic-fusion applications~\cite{malhotra2020}, and is validated
against the rotating-ellipse benchmark of~\cite{malhotra2019} as well as realistic
W7-X configurations.

\item A structure-preserving magnetic-field representation based on the curl of
a regular-harmonic vector potential. The resulting discrete field is exactly
divergence-free at every QBX order and provides an analytic gradient, together
with an error decomposition separating truncation, quadrature, and geometric
representation effects.
\end{enumerate}

Beyond these methodological contributions, AMPERe provides a unified QBX
evaluator supporting panel- and target-driven strategies, local correction,
and target-specific expansions. ArborX~\cite{ArborX_2020} enables near/far
classification, while the Kokkos implementation~\cite{Kokkos_2014} provides a
portable CPU/GPU execution model. The framework includes robust treatment of
complex fusion geometries through target-driven dispatch, orphan handling, and
two-sided surface classification, and provides a basis for future
boundary-integral solvers requiring accurate near-surface field evaluation.

The remainder of the paper is organized as follows.
Section~\ref{sec:problem} states the field-evaluation problem and the
Biot--Savart and virtual-casing/Stratton--Chu kernels, and explains why
Gaussian quadrature loses accuracy near the surface.
Section~\ref{sec:qbx} develops the spherical-harmonic QBX for these kernels.
Section~\ref{sec:hybrid} describes the target-driven hybrid near/far
dispatcher.
Section~\ref{sec:structure} establishes the structure-preservation results ---
the exact divergence-free theorem and the a-priori error decomposition.
Section~\ref{sec:acceleration} presents the local-correction and
target-specific accelerations and the Kokkos-based CPU/GPU portability
infrastructure.
Section~\ref{sec:results} reports the numerical results, including the
head-to-head comparison with the boundary-integral solver of
Malhotra et al.\ on their own geometries, a demonstration on realistic
W7-X stellarator geometries, and verification of the divergence-free
structure and of CPU/GPU portability.
Section~\ref{sec:software} covers software and reproducibility, and
Section~\ref{sec:conclusion} concludes.

\section{Biot--Savart magnetic-field evaluation in fusion
devices}
\label{sec:problem}

In a fusion device --- a stellarator or a tokamak --- the equilibrium field
that confines the plasma organizes it into a set of nested toroidal magnetic
surfaces: closed, twisted surfaces to which the field lines are everywhere
tangent\textcolor{red}{; however, in three-dimensional,
non-axisymmetric magnetic fields, the existence of globally nested flux
surfaces is not guaranteed. Nevertheless, when the magnetic field is sufficiently
close to an integrable configuration, KAM theory~\cite{ArnoldKozlovNeishtadt2006}
predicts the survival of a large set of invariant tori under sufficiently small
perturbations, provided appropriate non-resonance conditions are satisfied.
This theoretical picture has been experimentally confirmed in W7-X with extremely
high accuracy~\cite{pedersen2016}, although resonant surfaces can break into magnetic
islands and, with sufficiently strong perturbations, stochastic regions can develop.}
This field is produced by the currents flowing in the external coils
together with the currents flowing inside the plasma; in a stellarator the net
toroidal plasma current is designed to be small, so that the confining
rotational transform is generated mainly by the coils. Given these sources, the
field at a point $\mathbf{r}$ follows from the Biot--Savart law,
\begin{equation}\label{eq:BS}
\mathbf{B}(\mathbf{r}) = \frac{\mu_0}{4\pi}
\int_V \frac{\mathbf{J}(\mathbf{r}') \times (\mathbf{r} - \mathbf{r}')}
{|\mathbf{r} - \mathbf{r}'|^{3}}\,dV',
\end{equation}
where $\mathbf{J}$ [A\,m$^{-2}$] is the current density, $V$ the region that
carries it, and $\mu_0$ the vacuum permeability. The kernel is built from the
Green's function $G(\mathbf{r},\mathbf{r}')=1/|\mathbf{r}-\mathbf{r}'|$, which
is singular as the observation (target) point $\mathbf{r}$ approaches a source
point $\mathbf{r}'$. Integrating \eqref{eq:BS} numerically is therefore
straightforward when the target is well separated from the sources, but becomes
severely ill-behaved when the target lies on or near a current-carrying layer,
where every smooth quadrature rule loses accuracy. The same difficulty recurs
wherever the field is required close to a current-carrying boundary --- near an
antenna, or near a conducting wall on which eddy currents flow. We make this
breakdown quantitative in Section~\ref{sec:breakdown}, once the test
configurations and error metrics have been introduced.

The same near-singular behavior arises in a distinct and central task:
computing the field produced by the plasma itself in the region that
separates the plasma from the coils. \textcolor{red}{This region is not a
vacuum in the strict physical sense---the scrape-off layer has a low but
finite density and can sustain small, fluctuating currents. Nevertheless,
these currents are neglected in the present model, so the region is treated
as current-free. We therefore use the conventional term ``vacuum region'' throughout.}
In electromagnetic theory the
Stratton--Chu representation expresses the field inside a source-free region in
terms of the tangential and normal traces of the field on its bounding surface.
For the current-free and divergence-free vacuum region outside the plasma this
reduces to the virtual-casing principle: the field of the entire plasma volume
is reproduced, at any target outside the LCFS, by equivalent surface densities
supported on the LCFS alone. Writing $\hat{\mathbf{n}}$ for the outward unit
normal on the surface $S$ (the LCFS bounding the plasma volume $V$), the
equivalent surface current $\mathbf{K}$ and normal trace $\sigma$ are
\begin{equation}\label{eq:Ksigma}
\mathbf{K} = \frac{1}{\mu_0}\,\hat{\mathbf{n}}\times\mathbf{B}
\quad[\mathrm{A\,m^{-1}}],
\qquad
\sigma = \hat{\mathbf{n}}\cdot\mathbf{B}\quad[\mathrm{T}],
\end{equation}
and the field is reconstructed as the sum of a surface-current (Biot--Savart)
term and a surface-charge term,
\begin{equation}\label{eq:VCP}
\mathbf{B}(\mathbf{r}) =
\frac{\mu_0}{4\pi}\int_S
\frac{\mathbf{K}(\mathbf{r}')\times(\mathbf{r}-\mathbf{r}')}
{|\mathbf{r}-\mathbf{r}'|^{3}}\,dS'
+\frac{1}{4\pi}\int_S
\sigma(\mathbf{r}')\,\frac{\mathbf{r}-\mathbf{r}'}
{|\mathbf{r}-\mathbf{r}'|^{3}}\,dS'.
\end{equation}
The virtual-casing principle thus reduces the volume integral \eqref{eq:BS} to
a surface integral over the boundary of the plasma. It carries, however, exactly
the same $1/|\mathbf{r}-\mathbf{r}'|$ singularity: evaluating \eqref{eq:VCP} at
targets on or near $S$ is as ill-behaved as the coil problem above, and standard
quadrature fails there in the same way; we quantify this on the virtual-casing
benchmark in the numerical results (Section~\ref{sec:results}). \textcolor{red}
{Equations~\eqref{eq:Ksigma}--\eqref{eq:VCP} apply to an arbitrary
closed source surface $S$, which need not be a flux surface. For the LCFS and the
total equilibrium field, $\sigma=\hat{\mathbf n}\cdot\mathbf B=0$ by construction.
This is verified numerically for W7-X in Section~\ref{sec:res_w7x}, where $\sigma$
is at the $10^{-19}$ level; we nevertheless retain the $\sigma$ term as a consistency
check. The term is nonzero for non-flux source surfaces, such as the benchmark surfaces,
winding surfaces, coils, and conducting walls, and when the plasma-only field is used,
as in coil design.}

\textcolor{red}{It is convenient to introduce here the vector potential of a
surface current, on which the expansions of Section~\ref{sec:qbx} are built,
\begin{equation}\label{eq:Asingle}
\mathbf{A}(\mathbf{r}) = \frac{\mu_0}{4\pi}\int_S
\frac{\mathbf{K}(\mathbf{r}')}{|\mathbf{r}-\mathbf{r}'|}\,dS',
\qquad
\mathbf{B} = \nabla\times\mathbf{A}.
\end{equation}
Equation~\eqref{eq:Asingle} is a \emph{single-layer potential}: the source
density $\mathbf{K}$ is distributed on the surface $S$ and is integrated
against the Green's function $G(\mathbf{r},\mathbf{r}') =
1/|\mathbf{r}-\mathbf{r}'|$. The term ``single-layer'' refers to the fact that
the density itself multiplies the Green's function. In contrast, a
\emph{double-layer potential} involves the normal derivative of the Green's
function, $\partial G/\partial n'$, and can be interpreted as the potential
generated by a distribution of dipoles oriented normal to the surface. In the
same notation, the second term in Eq.~\eqref{eq:VCP} is $-\nabla\Phi$, where
$\Phi$ is the scalar single-layer potential generated by the surface density
$\sigma$.}

A representative and demanding instance of both difficulties is the tracing of
energetic-particle trajectories in a fusion device. Most tracer codes integrate the
equations of motion on a cylindrical $(R,\varphi,Z)$ grid; while an equilibrium
code supplies the field inside the plasma, the field at grid points outside it
must be assembled from the Biot--Savart contribution of the coils and, through
the virtual-casing principle, of the plasma. Singular and nearly singular
integrals then arise whenever a grid point falls on or close to the LCFS or a
coil --- a well-known obstacle in the community that only worsens as the nested
flux surfaces become more strongly shaped and twisted. The QBX-based evaluator
developed in this work is designed to bridge this gap.

\subsection{Model configurations and error metrics}
\label{sec:setup}

The AMPERe evaluator is assessed throughout the paper on a set of benchmark
configurations, each pairing a source surface with a reference field known
analytically or obtained to numerical convergence (the current-loop and W7-X
cases). For these benchmarks, the reference error is negligible compared with
the evaluation error, allowing the reported discrepancies to isolate the
accuracy of the near-surface field evaluation. The configurations are introduced once here and referenced by name throughout
the paper; their main parameters are summarized in Table~\ref{tab:configs},
while their reference solutions, when available, are detailed in
~\ref{app:configs}.

\paragraph{Surface discretization}
Every source surface $S$ is partitioned into $N$ curved quadrilateral panels of
characteristic size $h$. Each panel is described by an order-$g$ Lagrange map
$\mathbf{x}_p:[-1,1]^2\to S$ from a reference square --- $g=1$ (bilinear), $g=2$
(quadratic), or $g=3$ (cubic) --- so that the discrete surface reproduces the
true one to $O(h^{g+1})$; we call $g$ the \emph{geometry-representation order}.
The surface integral of a kernel $f$ is then evaluated panel by panel with an
$n_{\mathrm{gl}}\times n_{\mathrm{gl}}$ tensor-product Gauss--Legendre rule,
\begin{equation}\label{eq:disc}
\int_S f\,dS \;\approx\;
\sum_{p=1}^{N}\ \sum_{i,j=1}^{n_{\mathrm{gl}}}
w_i\,w_j\, f\!\big(\mathbf{x}_p(\xi_i,\xi_j)\big)\,J_p(\xi_i,\xi_j),
\end{equation}
where $\{\xi_i,w_i\}_{i=1}^{n_{\mathrm{gl}}}$ are the Gauss--Legendre nodes and
weights on $[-1,1]$, $J_p$ is the surface Jacobian of panel $p$, and
$n_{\mathrm{gl}}$ is the \emph{quadrature order}. These two parameters act on
independent error sources and recur throughout the paper: $n_{\mathrm{gl}}$ sets
the accuracy of the panel quadrature --- spectral in $n_{\mathrm{gl}}$ for an
integrand that is smooth at the evaluation distance --- while $g$ fixes the
geometric fidelity of the surface itself. Their separate roles are made precise
in the error analysis of Section~\ref{sec:structure}.

\paragraph{Planar sheet}
A flat square current-sheet patch of side $L=5$ carrying a uniform surface
current, discretized into $16\times16$ bilinear panels ($g=1$, $n_{\mathrm{gl}}=10$),
evaluated at QBX order $P=12$, with the analytic sheet-current field as
reference. This canonical flat test isolates the QBX behavior from surface
curvature; it exhibits the QBX-versus-Gauss--Legendre crossover and the clear
separation between the QBX truncation error and the geometric
surface-representation error.

\paragraph{Toroidal coil}
A toroidal surface current on a torus of major radius $R_0=1$ and minor radius
$a=0.3$ carrying unit total current, whose field is available in closed form.
Targets approach the surface from the exterior, toward the outer equator
$R_0+a$. It provides the near-surface breakdown demonstration of
Section~\ref{sec:breakdown}; its role in the results section is subsumed by
the planar-sheet and rotating-ellipse studies.

\paragraph{Rotating-ellipse virtual casing (dipole reference)}
A rotating-ellipse surface of major radius $R_0=3$ with an elliptical
cross-section of semi-axes $a=0.6$ and $b=0.3$, enclosing a magnetic dipole
placed at \textcolor{red}{$(R,Z,\varphi)=(3,0,0)$} with unit moment along $\hat{\mathbf z}$. The dipole field,
exact everywhere off the dipole, is cast to the Stratton--Chu surface traces
\eqref{eq:Ksigma}, and the reconstructed field \eqref{eq:VCP} is compared
against the analytic dipole.

\paragraph{Malhotra rotating-ellipse benchmark}
The rotating ellipse of Malhotra et al.~\cite{malhotra2019} (their \textcolor{red}{Eq.~(55)}:
$R_0=2$, semi-axes $a=0.7$ and $b=1.0$, cross-section rotation rate $3/2$), with
the reference field produced by a circular current loop of radius $5$ placed
well outside the domain. Because the loop is a filament, its Biot--Savart field
is a vacuum field (\textcolor{red}{where $\nabla \times \mathbf{B} = 0$}), accurate to sixteen digits at and inside the
surface; its traces \eqref{eq:Ksigma} drive the virtual-casing reconstruction.
This configuration provides a comparison against the
boundary-integral solver of~\cite{malhotra2019} in Section~\ref{sec:results},
on that solver own geometry and reference field.

\paragraph{W7-X last closed flux surface}
The last closed flux surface of the Wendelstein~7-X stellarator --- a fully
three-dimensional, non-axisymmetric geometry. The plasma field is supplied through the equilibrium pipeline
and cast to the Stratton--Chu traces \eqref{eq:Ksigma}, exercising the evaluator
on a realistic device rather than a model surface.

\paragraph{Error metrics}
For a set of target points $\{\mathbf r_i\}_{i=1}^{M}$, let $\mathbf
B(\mathbf r_i)$ be the evaluated field and $\mathbf B_{\mathrm{ref}}(\mathbf
r_i)$ the reference. We report the relative $\ell_2$ error and, for the
comparison with~\cite{malhotra2019}, the relative $L^\infty$ error,
\begin{equation}\label{eq:metrics}
E_{\ell_2} =
\frac{\Big(\sum_i \lVert \mathbf B(\mathbf r_i)-\mathbf B_{\mathrm{ref}}(\mathbf r_i)\rVert^2\Big)^{1/2}}
     {\Big(\sum_i \lVert \mathbf B_{\mathrm{ref}}(\mathbf r_i)\rVert^2\Big)^{1/2}},
\qquad
E_{\infty} =
\frac{\max_i \lVert \mathbf B(\mathbf r_i)-\mathbf B_{\mathrm{ref}}(\mathbf r_i)\rVert}
     {\max_i \lVert \mathbf B_{\mathrm{ref}}(\mathbf r_i)\rVert}.
\end{equation}
Near-surface behavior is reported against the dimensionless target--surface
distance $d/h$, where $d=\min_{\mathbf r'\in S}\lVert\mathbf r-\mathbf r'\rVert$
is the distance from the target to the source surface and $h$ is the local panel
size; we write $d/h_{\mathrm{panel}}$ where the specific panel measure matters.
The limit $d/h\to 0$ is the near-surface regime in which standard quadrature
fails.

\begin{table}[t]
\centering
\small
\color{red}
\caption{Test configurations used in this work.}
\label{tab:configs}
\setlength{\tabcolsep}{5pt}
\begin{tabular}{p{2.4cm}p{3.5cm}p{4.0cm}}
\hline
Case & Reference field & Purpose \\
\hline
Planar sheet &
  Analytic &
  QBX--GL crossover and error sources \\
Toroidal coil &
  Analytic &
  Near-surface coil field \\
Rotating ellipse &
  Analytic dipole &
  Virtual-casing accuracy and parameter convergence \\
Malhotra ellipse &
  High-accuracy current-loop solution~\cite{malhotra2019} &
  Comparison with~\cite{malhotra2019} \\
W7-X LCFS &
  Stratton--Chu reconstruction &
  Realistic 3D stellarator geometry \\
\hline
\end{tabular}
\end{table}

\subsection{Standard quadrature and its near-surface breakdown}
\label{sec:breakdown}

All of the configurations above reduce to integrating a kernel of the form
$1/|\mathbf{r}-\mathbf{r}'|$, and its gradient, over the source surface. Two
standard quadratures serve as our baselines. Composite Gauss--Legendre
quadrature --- the panel rule~\eqref{eq:disc} --- is highly accurate for smooth
integrands and is the natural choice on a panelized surface; the midpoint
(trapezoidal) rule, spectrally accurate for smooth periodic integrands,
underlies the upsampled off-surface scheme of~\cite{malhotra2019}. We do not
derive the underlying one-dimensional rules --- both are
standard~\cite{atkinson1997} --- because what matters here is not their
construction but their common failure mode near the surface.

As the target approaches the surface, the integrand ceases to be smooth on the
scale of the quadrature: the factor $1/|\mathbf{r}-\mathbf{r}'|$ develops a peak
of width comparable to the target--surface distance $d$, centred at the nearest
source point. A fixed rule with nodes spaced $\sim h$ samples this peak with only
$O(d/h)$ points, and once $d\lesssim h$ it no longer resolves the peak at all.
\textcolor{red}{The quadrature error then grows monotonically as $d/h\to0$ and
saturates at a finite level set by the source strength, whatever the order of
the rule: it remains bounded, but all significant digits are lost.}
The difficulty is analytic, inherited from the kernel, and
cannot be removed by refining the mesh or raising the quadrature order. This is the
breakdown that quadrature by expansion is designed to overcome.
Figure~\ref{fig:coil_breakdown} shows it for the toroidal coil of
Section~\ref{sec:setup}: both the midpoint rule and Gauss--Legendre are accurate
far from the winding but lose all significant digits as $d/h\to0$.
\textcolor{red}{We evaluate the toroidal component $B_\varphi$ at an outboard
midplane target, located at $R=(R_0+a)+d$, where $d>0$ is the distance from the
coil surface and $h$ is the panel size. The field is computed from the
surface-current form of Eq.~\eqref{eq:BS} using the two baseline quadratures:
the midpoint rule (\emph{direct}) and the composite panel Gauss--Legendre rule
of Eq.~\eqref{eq:disc}. The reference solution is the closed-form field of the
ideal toroidal solenoid, Eq.~\eqref{eq:torus_field}, which vanishes identically
outside the coil. Thus, for all targets in Fig.~\ref{fig:coil_breakdown}, the
error is simply $|B_\varphi|$. In the normalized units
$\mu_0/(4\pi)=I=1$ used in Appendix~\ref{app:configs}, both quadratures recover
the exterior field accurately at $d/h=10$, with the midpoint rule reaching
machine precision and Gauss--Legendre reaching a $\sim10^{-10}$ floor. In
contrast, at $d/h=10^{-4}$ both yield $|B_\varphi|\approx0.75$, compared with
the interior field $2/R\approx1.54$ at the same radius. Thus, sufficiently
close to the coil, the quadrature error can become comparable to the physical
field itself.}

\begin{figure}[t]
  \centering
  \IfFileExists{figures/fig_coil_near_surface.png}
    {\includegraphics[width=0.6\linewidth]{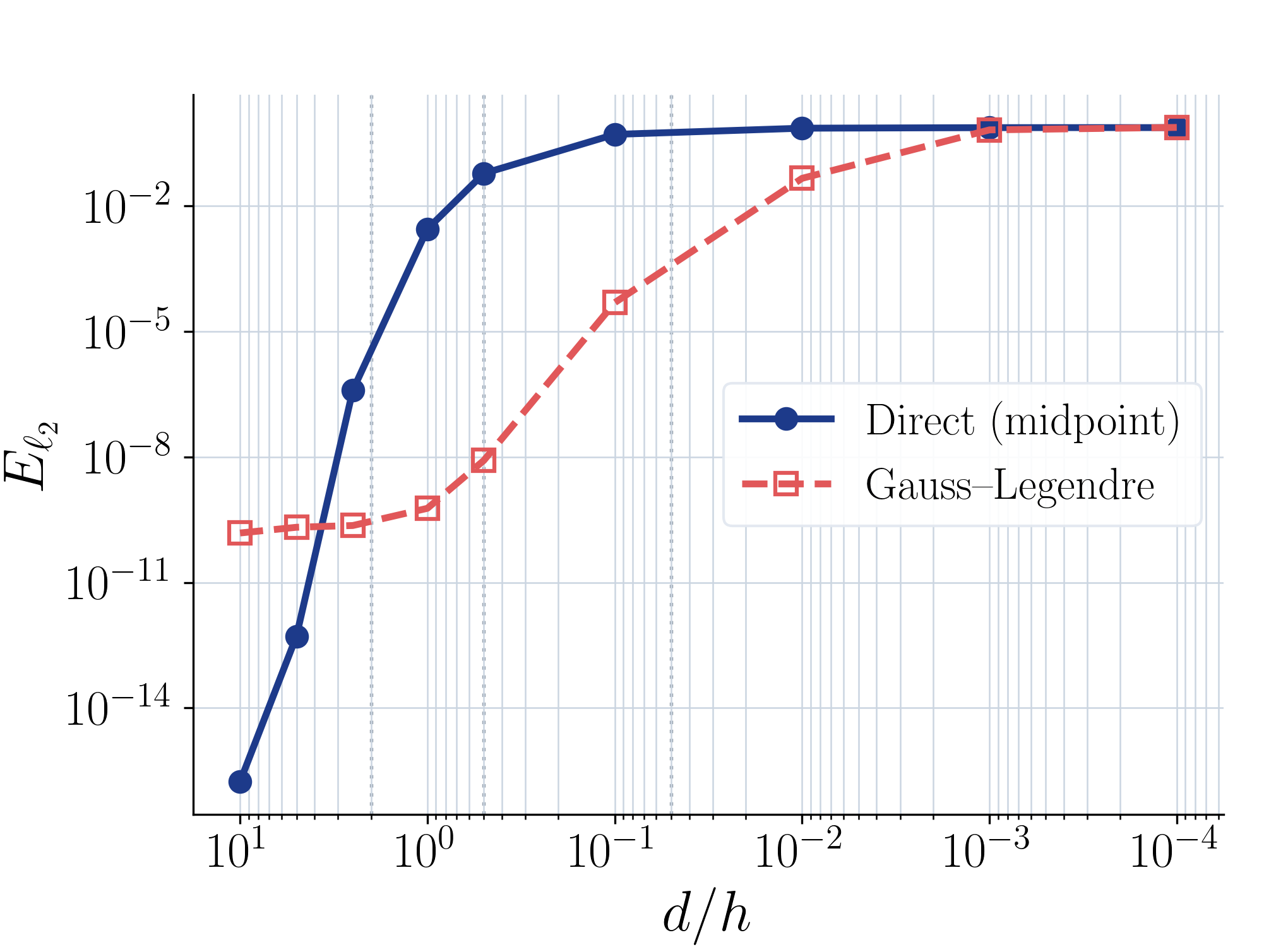}}
    {\fbox{\parbox[c][0.34\linewidth][c]{0.6\linewidth}{\centering\itshape
      [placeholder: export \texttt{plot\_near\_surface(direct, gl)} to
      \texttt{figures/fig\_coil\_near\_surface.png}]}}}
      \caption{\textcolor{red}{Near-surface breakdown of standard quadrature for the toroidal-coil
field (Section~\ref{sec:setup}). The relative $\ell_2$ field error
$E_{\ell_2}$ is shown versus the normalized target--surface distance $d/h$
for exterior targets approaching the coil surface; the distance decreases from
left to right. Both midpoint quadrature (\emph{direct}) and Gauss--Legendre
(GL) are accurate for $d/h\gtrsim1$ but lose accuracy as the target approaches
the surface, illustrating the near-singular breakdown of fixed-order
quadrature.}}
  \label{fig:coil_breakdown}
\end{figure}

\section{Spherical-harmonic quadrature by expansion}
\label{sec:qbx}

Quadrature by expansion (QBX), introduced by Kl\"ockner et
al.~\cite{kbgo2013qbx}, removes the near-surface breakdown of
Section~\ref{sec:breakdown} by replacing the near-singular quadrature with a
local series. The idea, illustrated in Figure~\ref{fig:qbx_schematic}, has three
steps: (i) place an \emph{expansion centre} $\mathbf{c}$ a small distance off
the source surface, where the field is smooth; (ii) build a local expansion of
the field about $\mathbf{c}$ from the source data, using a quadrature that is now
well separated from the surface; and (iii) evaluate the field at a target
$\mathbf{r}$ by summing that expansion. Each centre is valid within its
\emph{ball of convergence} $\mathcal{B}(\mathbf{c},r)$, which may contain the
near- and on-surface targets that direct quadrature cannot reach.

\begin{figure}[t]
  \centering
  \IfFileExists{figures/fig_qbx_schematic.png}
    {\includegraphics[width=0.6\linewidth]{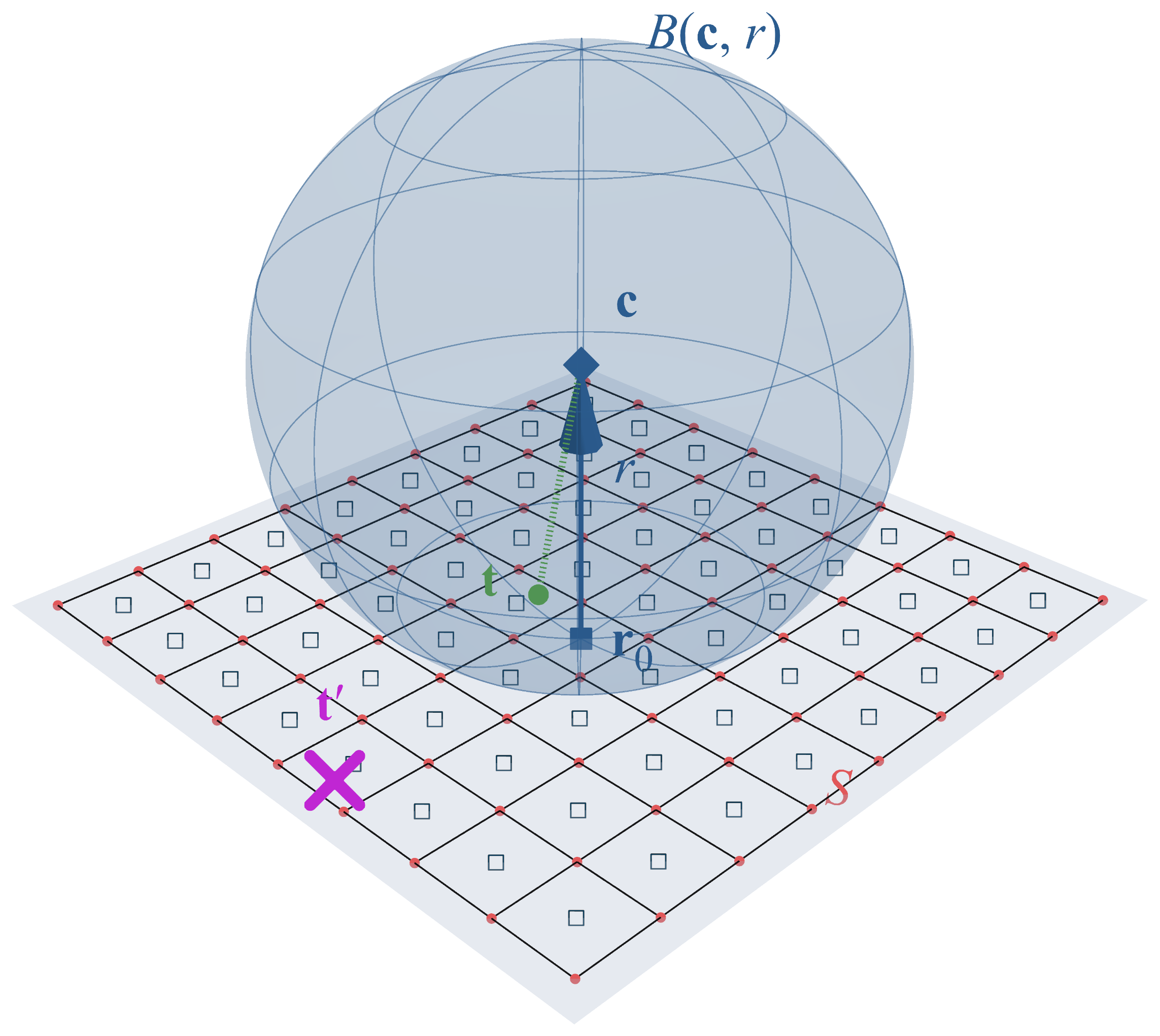}}
    {\fbox{\parbox[c][0.36\linewidth][c]{0.6\linewidth}{\centering\itshape
      [placeholder: QBX schematic --- source surface $S$, expansion centre $c$
      at offset $r$, ball of convergence, near target $t$ with $|t-c|<r$]}}}
  \caption{\textcolor{red}{Schematic illustration of QBX. An expansion center $\mathbf{c}$ (blue)
is positioned a distance $r$ from the source surface $S$ (red). The convergence
ball $\mathcal{B}(\mathbf{c},r)$ is tangent to $S$ at the source-panel centroid
$\mathbf{r}_0$. A target $\mathbf{t}$ inside the ball lies within the convergence
region, whereas $\mathbf{t}'$ outside it does not. Red points denote source
quadrature nodes, open squares panel centroids, and the filled square the
centroid $\mathbf{r}_0$ from which the expansion center is offset along the
surface normal.}}
  \label{fig:qbx_schematic}
\end{figure}

\paragraph{Centre placement.}
A centre is built from the source panel nearest the target. Let that panel have
centroid $\mathbf{r}_0$, tangent vectors $\mathbf{t}_\zeta,\mathbf{t}_\eta$, and
outward unit normal $\mathbf{n}$. The two centres and the radius of the
convergence ball are
\begin{equation}\label{eq:centre}
\mathbf{c}^{\pm} = \mathbf{r}_0 \pm r\,\mathbf{n},
\qquad
r = \sigma_1\!\big([\mathbf{t}_\zeta\ \ \mathbf{t}_\eta]\big) = \tfrac{1}{2}h,
\end{equation}
where $\sigma_1$ is the largest singular value of the $3\times2$ panel Jacobian
$[\mathbf{t}_\zeta\ \ \mathbf{t}_\eta]$, a local measure of the panel size $h$.
The two signs place a centre on each side of the surface, so a target can be
served from whichever side it lies on --- the basis of the two-sided treatment
of Section~\ref{sec:hybrid}. \textcolor{red}{The centre is thus offset by one
panel-size measure, $r=\sigma_1$, or equivalently $r=h/2$. This choice controls
the accuracy with which the coefficients~\eqref{eq:Lnm_num} are integrated. On
the source panel the integrand $S_n^{-m}(\mathbf{r}'-\mathbf{c})$ is smooth, but
its analytic continuation is singular at the centre $\mathbf{c}$, at distance
$r$ from the panel; the resulting Gauss--Legendre error therefore decays
geometrically in the ratio of the panel half-length $h/2$ to that distance. For
the panel that carries the centre, an $n_{\rm gl}$-point rule gives a
$\big(h/4r\big)^{2n_{\rm gl}+1}$ decay, so $r=h/2$ fixes this ratio at
$\tfrac12$ and yields a clean $2^{-(2n_{\rm gl}+1)}$
rate~\cite{kbgo2013qbx,wala2018gigaqbx,wala2019gigaqbx3d}.}
For a target $\mathbf{r}$ we define the convergence
ratio~\cite{kbgo2013qbx}
\begin{equation}\label{eq:qratio}
q = \frac{|\mathbf{r}-\mathbf{c}|}{r},
\end{equation}
\textcolor{red}{so that $q<1$ denotes the interior of the ball, where the series
converges. A target is nevertheless \emph{accepted} if
$q\le 1+\epsilon_{\rm tgt}$, with $\epsilon_{\rm tgt}$ a small tolerance. This
extends the set of targets assigned to a centre, but not the convergence region
of the expansion; the resulting accuracy loss is quantified below.}

\paragraph{From scalar potentials to the magnetic field.}
QBX and its fast variants have so far been developed for \emph{scalar} layer
potentials~\cite{kbgo2013qbx,siegel2018tsqbx,wala2018gigaqbx,wala2019gigaqbx3d}.
The quantities we need are instead the vector potential $\mathbf{A}$ and the
magnetic field $\mathbf{B}$. AMPERe provides two routes, both built on the same
expansion. \emph{Path~1} treats $\mathbf{A}$ as three independent scalar
potentials --- to which the scalar theory applies component by component --- and
recovers the field from $\mathbf{B}=\nabla\times\mathbf{A}$. \emph{Path~2}
expands the magnetic field directly. As shown below, the two coincide
analytically, but Path~2 evaluates the field through the closed-form gradient of
the expansion, which makes it an \emph{exact curl} and therefore exactly
divergence-free at every order (Section~\ref{sec:structure}); a component-wise
\emph{numerical} differentiation of Path~1, by contrast, destroys that
structure. To our knowledge, QBX has not previously been extended to the vector
magnetostatic (Biot--Savart / virtual-casing) kernel with a formulation that
returns the magnetic field directly in a structure-preserving form.

\paragraph{The spherical-harmonic expansion.}
Both paths rest on expanding the Green's function about the centre. Among the
admissible local bases --- Taylor, Fourier--Hankel, spherical harmonics ---
spherical harmonics are preferred for their accuracy and numerical
stability~\cite{wala2019gigaqbx3d}. With $\mathbf{r}$ the target and
$\mathbf{r}'$ a source point, the Laplace addition theorem reads
\begin{equation}\label{eq:addition}
\frac{1}{|\mathbf{r}-\mathbf{r}'|}
= 4\pi\sum_{n=0}^{\infty}\sum_{m=-n}^{n}\frac{1}{2n+1}\,
R_n^m(\mathbf{r}-\mathbf{c})\,S_n^{-m}(\mathbf{r}'-\mathbf{c}),
\qquad |\mathbf{r}-\mathbf{c}|<|\mathbf{r}'-\mathbf{c}| ,
\end{equation}
where the \emph{regular} and \emph{irregular} solid harmonics are
\begin{equation}\label{eq:solid}
\begin{gathered}
R_n^m(\mathbf{r}) = |\mathbf{r}|^{\,n}\,Y_n^m(\hat{\mathbf{r}}),
\qquad
S_n^m(\mathbf{r}) = \frac{Y_n^m(\hat{\mathbf{r}})}{|\mathbf{r}|^{\,n+1}},
\\[4pt]
Y_n^m(\theta,\phi)=\sqrt{\tfrac{2n+1}{4\pi}\,\tfrac{(n-|m|)!}{(n+|m|)!}}\;
\textcolor{red}{\mathcal{P}}_n^{|m|}(\cos\theta)\,e^{im\phi},
\end{gathered}
\end{equation}
with $n\ge0$ and $-n\le m\le n$. \textcolor{red}{Here, $\mathcal{P}$ is Legendre polynomial}. The regular part $R_n^m$ is smooth at the origin
and depends only on the target; the irregular part $S_n^m$ depends only on the
source. This target/source separation is the heart of QBX. Because $\mathbf{c}$
is pushed off the surface, every source sits at distance $\gtrsim r$ from the
centre, so $S_n^{-m}(\mathbf{r}'-\mathbf{c})$ is never evaluated at its
singularity. \textcolor{red}{The condition in Eq.~\eqref{eq:addition} requires the target to
be closer to $\mathbf{c}$ than every source point,
$|\mathbf{r}-\mathbf{c}|<d_{\min}(\mathbf{c})$, where
$d_{\min}(\mathbf{c})=\min_{\mathbf{r}'\in S}|\mathbf{r}'-\mathbf{c}|$.
Since the centre is offset by $r$ from its associated panel,
$d_{\min}\le r$, with equality when no other part of the surface lies closer
to $\mathbf{c}$. This is the case for the geometries considered here, so
$d_{\min}=r$ and the ball $\mathcal{B}(\mathbf{c},r)$ coincides with the
convergence region. For strongly curved or geometrically complex surfaces,
however, another source may lie closer to the centre, giving
$d_{\min}<r$; in that case, the convergence region is smaller than
$\mathcal{B}(\mathbf{c},r)$ and the criterion $q<1$ no longer coincides with
the nominal ball.} Truncating the series at
$n\le P$, with $P$ the \emph{expansion order}, is the QBX approximation; the
discarded tail decays as $q^{P+1}$.

\textcolor{red}{This clarifies the role of the tolerance $\epsilon_{\rm tgt}$ in
Eq.~\eqref{eq:qratio}. QBX evaluates only the finite truncation with
$n\le P$, so accepting targets with $q\le1+\epsilon_{\rm tgt}$ does not
invalidate the finite sum. It increases the truncation contribution by at most
$(1+\epsilon_{\rm tgt})^{P+1}$, which is a factor of $1.55$ for
$\epsilon_{\rm tgt}=0.05$ and $P=8$. However, in the shell
$1<q\le1+\epsilon_{\rm tgt}$ the truncation error grows with $P$, so increasing
the expansion order does not improve accuracy. The tolerance is therefore kept
small; it also prevents on-surface targets, for which $q=1$ ideally, from being
rejected due to round-off.}

\paragraph{The two field expansions}
Substituting~\eqref{eq:addition} into the single-layer vector
potential~\textcolor{red}{\eqref{eq:Asingle}} and into the corresponding
Biot--Savart integral collapses all source dependence into one set of
coefficients per centre, giving
\begin{equation}\label{eq:paths}
\begin{aligned}
\mathbf{A}(\mathbf{r}) &= \mu_0\sum_{n=0}^{P}\sum_{m=-n}^{n}
  R_n^m(\mathbf{r}-\mathbf{c})\,\mathbf{L}_n^m
  && \text{(Path 1)},\\[2pt]
\mathbf{B}(\mathbf{r}) &= \mu_0\sum_{n=0}^{P}\sum_{m=-n}^{n}
  \nabla R_n^m(\mathbf{r}-\mathbf{c})\times\mathbf{L}_n^m
  && \text{(Path 2)},
\end{aligned}
\end{equation}
with the vector expansion coefficients
\begin{equation}\label{eq:Lnm}
\mathbf{L}_n^m = \frac{1}{2n+1}\int_S \mathbf{K}(\mathbf{r}')\,
S_n^{-m}(\mathbf{r}'-\mathbf{c})\,\mathrm{d}S(\mathbf{r}') .
\end{equation}
Path~1 is the scalar single-layer expansion~\cite{wala2018gigaqbx} applied to
the three components of $\mathbf{A}$. Path~2 is the direct field expansion
introduced here: since $\mathbf{L}_n^m$ is a constant vector,
$\nabla\times(R_n^m\,\mathbf{L}_n^m)=\nabla R_n^m\times\mathbf{L}_n^m$, so
Path~2 is precisely the analytic curl of Path~1 and the two are equivalent term
by term. Evaluated this way, $\mathbf{B}$ is an exact curl of $\mathbf{A}$ ---
hence exactly divergence-free for every $P$ (Section~\ref{sec:structure}) --- and
the gradient $\nabla R_n^m$ is available in closed form from the solid-harmonic
recurrences, so the field and its Jacobian are analytic rather than
finite-differenced.

\paragraph{Forming the coefficients.}
All source and geometry information for a centre is compressed into the
$(P+1)^2$ coefficients $\mathbf{L}_n^m$, which are then reused for every target
in $\mathcal{B}(\mathbf{c},r)$ --- built once, evaluated many times. Because
$S_n^{-m}$ is non-singular at every source \textcolor{red}{point}, $\mathbf{L}_n^m$ is formed by the
ordinary panel Gauss--Legendre rule~\eqref{eq:disc},
\begin{equation}\label{eq:Lnm_num}
\mathbf{L}_n^m=\frac{1}{2n+1}\sum_{p}\sum_{i,j=1}^{n_{\rm gl}}
w_i\,w_j\,|J_p(\xi_i,\eta_j)|\,\mathbf{K}_p(\xi_i,\eta_j)\,
S_n^{-m}\!\big(\mathbf{r}'_p(\xi_i,\eta_j)-\mathbf{c}\big).
\end{equation}
Forming the coefficients of every centre by summing over all $N_p$ panels costs
$O\!\big(N_p^2\,n_{\rm gl}^2\,(P+1)^2\big)$. This global form phase is expensive
--- Section~\ref{sec:acceleration} reduces its cost --- but the source/target
decoupling it exploits is exactly what makes QBX compatible with fast multipole
acceleration.

\paragraph{Error sources.}
The QBX field carries three principal error sources: truncation of the spherical-harmonic series (set by the
expansion order $P$), the Gauss--Legendre error in forming the
coefficients~\eqref{eq:Lnm_num} (set by $n_{\rm gl}$), and the geometric error of
the order-$g$ surface representation (set by $g$); a round-off floor bounds them
from below. Because the three sources are independent, the errors are separately
controllable. We assemble them into an a-priori decomposition, alongside the
exact divergence-free property, in Section~\ref{sec:structure}, and exhibit their
separation numerically in Section~\ref{sec:results}.

The construction above assumes that each target is served by a well-placed
nearby centre. The target-driven dispatch, orphan handling, and two-sided
classification that make this robust are the subject of the next section; the
quadrature refinements and fast-algorithm machinery that make it efficient
follow in Section~\ref{sec:acceleration}.

\section{The hybrid near/far dispatcher}
\label{sec:hybrid}

AMPERe is, in its present form, primarily a magnetic-field \emph{evaluator}
for fusion applications, although the underlying machinery is more general. In
the workflows that motivate it --- such as field tabulation on cylindrical
grids or charged-particle orbit integration near the LCFS --- most evaluation
points are well separated from the source surface, where direct Gauss--Legendre
quadrature remains accurate, while only a small subset enters the near-surface
regime requiring QBX. The global (panel-driven) QBX formulation of
Section~\ref{sec:qbx}, however, constructs an expansion centre for every panel
regardless of the target distribution. It therefore incurs the
$O\!\big(N_p^2\,n_{\rm gl}^2\,(P+1)^2\big)$ form-phase cost even when only a
small fraction of targets require expansions, leading to substantial
unnecessary work.

\paragraph{Target-driven dispatch.}
The hybrid dispatcher inverts this logic. It does not precompute any centre;
instead, each target is first \emph{classified}, and a centre is formed
\emph{on the fly} only for the targets that need one. For a target $\mathbf{r}$
we locate the nearest source panel, evaluate the local panel-size measure $h$
there (Section~\ref{sec:setup}), and form the ratio $d/h$ with
$d=|\mathbf{r}-\mathbf{r}_0|$ the distance to that panel's centroid. Writing
$\tau$ for a near/far threshold (default $\tau=1$, the empirical cross-over on
the toroidal-coil benchmark of Section~\ref{sec:results}),
\begin{equation}\label{eq:dispatch}
d/h > \tau:\ \text{Direct Gauss--Legendre},
\qquad\qquad
d/h \le \tau:\ \text{target-driven QBX},
\end{equation}
where a target-driven QBX centre is placed one-shot beneath the target and its
local coefficients~\eqref{eq:Lnm} are accumulated from all panels. The resulting
cost is
\begin{equation}\label{eq:hybridcost}
O\!\Big(\big(N_{\rm far}+N_{\rm near}\,(P+1)^2\big)\,N_p\,n_{\rm gl}^2\Big),
\end{equation}
which, for the small near-surface fractions typical of the target workflows,
is far below the panel-driven $O\!\big(N_p^2\,n_{\rm gl}^2\,(P+1)^2\big)$. The
nearest-panel classification is an $O(N_p)$ scan per target as written; replacing
it with a bounding-volume-hierarchy \textcolor{red}{(BVH)} query --- bit-for-bit identical in result,
$O(\log N_p)$ per target --- is one of the accelerations of
Section~\ref{sec:acceleration}.

\paragraph{Orphan targets and the fallback.}
Placing centres beneath panel centroids does not guarantee that every near
target falls inside some convergence ball. On a structured mesh the balls of the
panels meeting at a shared vertex are each pushed a distance $r$ along their
normals, leaving an uncovered gap directly above the vertex: a target sitting
there is $\sim r\sqrt{2}$ from the nearest centre, outside every ball. Such a
target is an \emph{orphan} (the crossed point in
Figure~\ref{fig:qbx_schematic}). AMPERe detects orphans and, for each, spawns a
one-shot \emph{target-specific} centre: it takes the nearest panel's outward
normal, orients it to the requested side, and places the centre so that the
orphan lies inside the new ball ($q<1$). If even this fails --- the target is too
far from the surface for any ball, or lies on the opposite side from the one
requested --- the dispatcher falls back to Direct Gauss--Legendre on that single
target. The evaluator therefore always returns a value and  \emph{"always returns a finite value".}

Placing an expansion centre directly beneath an uncovered target corresponds
to the \emph{target-specific} QBX strategies of Siegel and Tornberg
~\cite{siegel2018tsqbx} and Wala and Kl\"ockner~\cite{wala2020tse}; the orphan
fallback itself is therefore not a new expansion scheme. The contribution of
AMPERe lies in the surrounding robustness layer required for reliable
large-scale use: automatic orphan detection, side-aware centre placement that
keeps the expansion on the target side even for tilted tangent planes of curved
meshes (where a naive ``flip toward the target'' rule can silently select the
wrong side), a guaranteed-return evaluation strategy, and integration of this
fallback mechanism across all AMPERe evaluation paths.

\paragraph{Two-sided (straddling) targets}
A further complication, generic to current-carrying surfaces, arises when near
targets exist on both sides of the source surface simultaneously, as occurs
when an evaluation grid crosses a surrounding surface and samples points both
inside and outside it. The dispatcher described above assumes a single global
orientation: each QBX centre is placed on either the outward
($+\mathbf n$) or inward ($-\mathbf n$) side of the nearest panel
(the $\mathbf c^\pm$ of \eqref{eq:centre}). Such a choice is valid only for
targets on the corresponding side of the surface; targets on the opposite side
cannot use those centres and revert to the direct quadrature path. Hence, a
single global orientation is insufficient for a straddling target cloud.

AMPERe resolves this through a two-sided dispatch using the same underlying
evaluator. Each target is first assigned a side label; the targets are then
separated into two groups, each processed with the matching QBX centre
orientation $\mathbf{c}^{\pm}$, and the resulting fields are scattered back into
a single output array.

\textcolor{red}{For the toroidal surfaces considered here, the label is obtained
from a point-in-polygon test in the poloidal plane. At the target's cylindrical
angle $\varphi$, the source surface is intersected with the plane
$\varphi=\mathrm{const}$, giving the closed contour
$\{R(\theta,\varphi),Z(\theta,\varphi)\}$, reconstructed by linear
interpolation between the two stored cross-sections bracketing $\varphi$.
The target $(R,Z)$ is classified using the even--odd ray-crossing rule: the
number of contour segments that straddle the target height $Z$ and intersect
the horizontal ray at $R_\times>R$ is counted; an odd count indicates an
interior point and an even count an exterior point. The test requires no
convexity assumption and applies to any simple closed contour, making it
suitable for the strongly shaped, non-convex W7-X cross-sections.}

\textcolor{red}{The dispatch uses only this binary label and is therefore
independent of how it is obtained. The particular poloidal test, however,
assumes that each plane $\varphi=\mathrm{const}$ intersects $S$ in a single
simple closed contour. This holds for the LCFS and winding surfaces considered
here, but not for arbitrary source geometries. When this assumption does not
hold---for example, for a sphere, an open wall, or a coil---the same binary
label can instead be supplied by another inside--outside classifier, such as
the local signed-distance test already implemented in the code, or a
solid-angle/winding-number test for a closed surface. Only the classification
step changes; the subsequent two-sided evaluation is unchanged.}

Together, the target-driven dispatch, orphan fallback, and two-sided
classification provide a robust near/far evaluation framework. The associated
accelerations --- BVH classification, refined near-panel quadrature, and
target-specific near-field evaluation --- are presented in
Section~\ref{sec:acceleration}, with their performance and robustness validated
on the benchmark configurations of Sections~\ref{sec:setup}--\ref{sec:results}.

\section{Structure preservation and error decomposition}
\label{sec:structure}

Two properties of the QBX field must be kept apart. One is \emph{accuracy} --- how
close $\mathbf{B}_P$ lies to the true field, as discussed in
Section~\ref{sec:qbx}.  \textcolor{red}{In this section, the subscript instead
denotes the QBX expansion order.} The other is \emph{structure preservation}: whether
$\mathbf{B}_P$ is, irrespective of its accuracy, a divergence- and curl- free field.
It is important to note, however, that the requirement for the field to be curl-free
arises from our focus on vacuum regions, where electric currents are absent.
The Path-2 expansion of \eqref{eq:paths} is exactly solenoidal at every order
(Proposition~\ref{prop:divfree}) yet only asymptotically irrotational
(Remark~\ref{rem:curl}), and its accuracy admits \textcolor{red}{an \emph{a priori}} decomposition
(Proposition~\ref{prop:errdec}).

\begin{proposition}[Exact solenoidality at every order]
\label{prop:divfree}
For every \textcolor{red}{expansion} order $P\ge0$ and arbitrary coefficients
$\{\mathbf{L}_n^m\}$, the Path-2 field \eqref{eq:paths} is pointwise
divergence-free away from the centre:
$\nabla\cdot\mathbf{B}_P(\mathbf{r})\equiv 0$ for $\mathbf{r}\ne\mathbf{c}$.
\end{proposition}

\begin{proof}
Each coefficient $\mathbf{L}_n^m$ is a constant vector, so
$\nabla R_n^m\times\mathbf{L}_n^m=\nabla\times(R_n^m\,\mathbf{L}_n^m)$. Hence
$\mathbf{B}_P$ is an exact curl, $\mathbf{B}_P=\nabla\times\mathbf{A}_P$ with the
Path-1 potential $\mathbf{A}_P=\mu_0\sum_{n=0}^{P}\sum_{m=-n}^{n}
R_n^m(\mathbf{r}-\mathbf{c})\,\mathbf{L}_n^m$, and
$\nabla\cdot(\nabla\times\mathbf{A}_P)\equiv0$ identically.
\end{proof}

The statement holds \emph{before any convergence}: $\mathbf{B}_P$ is solenoidal
even where its value is still far from the true field, so solenoidality is a
property of the representation, not of its accuracy.

\begin{remark}[The curl asymmetry]
\label{rem:curl}
The true vacuum field of a surface current is also curl-free off the surface,
but $\mathbf{B}_P$ is only asymptotically so. With $\mathbf{B}_P=\nabla\times
\mathbf{A}_P$, the identity $\nabla\times(\nabla\times\mathbf{A}_P)=
\nabla(\nabla\cdot\mathbf{A}_P)-\Delta\mathbf{A}_P$, and $\Delta R_n^m=0$ (so
$\Delta\mathbf{A}_P=0$),
\begin{equation}\label{eq:curl}
\nabla\times\mathbf{B}_P=\nabla(\nabla\cdot\mathbf{A}_P),
\qquad
\nabla\cdot\mathbf{A}_P=\mu_0\sum_{n=0}^{P}\sum_{m=-n}^{n}
\mathbf{L}_n^m\cdot\nabla R_n^m\neq0\ \text{ in general.}
\end{equation}
Thus $\mathbf{B}_P$ is exactly solenoidal but not, in general, irrotational.
The residue is sharper than a truncation bound, however:
$\nabla\cdot\mathbf{A}_P$ is a sum of homogeneous polynomials of
\emph{distinct} degrees, so if the exact $\nabla\cdot\mathbf{A}$ vanishes ---
the case whenever the source current is conserved,
$\nabla_s\cdot\mathbf{K}=0$, as it is analytically for every Stratton--Chu
and REGCOIL trace --- each degree must vanish separately, and $\mathbf{B}_P$
with \emph{exact} coefficients is exactly irrotational at every order $P$.
The measured curl of a computed expansion is therefore the imprint of the
\emph{coefficient} errors --- source sampling, panel quadrature, surface
representation --- and is essentially independent of $P$. The asymmetry is a
useful diagnostic: $\nabla\cdot\mathbf{B}_P\equiv0$ tests the implementation,
while $\nabla\times\mathbf{B}_P$ measures the fidelity of the computed source
data, a quantity to which the divergence is structurally blind
(Section~\ref{sec:res_divfree}).
\end{remark}

\begin{proposition}[A-priori error decomposition]
\label{prop:errdec}
For a target $\mathbf{t}$ served by a centre $\mathbf{c}$ with convergence ratio
$q=q(\mathbf{t})\le1$, the Path-2 error splits into three independently
controllable contributions,
\begin{equation}\label{eq:errdec}
\|\mathbf{B}_P(\mathbf{t})-\mathbf{B}(\mathbf{t})\|
\;\lesssim\;
\underbrace{C_{\rm t}\,q^{P+1}}_{\text{(i) truncation}}
\;+\;
\underbrace{C_{\rm q}(P)\,\rho^{-2n_{\rm gl}}}_{\text{(ii) coefficient quadrature}}
\;+\;
\underbrace{C_{\rm g}\,h^{g+1}}_{\text{(iii) surface representation}} ,
\end{equation}
with $\rho>1$ a geometric convergence factor defined below.
\end{proposition}

We do not re-derive these rates but assemble established estimates matched to
AMPERe's discretization. Term~(i), the local-expansion truncation, converges
geometrically inside the ball at rate $q^{P+1}$~\cite{egk2013,kt2017}. Term~(iii),
the displacement of the order-$g$ discrete surface from the true one, is
$O(h^{g+1})$ and independent of $P$. Term~(ii) is the Gauss--Legendre error in forming the
coefficients~\eqref{eq:Lnm_num}: on each panel the integrand
$S_n^{-m}(\cdot-\mathbf{c})$ is analytic, its nearest singularity being the
centre $\mathbf{c}$ itself, so $n_{\rm gl}$-point Gauss--Legendre converges
geometrically as $\rho^{-2n_{\rm gl}}$, where $\rho>1$ is the Bernstein-ellipse
parameter of the reference-interval quadrature~\cite{trefethen2008}; the
corresponding error estimate for QBX coefficient quadrature is that of af
Klinteberg and Tornberg~\cite{kt2017}.

\begin{remark}[Structure preservation is independent of accuracy]
\label{rem:orthogonal}
Propositions~\ref{prop:divfree} and~\ref{prop:errdec} describe independent
properties of the representation: the discrete field satisfies
$\nabla\cdot\mathbf{B}_P\equiv0$ at every truncation order, irrespective of the
magnitude of the approximation error in \eqref{eq:errdec}. Thus, even when the
field approximation is not yet converged, it remains an exactly solenoidal
field. This distinction is important for applications such as field-line
tracing and guiding-centre integration, where spurious $\nabla\cdot\mathbf B$
introduces artificial magnetic sources that violate the underlying
flux-conservation structure and can contaminate long-time particle trajectories. The
spherical-harmonic curl representation therefore preserves a structural
property that is generally lost when the field is reconstructed component by
component through interpolation or finite differences.
\end{remark}

Realizing this structure-preserving representation for the vector magnetostatic
kernel, rather than a scalar potential kernel, is one of the central \textcolor{red}{contributions}
of this work.

\section{Acceleration and portability}
\label{sec:acceleration}

AMPERe achieves performance through two complementary ingredients:
algorithmic reduction of the near-surface evaluation cost and a
performance-portable implementation. The former reorganizes the QBX
calculation so that the same field can be obtained with substantially fewer
operations; the latter is provided by a single header-only implementation
based on the Kokkos framework~\cite{Kokkos_2014}, which executes unchanged on
CPUs, GPUs, and other supported architectures.

Beyond the hybrid strategy of Section~\ref{sec:hybrid}, AMPERe provides two
additional acceleration mechanisms. The \emph{local-correction} formulation
preserves the source/target separation of the standard QBX construction:
only a localized subset of the near interactions is corrected explicitly,
while the remaining contribution retains the structure required by fast
algorithms such as the fast multipole method (FMM) and boundary-integral
solvers. The \emph{target-specific} formulation takes the opposite approach:
it evaluates the same truncated expansion by contracting the coefficients
directly for each target, eliminating the $(P+1)^2$ coefficient set in favour
of $(P+1)$ partial sums. This preserves the field to round-off accuracy but
sacrifices the source/target separability exploited by FMM-based methods.
The two approaches are therefore complementary and are available through the
same interface described in Section~\ref{sec:hybrid}.
\textcolor{red}{Table~\ref{tab:terminology} fixes the naming used for these
four ingredients throughout the paper and in the figure legends: the first two
are \emph{drivers} --- alternative ways to organize the same QBX evaluation ---
while the last two are \emph{accelerations} applied within the target-driven
driver.}

\begin{table}[t]
\centering
\small
\color{red}
\caption{Terminology used throughout the paper. The two drivers are
alternative organizations of the same QBX evaluation; the two accelerations
operate within the target-driven driver.}
\label{tab:terminology}
\setlength{\tabcolsep}{5pt}
\begin{tabular}{llp{5.6cm}l}
\hline
Role & Name & Mechanism & Introduced \\
\hline
Driver & panel-driven QBX &
  one centre per panel, all expansions formed upfront &
  \S\ref{sec:qbx} \\
Driver & target-driven QBX &
  one centre per near target, formed on demand; near/far dispatch &
  \S\ref{sec:hybrid} \\
Acceleration & local correction &
  near/far split of the panel sum: QBX correction on the near set only &
  \S\ref{sec:acceleration} \\
Acceleration & target-specific evaluation (TSE) &
  reordered summation of the same expansion; no coefficient array &
  \S\ref{sec:accel_tse} \\
\hline
\end{tabular}
\end{table}

This section introduces the underlying algorithms; their performance and
accuracy are assessed in Section~\ref{sec:results}. Algorithm~\ref{alg:ampere}
summarizes the complete AMPERe pipeline, including geometric preprocessing,
the two-sided target-driven dispatcher of Section~\ref{sec:hybrid}, and the
four interchangeable evaluation strategies. Each step is annotated with its
mathematical operation, software realization, and associated correctness
guarantee, while Table~\ref{tab:complexity} summarizes the resulting operation
counts.

\begin{algorithm}[!tp]
\caption{The AMPERe evaluation pipeline: the magnetostatic field
$\mathbf{B}$ of known surface densities $(\mathbf{K},\sigma)$ on an
arbitrary three-dimensional surface
(Sections~\ref{sec:qbx}--\ref{sec:acceleration}). Every phase is a uniform
Kokkos kernel, portable CPU/GPU without code change
(\S\ref{sec:accel_kokkos}); the marginal comments name the component or
guarantee attached to each step. The surface-charge term of~\eqref{eq:VCP}
is evaluated by the same phases through the gradient kernels acting on
$\sigma$.}
\label{alg:ampere}
\footnotesize
\begin{algorithmic}[1]
\Require order-$g$ panel surface ($N_p$ panels, $n_{\rm gl}^2$
  Gauss--Legendre nodes each,~\eqref{eq:disc}); densities
  $\mathbf{K},\sigma$~\eqref{eq:Ksigma}; targets $T$; expansion order $P$;
  threshold $\tau$; slack $\epsilon_{\rm tgt}$; strategy
  $\mathcal{S}\in\{\textsc{pd},\textsc{td},\textsc{lc},\textsc{tse}\}$
  (panel-driven, target-driven, local-correction, target-specific);
  near-set factor $\alpha$; subdivision factor $k$
\Ensure $\mathbf{B}(\mathbf{t})$ for \emph{every} target
  (guaranteed return, \S\ref{sec:hybrid}); each QBX-served value is an
  exact curl, $\nabla\!\cdot\!\mathbf{B}\equiv0$ at machine precision for
  all $P$ (Prop.~\ref{prop:divfree}), with analytic $\nabla\mathbf{B}$
\Statex\textbf{Phase A --- geometric setup (once).}
\State panel index: centroids $\mathbf{r}_{0,p}$, normals $\mathbf{n}_p$,
  sizes $h_p=\sigma_1([\mathbf{t}_\zeta\ \mathbf{t}_\eta])$~\eqref{eq:centre}
  \Comment{precomputed, \S\ref{sec:accel_exact}}
\State build BVH over $\{\mathbf{r}_{0,p}\}$
  \Comment{ArborX; $O(N_p)$ build, $O(\log N_p)$/query, bit-identical}
\State \textbf{if} $\mathcal{S}=\textsc{pd}$: place
  $\mathbf{c}_p^{\pm}=\mathbf{r}_{0,p}\pm\tfrac{h_p}{2}\mathbf{n}_p$, form
  all $\mathbf{L}_n^m$~\eqref{eq:Lnm_num}
  \Comment{form once, reuse per ball}
\Statex\textbf{Phase B --- two-sided classification (\S\ref{sec:hybrid}).}
\State split $T$ into outer $T^{+}$/inner $T^{-}$ (signed poloidal
  point-in-polygon test)
  \Comment{straddling targets}
\State \textbf{for} each side $s\in\{+,-\}$ with $T^{s}\neq\emptyset$:
  run Phases C--E on $T^{s}$
\Statex\textbf{Phase C --- near/far dispatch~\eqref{eq:dispatch}.}
\State batched BVH query: nearest panel $p(\mathbf{t})$,
  $d=|\mathbf{t}-\mathbf{r}_{0,p}|$
\State $F\gets\{\mathbf{t}:d/h_p>\tau\}$;\quad
  $N\gets\{\mathbf{t}:d/h_p\le\tau\}$
  \Comment{segregated streams: SIMD/GPU-uniform}
\Statex\textbf{Phase D --- far stream.}
\State \textbf{for all} $\mathbf{t}\in F$:
  $\mathbf{B}(\mathbf{t})\gets\mathrm{DGL}_{\rm all}(\mathbf{t})$
  \Comment{smooth rule~\eqref{eq:disc} is accurate here}
\Statex\textbf{Phase E --- near stream.}
\For{$\mathbf{t}\in N$}
  \State $\mathbf{c}\gets$ nearest precomputed $\mathbf{c}_p^{\,s}$
    (\textsc{pd}) or one-shot
    $\mathbf{r}_{0}+s\tfrac{h}{2}\mathbf{n}$;\quad
    $q\gets|\mathbf{t}-\mathbf{c}|/(h/2)$
  \State \textbf{if} $q>1+\epsilon_{\rm tgt}$: re-place $\mathbf{c}$
    beneath $\mathbf{t}$, side-aware on side $s$
    \Comment{orphan rescue, \S\ref{sec:hybrid}}
  \State \textbf{if} still $q>1+\epsilon_{\rm tgt}$:
    $\mathbf{B}(\mathbf{t})\gets\mathrm{DGL}_{\rm all}(\mathbf{t})$,
    \textbf{next} $\mathbf{t}$
    \Comment{fallback: a value,  "always returns a finite value"}
  \State subdivide panels nearest $\mathbf{c}$ into $k\times k$ sub-cells
    \Comment{controls term (ii) of~\eqref{eq:errdec}}
  \State \textbf{case} $\textsc{pd},\textsc{td}$: reuse/form
    $\mathbf{L}_n^m$~\eqref{eq:Lnm_num};
    $\mathbf{B}(\mathbf{t})\gets\mu_0\!\sum_{n\le P,m}\!
      \nabla R_n^m(\mathbf{t}-\mathbf{c})\times\mathbf{L}_n^m$
    \Comment{separable: FMM/BIE substrate}
  \State \textbf{case} $\textsc{lc}$:
    $\mathbf{B}(\mathbf{t})\gets\mathrm{DGL}_{\rm all}(\mathbf{t})
      +\sum_{p\in\mathcal{N}(\mathbf{t})}
      \big[\mathrm{QBX}_p-\mathrm{DGL}_p\big]$~\eqref{eq:lc},
    $\mathcal{N}(\mathbf{t})$ by BVH radius query
    $\alpha h$~\eqref{eq:nearset}
    \Comment{recovers \textsc{td} as $\alpha\to\infty$}
  \State \textbf{case} $\textsc{tse}$:
    $\mathbf{B}(\mathbf{t})\gets\frac{\mu_0}{4\pi}\sum_{q}
      \nabla_{\mathbf{t}}\,g_P(\mathbf{t},\mathbf{r}'_q)\times
      \widetilde{\mathbf{K}}_q$~\eqref{eq:tse_fields}
    \Comment{Legendre recurrence; round-off $=$ \textsc{td}}
\EndFor
\Statex\textbf{Phase F --- merge.}
\State scatter the $T^{+}$ and $T^{-}$ results into one output array;
  \textbf{return} $\mathbf{B}$ (and analytic $\nabla\mathbf{B}$ on request)
\end{algorithmic}
\end{algorithm}

\subsection{Exact algorithmic optimizations and BVH classification}
\label{sec:accel_exact}

Three exact transformations reduce the evaluation cost without changing the
computed field. First, all geometric quantities required by the dispatcher ---
the panel centroids and the panel-size measure $h$ from \eqref{eq:centre} ---
are precomputed once and stored with the panel data. The per-target stage
therefore avoids repeated geometric evaluations and only performs the final
distance-based classification.

Second, the dispatch is \emph{segregated}: targets are partitioned into far-
and near-field streams, which are processed independently by uniform Direct
Gauss--Legendre and target-driven QBX kernels, respectively. This avoids
mixed execution paths and can improve vectorization and GPU utilization.

Third, the nearest-panel search is accelerated. A linear search over all panel
centroids requires $O(N_p)$ work per target, resulting in an overall
$O(N_tN_p)$ cost. Instead, a bounding-volume hierarchy constructed once
over the panel centroids using the ArborX library~\cite{ArborX_2020} reduces
the cost to
\begin{equation}\label{eq:bvhcost}
\underbrace{O(N_p)}_{\text{one-time construction}}
\;+\;
\underbrace{O(N_t\log N_p)}_{\text{batched nearest-panel queries}} .
\end{equation}
The BVH returns only the nearest-panel index. The corresponding distance and
the ratio $d/h$ required by the dispatch criterion \eqref{eq:dispatch} are
then evaluated explicitly.

\subsection{Local-correction evaluation}
\label{sec:accel_lc}

For a near target, the near-singular behavior of the kernel is localized to
only a small subset of panels surrounding the target; contributions from the
remaining panels are smooth and can be evaluated accurately by ordinary
Gauss--Legendre quadrature. Nevertheless, the target-driven QBX branch of
Section~\ref{sec:hybrid} applies the expansion machinery to the complete panel
set, accumulating QBX coefficients from all $N_p$ panels and therefore paying
the $(P+1)^2$ contraction cost even for well-separated interactions.

The local-correction formulation exploits this separation. The panel set is
split into a near region $\mathcal{N}(\mathbf{x})$ containing the panels that
require QBX treatment and a complementary far region
$\overline{\mathcal{N}}(\mathbf{x})$ evaluated directly by Gauss--Legendre
quadrature. Using the decomposition
\begin{equation}\label{eq:lc}
\mathbf{B}(\mathbf{x})
=\underbrace{\mathrm{DGL}_{\rm all}(\mathbf{x})}_{\text{smooth rule, all panels}}
\;+\;\sum_{p\in\mathcal{N}(\mathbf{x})}
\Big[\mathrm{QBX}_p(\mathbf{x})-\mathrm{DGL}_p(\mathbf{x})\Big]
\end{equation}
where $\mathrm{DGL}_p$ denotes the direct Gauss--Legendre contribution of panel
$p$ from \eqref{eq:disc}, and $\mathrm{QBX}_p$ its contribution to the local
expansion \eqref{eq:paths}. The direct contributions of the near panels cancel
explicitly and are replaced by their QBX-regularized values; the far panels
remain evaluated by the cheaper direct rule. Thus, the expansion cost scales
with the number of genuinely near interactions rather than with the total
surface size. The near set is selected by a radius query against the same BVH,

\begin{equation}\label{eq:nearset}
\mathcal{N}(\mathbf{x})=\{\,p:\ \lVert\mathbf{c}_p-\mathbf{x}\rVert<R_{\rm split}\,\},
\qquad R_{\rm split}=\alpha\,h ,
\end{equation}
where $\mathbf{c}_p$ denotes the panel centroid and $\alpha$ is a
dimensionless radius factor controlling the size of the local correction
region. The only additional error introduced with respect to the QBX
decomposition of \eqref{eq:errdec} is the direct-quadrature error from the
nearest excluded panel. Since this panel lies at a distance
$\gtrsim R_{\rm split}$ from the target, its contribution is smooth and the
associated quadrature error decreases geometrically as $\alpha$ is increased.
In the limit $\alpha\rightarrow\infty$, the local-correction formulation
recovers the full-QBX result, while the computational advantage decreases.
The cost per near target is reduced from the hybrid baseline to
\begin{equation}\label{eq:lccost}
\underbrace{O\big(N_p\,n_{\rm gl}^2\big)}_{\text{$N$-body floor}}
\;+\;
\underbrace{O\big(\overline{N}_{\mathcal N}\,n_{\rm gl}^2\,(P+1)^2\big)}_{\text{sparse correction}},
\qquad \overline{N}_{\mathcal N}=O(1),
\end{equation}
The $(P+1)^2$ expansion cost is therefore no longer multiplied by the total
number of panels: a near target requires QBX only for the small local
correction region, while the remaining panels contribute through the same
direct smooth evaluation used for well-separated targets. The remaining
$O(N_p n_{\rm gl}^2)$ direct sweep is the classical $N$-body bottleneck; its
removal requires replacing the far-field contribution
$\mathrm{DGL}_{\rm all}$ by a fast summation method. The decomposition in
\eqref{eq:lc} is therefore naturally compatible with established fast
frameworks, including the FMM~\cite{greengard1997fmm} and QBX--FMM
couplings~\cite{rachh2017,wala2018gigaqbx}, since these methods can accelerate
the smooth far-field term without modifying the local QBX correction.

\begin{figure}[t]
  \centering
  \IfFileExists{figures/fig_lc_spheres.png}
    {\includegraphics[width=0.62\linewidth]{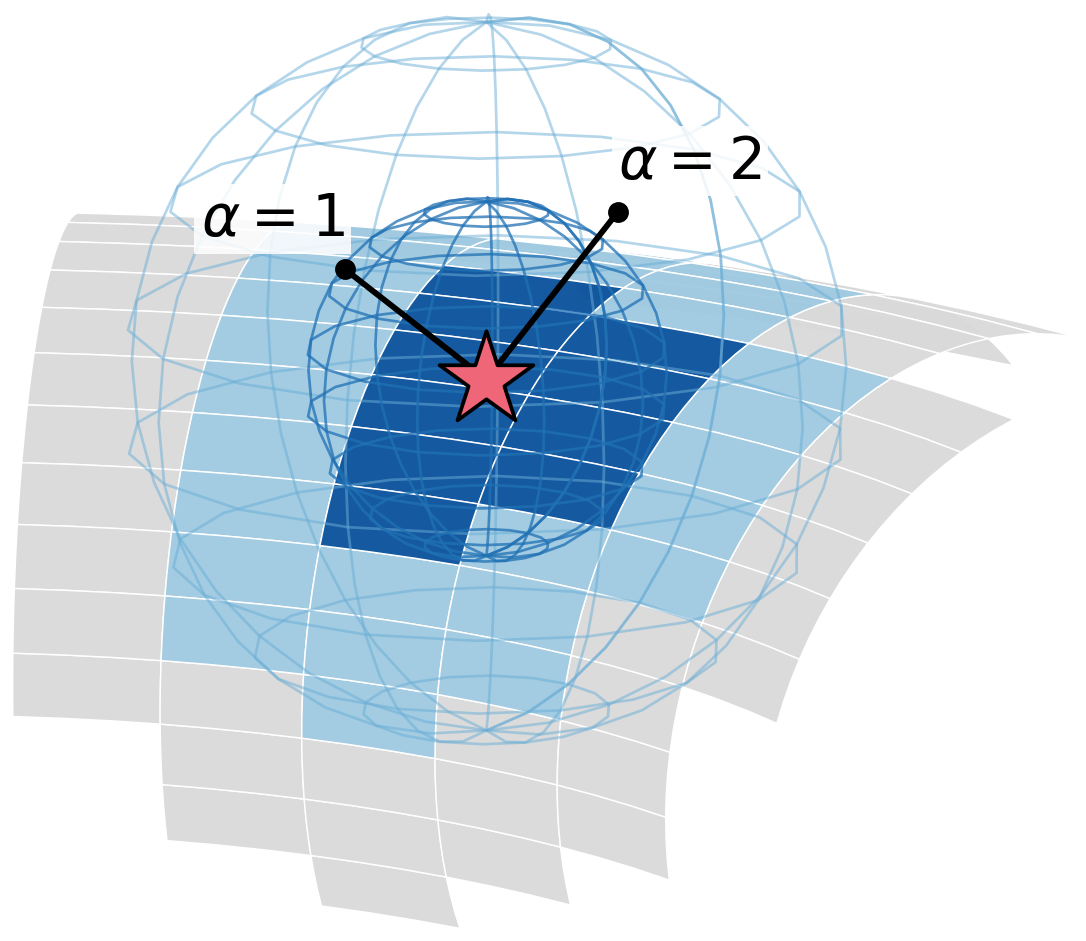}}
    {\fbox{\parbox[c][0.4\linewidth][c]{0.7\linewidth}{\centering\itshape
      [placeholder: local-correction near-set geometry ---
      torus, near target, nested $R_{\rm split}=\alpha h$ spheres,
      centroids split near/far; export to
      \texttt{figures/fig\_lc\_spheres.png}]}}}
\caption{\textcolor{red}{Geometry of the local-correction split \eqref{eq:lc} on the
toroidal-coil surface introduced in Section~\ref{sec:setup}. The tiles represent
the source-surface panels, viewed from slightly off the surface normal. A
near-surface target $\mathbf{x}$ (red star) is associated with two spheres of
radius $R_{\rm split}=\alpha h$, shown for $\alpha=1$ and $\alpha=2$ and
indicated by the corresponding radius leaders. A panel is assigned to the near
set $\mathcal{N}(\mathbf{x})$ if its centroid lies within the sphere. The
dark-blue panels are selected for $\alpha=1$, the light-blue panels are added
when $\alpha$ is increased to $2$, and the grey panels lie outside both
spheres. Panels in $\mathcal{N}(\mathbf{x})$ are evaluated with the QBX
correction in~\eqref{eq:lc}, whereas all remaining panels are evaluated using
the smooth-panel quadrature alone. Thus, increasing $\alpha$ enlarges the
region treated by the local correction while increasing its computational
cost. The resulting accuracy--cost trade-off is quantified in
Fig.~\ref{fig:w7x_accel}.}}
  \label{fig:lc_spheres}
\end{figure}

\paragraph{Near-panel subdivision.}
A separate accuracy-driven refinement targets the coefficient quadrature.
When a source panel is close to the expansion centre, its Bernstein
parameter \eqref{eq:errdec} is minimized and the panel dominates term (ii) of
\eqref{eq:errdec}. AMPERe therefore subdivides such panels into $(k\times k)$
sub-cells, reducing the effective panel size to $(h/k)$ yielding geometric
improvement of the Gauss--Legendre error, while increasing quadrature
cost only locally by $(k^2)$. The same refinement is used for both coefficient
and target-specific paths.

\subsection{Target-specific evaluation}
\label{sec:accel_tse}
The second method accelerates the near branch by changing the \emph{order of
summation}, not the underlying expansion. In the coefficient path of
Section~\ref{sec:qbx}, each of the $(N_s=N_p n_{\rm gl}^2)$ source nodes
contributes to $((P+1)^2)$ coefficients, requiring
$(O(N_s(P+1)^2))$ work per centre. For target-driven dispatch, each target has
a private centre, so this coefficient array is constructed, contracted, and
discarded without reuse. The \emph{target-specific expansion} (TSE) eliminates
this intermediate array: using the Legendre addition theorem, the
$(m)$-summation collapses into a single Legendre polynomial, converting the
truncated addition theorem \eqref{eq:addition} into a one-dimensional kernel
for a source $(\mathbf{r}')$ and target $(\mathbf{r})$ sharing centre
$(\mathbf{c})$,

\begin{equation}\label{eq:tse}
4\pi\sum_{n=0}^{P}\sum_{m=-n}^{n}\frac{1}{2n+1}\,
R_n^m(\mathbf{r}-\mathbf{c})\,S_n^{-m}(\mathbf{r}'-\mathbf{c})
\;=\;
\sum_{n=0}^{P}\frac{u^{\,n}}{s^{\,n+1}}\,P_n(\cos\gamma)
\;\equiv\;
g_P(\mathbf{r},\mathbf{r}';\mathbf{c}) ,
\end{equation}
with
\begin{equation}\label{eq:tse_vars}
u=|\mathbf{r}-\mathbf{c}|,\qquad
s=|\mathbf{r}'-\mathbf{c}|,\qquad
\cos\gamma=\frac{(\mathbf{r}-\mathbf{c})\cdot(\mathbf{r}'-\mathbf{c})}{u\,s},
\end{equation}
and $P_n$ the Legendre polynomials, evaluated by their stable three-term
recurrence. \textcolor{red}{This step destroys the factorisation underlying
\eqref{eq:paths}--\eqref{eq:Lnm}. In \eqref{eq:solid}, $R_n^m$ depends only on
the target and $S_n^{-m}$ only on the source, whereas $g_P$ couples them
through $\cos\gamma$. Hence, the target-independent coefficients
$\mathbf{L}_n^m$ of~\eqref{eq:Lnm} have no counterpart here; the consequences
are discussed below.} The identity \eqref{eq:tse} is \emph{exact at every truncation
order} $P$, so TSE is not an approximation of the QBX
field but the identical field summed differently; the two implementations
agree to round-off, and the structure-preservation results of
Section~\ref{sec:structure} carry over identically. Every quantity the evaluator
returns reduces to $g_P$ and its target gradient: with
$\widetilde{\mathbf{K}}_q = w_q J_q \mathbf{K}(\mathbf{r}'_q)$ the weighted
source samples of \eqref{eq:disc},
\begin{equation}\label{eq:tse_fields}
\mathbf{A}(\mathbf{r})=\frac{\mu_0}{4\pi}\sum_{q}
\widetilde{\mathbf{K}}_q\,g_P(\mathbf{r},\mathbf{r}'_q\textcolor{red}{;\mathbf{c}}),
\qquad
\mathbf{B}(\mathbf{r})=\frac{\mu_0}{4\pi}\sum_{q}
\nabla_{\mathbf{r}}\,g_P(\mathbf{r},\mathbf{r}'_q\textcolor{red}{;\mathbf{c}})
\times\widetilde{\mathbf{K}}_q,
\end{equation}
\textcolor{red}{where $q$ runs over the source quadrature nodes
of~\eqref{eq:disc}, while the target $\mathbf{r}$ and its centre $\mathbf{c}$
are fixed throughout the sum; only the source argument carries the summation
index.} The surface-charge ($\sigma$) term follows likewise through
$\nabla_{\mathbf{r}}\,g_P$. A source now contributes to a target in $O(P+1)$
operations instead of $O\!\big((P+1)^2\big)$ --- a factor $P+1$ fewer per
source--target pair.

The price of TSE is structural rather than numerical. In the target-specific
construction, the effective coefficient
$P_n(\cos\gamma)/s^{\,n+1}$ depends explicitly on the target through the angle
$\gamma$. Formation and evaluation are therefore fused: each target receives its
own expansion, and the resulting representation cannot be translated or reused
for a group of targets. Unlike the coefficient formulation, it does not provide
the separable source representation required by a conventional FMM far field,
where a single expansion must represent the influence of many targets.

This distinction motivates retaining both paths behind the same AMPERe
interface. The coefficient path provides the separable foundation on which
FMM acceleration and boundary-integral solvers can be built, whereas TSE
provides the lowest-cost direct near-field evaluator. Target-specific
expansions for three-dimensional layer potentials were introduced by Siegel
and Tornberg~\cite{siegel2018tsqbx}; the addition-theorem formulation
\eqref{eq:tse} implemented here follows Eq.~(18) of Wala and
Kl\"ockner~\cite{wala2019gigaqbx3d}, who developed its use within the global
QBX framework~\cite{wala2020tse}. Our contribution is to extend this approach
to the vector magnetostatic kernel \eqref{eq:tse_fields} and demonstrate that
the target-specific contraction reproduces the coefficient formulation to
round-off accuracy. Figure~\ref{fig:tse_schematic} contrasts the two data
flows.

\begin{figure}[t]
  \centering
  \IfFileExists{figures/fig_tse_schematic.png}
    {\includegraphics[width=1.0\linewidth]{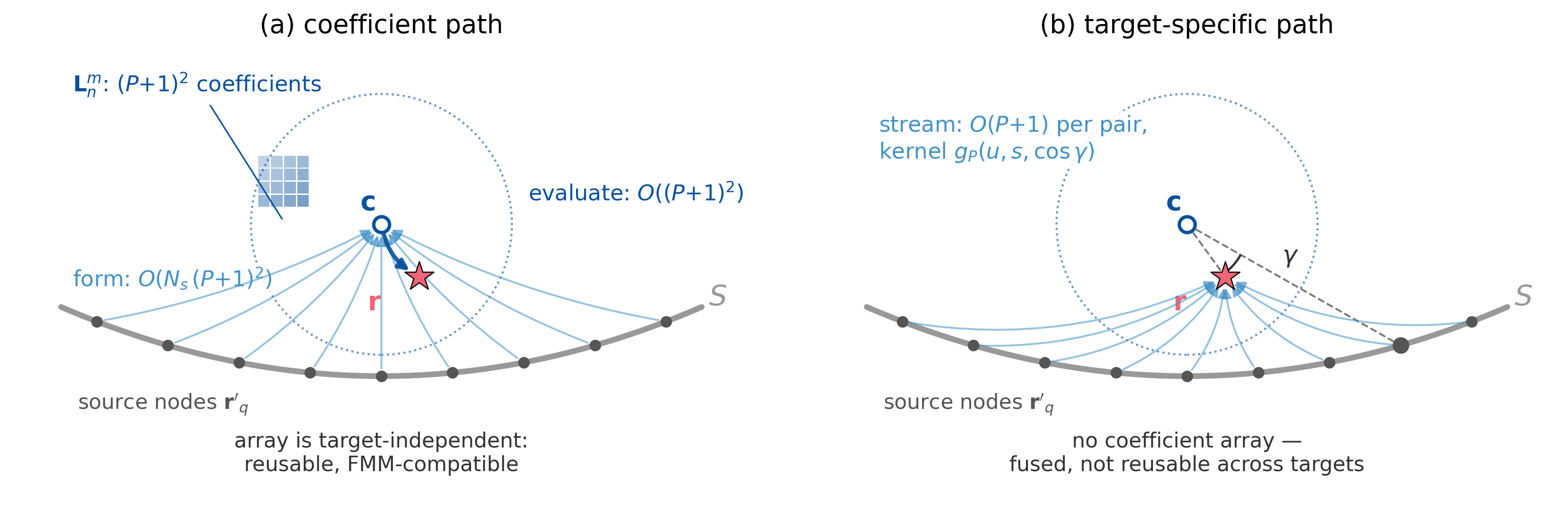}}
    {\fbox{\parbox[c][0.3\linewidth][c]{0.8\linewidth}{\centering\itshape
      [placeholder: TSE schematic --- coefficient path (sources
      $\to$ $(P{+}1)^2$ array $\mathbf{L}_n^m$ at $\mathbf{c}$ $\to$ target)
      vs.\ target-specific path (each source streams to the target through
      $g_P$, $O(P{+}1)$ per pair); export to
      \texttt{figures/fig\_tse\_schematic.png}]}}}
  \caption{Two summation orders of the same truncated expansion.
\emph{Left} (coefficient path, Section~\ref{sec:qbx}): sources accumulate into
the reusable coefficient array \(\mathbf{L}_n^m\), followed by target
contraction. \emph{Right} (target-specific path): sources couple directly to
the target through the one-dimensional kernel \(g_P\) of \eqref{eq:tse},
avoiding the intermediate array at the cost of losing reuse.}
  \label{fig:tse_schematic}
\end{figure}

\subsection{Portability and correctness guarantees}
\label{sec:accel_kokkos}

The implementation is a single header-only C++ code base built on
Kokkos~\cite{Kokkos_2014}, compiling unchanged across CPU, GPU, and other
supported backends. ArborX provides the underlying BVH construction and
queries within the same performance-portable framework~\cite{ArborX_2020}.
The acceleration strategies are therefore algorithmic restructurings rather
than architecture-specific optimizations.

Each transformation is accompanied by a precise correctness guarantee:
geometric preprocessing, indexing, and dispatch preserve the baseline result;
the local-correction formulation converges to the full-QBX field as
$\alpha\rightarrow\infty$; and the target-specific expansion is algebraically
equivalent to the coefficient formulation, with deviations limited to
round-off through \eqref{eq:tse}. Table~\ref{tab:complexity} summarizes the
resulting operation counts.

\begin{table}[t]
\centering
\small
\caption{Operation counts of the evaluation strategies, for $N_p$ panels,
$n_{\rm gl}^2$ quadrature nodes per panel, $N_t$ targets ($N_{\rm near}$ of
them near), expansion order $P$, and mean near-set size
$\overline{N}_{\!\mathcal N}=O(1)$. ``Exact vs.\ baseline'' refers to the
target-driven hybrid dispatcher (Section~\ref{sec:hybrid}).}
\label{tab:complexity}
\setlength{\tabcolsep}{4pt}
\begin{tabular}{p{3.1cm}p{2.6cm}p{3.4cm}p{2.6cm}}
\hline
Strategy & Classification & Cost per near target & Exact vs.\ baseline \\
\hline
Panel-driven QBX (\S\ref{sec:qbx}) &
  --- &
  form $O\!\big(N_p^2 n_{\rm gl}^2 (P{+}1)^2\big)$ once; $O\!\big((P{+}1)^2\big)$ per target &
  same expansion, panel-placed centres \\
Hybrid, linear scan (\S\ref{sec:hybrid}) &
  $O(N_p)$ per target &
  $O\!\big(N_p n_{\rm gl}^2 (P{+}1)^2\big)$ &
  baseline \\
$+$ BVH classification (\S\ref{sec:accel_exact}) &
  $O(N_p)$ build $+$ $O(\log N_p)$ per target &
  unchanged &
  bit-identical \\
$+$ local correction (\S\ref{sec:accel_lc}) &
  same, $+$ radius query &
  $O\!\big(N_p n_{\rm gl}^2 + \overline{N}_{\!\mathcal N} n_{\rm gl}^2 (P{+}1)^2\big)$ &
  tunable ($\alpha$); bitwise as $\alpha\to\infty$ \\
$+$ TSE near field (\S\ref{sec:accel_tse}) &
  same &
  $O\!\big(N_p n_{\rm gl}^2 (P{+}1)\big)$ &
  round-off identical \\
\hline
\end{tabular}
\end{table}

\section{Numerical results}
\label{sec:results}

The numerical study is organized by increasing geometric complexity to isolate
successive sources of numerical error. We begin with the planar sheet of
Section~\ref{sec:setup}, for which the surface is exactly flat. In this
configuration, the kernel singularity is the sole source of numerical
difficulty, free from the effects of surface curvature and the
$O(h^{g+1})$ geometric discretization error of \eqref{eq:errdec}.
Section~\ref{sec:res_sheet_fail} first characterizes the breakdown of the two
smooth quadrature schemes in this setting. Section~\ref{sec:res_sheet_qbx}
then demonstrates that the panel-driven QBX formulation of
Section~\ref{sec:qbx}, together with the orphan-target fallback strategy of
Section~\ref{sec:hybrid}, recovers accurate near-surface field evaluations.
The study then proceeds to curved geometries using the rotating-ellipse
virtual-casing configurations, where geometric discretization and QBX
truncation errors coexist, before concluding with a realistic W7-X
stellarator configuration.

\subsection{Planar sheet: the failure of smooth quadrature}
\label{sec:res_sheet_fail}

The source is the uniform square current sheet introduced in
Section~\ref{sec:setup}, occupying $(x,y)\in[-L,L]^2$ in the plane $z=0$
with $L=5$, and carrying the constant surface current
$\mathbf{K}=K_0\hat{\mathbf{x}}$, where $K_0=1$ in normalized units
($\mu_0/4\pi=1$). For this configuration,
$B_x\equiv0$, the tangential component $B_y$ approaches the infinite-sheet
jump $|B_y|\rightarrow2\pi$ at the surface, and the normal component $B_z$
vanishes only on the symmetry axis $x=y=0$. Consequently, the off-axis
evaluation line shown in Fig.~\ref{fig:ps_failure} exercises both nonzero
field components. The analytical solution is given by
\eqref{eq:sheet_field} in ~\ref{app:fields}, so all errors reported
below are computed directly with respect to the exact field.

The sheet is discretized into $N\times N$ panels of size
$h=2L/N$. Two smooth quadrature schemes are considered: the midpoint rule,
which uses one source point per panel and achieves $O(h^2)$ far-field
accuracy, and the panel Gauss--Legendre rule \eqref{eq:disc}, which employs
$n_{\rm gl}^2$ quadrature nodes per panel and provides spectral far-field
accuracy. Fields are evaluated along the on-axis line $(0,0,z)$ and the
off-axis line $(0.5L,\,0.3L,\,z)$ at fifteen target heights spanning
$z/h_{\rm ref}\in[10^{-4},15]$, where
$h_{\rm ref}=0.078125$ corresponds to the reference discretization
$N_{\rm ref}=128$. Errors are measured using the per-target
$E_{\ell_2}$ norm defined in \eqref{eq:metrics} with $M=1$.

Figure~\ref{fig:ps_failure} presents the results for the reference
discretization ($N=128$; Gauss--Legendre quadrature with
$n_{\rm gl}=10$, $g=3$). In the far field, both quadrature schemes attain
their expected accuracy: at $z/h_{\rm ref}=15$, the midpoint rule has an
error of $1.6\times10^{-5}$, whereas the Gauss--Legendre rule achieves
$8.9\times10^{-15}$. As the target approaches the sheet, however, both
schemes lose accuracy, with the deterioration starting at $z/h_{\rm ref}\approx2$ for the midpoint rule and
$z/h_{\rm ref}\approx1$ for the Gauss--Legendre rule. For
$z/h_{\rm ref}\lesssim10^{-2}$, the errors reach $O(1)$ despite the refined
discretization. The right panel confirms this behavior: both quadratures
accurately reproduce the analytical $B_y$ and $B_z$ fields away from the
sheet but deviate substantially as the target enters the near-singular
region, illustrating the breakdown mechanism discussed in
Section~\ref{sec:breakdown}.

\begin{figure}[t]
  \centering
  \includegraphics[width=0.495\linewidth]{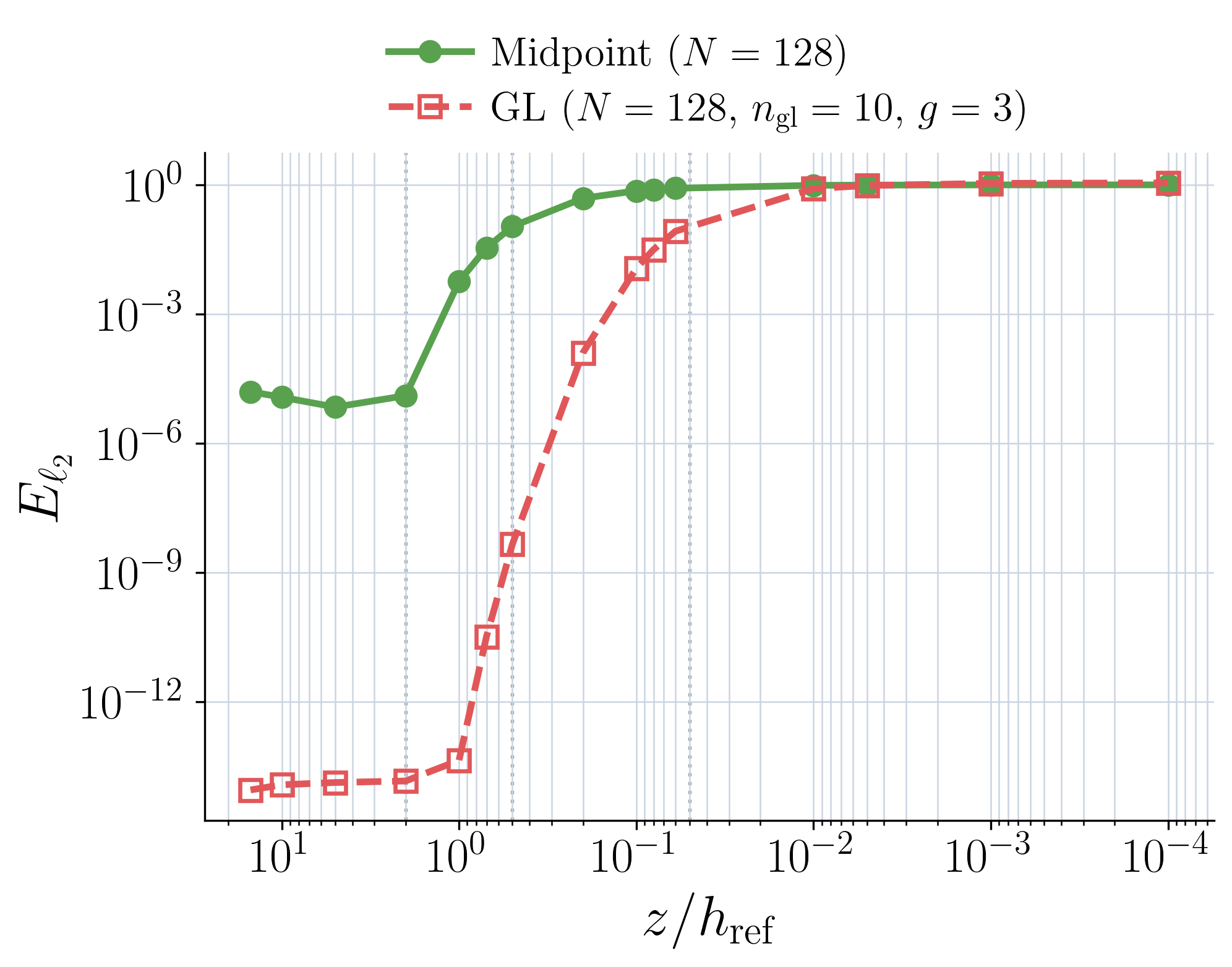}
  \hfill
  \includegraphics[width=0.495\linewidth]{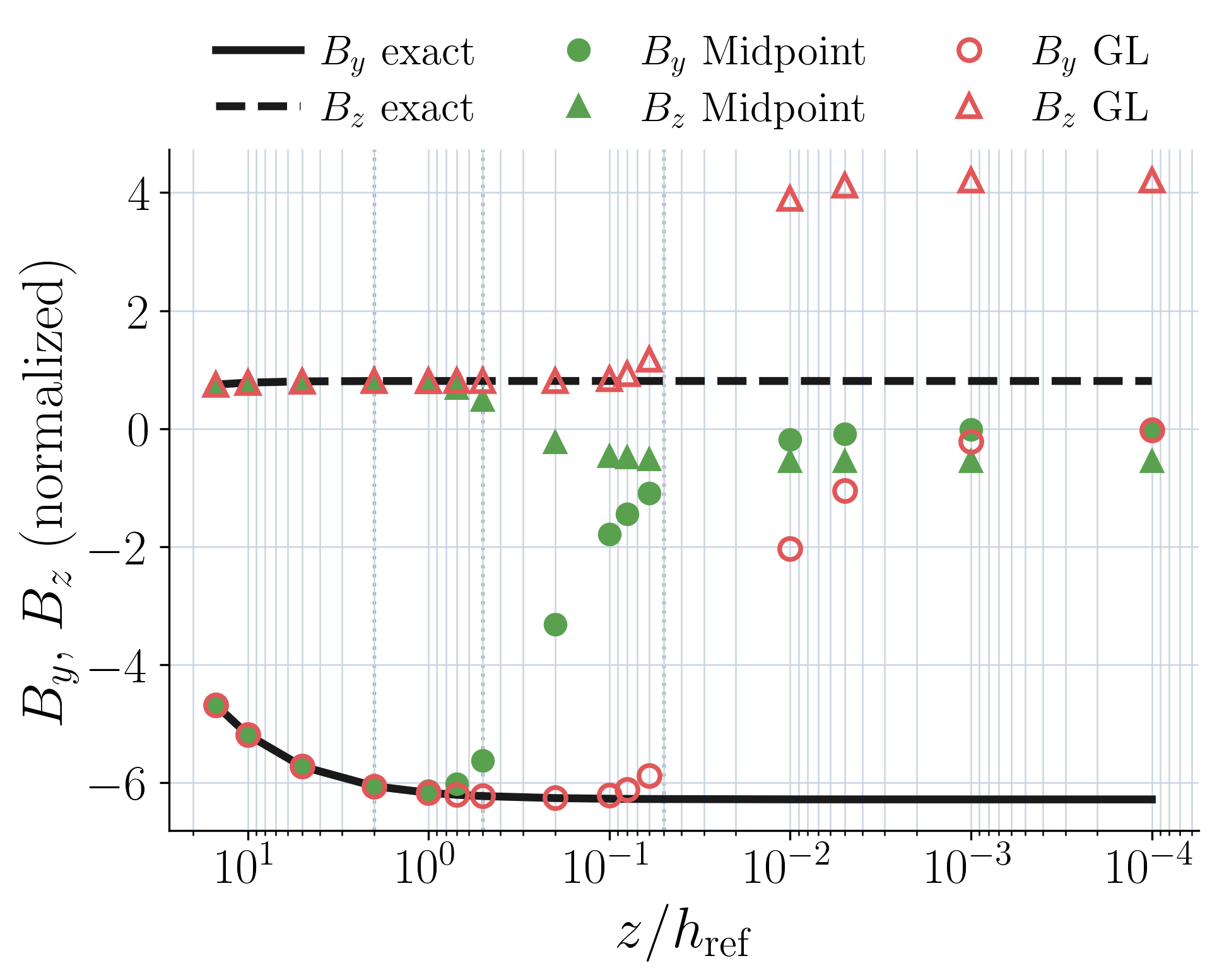}
  \caption{Failure of smooth quadrature for the planar sheet (off-axis column,
$N=128$). \emph{Left:} relative error $E_{\ell_2}$ versus normalized target
distance $z/h_{\rm ref}$. \emph{Right:} computed field components compared with
the analytical solution \eqref{eq:sheet_field}.}
  \label{fig:ps_failure}
\end{figure}

Neither refinement strategy restores near-surface accuracy, and this
conclusion remains unchanged when the methods are compared at equal
computational cost. We measure the cost by the number of kernel evaluations
per target,
\begin{equation}\label{eq:cost}
\mathrm{cost}_{\rm mid}=N^{2},
\qquad
\mathrm{cost}_{\rm GL}=N^{2}n_{\rm gl}^{2},
\end{equation}
where each evaluation corresponds to one source--target contribution
$\mathbf{K}(\mathbf{r}')\times(\mathbf{r}-\mathbf{r}')/
|\mathbf{r}-\mathbf{r}'|^{3}$. Mesh refinement
($N\in\{16,32,64,128,256\}$) produces the expected improvement in the
far field. Near the surface, however, the errors stagnate because resolving
the sharply peaked kernel requires a source resolution determined by the
target distance rather than by the panel size $h$. Increasing the
Gauss--Legendre order ($n_{\rm gl}\in\{1,\ldots,10\}$) similarly improves
far-field accuracy but provides little benefit in the near field. At the
representative distance $z/h_{\rm ref}=0.1$, error-versus-cost comparisons
for both quadrature families show that increasing the computational effort
within smooth quadrature does not efficiently recover near-surface accuracy.
The limitation is therefore intrinsic to the smooth-quadrature formulation,
as discussed in Section~\ref{sec:breakdown}, rather than a consequence of
insufficient quadrature order.

The final test isolates surface-representation effects
(Figure~\ref{fig:ps_geom}). For the flat sheet, the bilinear, biquadratic,
and bicubic mappings ($g=1,2,3$) exactly represent the geometry, leading to
error curves that coincide up to round-off (maximum difference
$1.3\times10^{-12}$ for $N=128$, $n_{\rm gl}=8$). Therefore, the measured
error originates solely from the kernel singularity, while geometric
discretization errors described by \eqref{eq:errdec} arise only for curved
surfaces and are examined in the following sections.

\begin{figure}[t]
  \centering
  \includegraphics[width=0.85\linewidth]{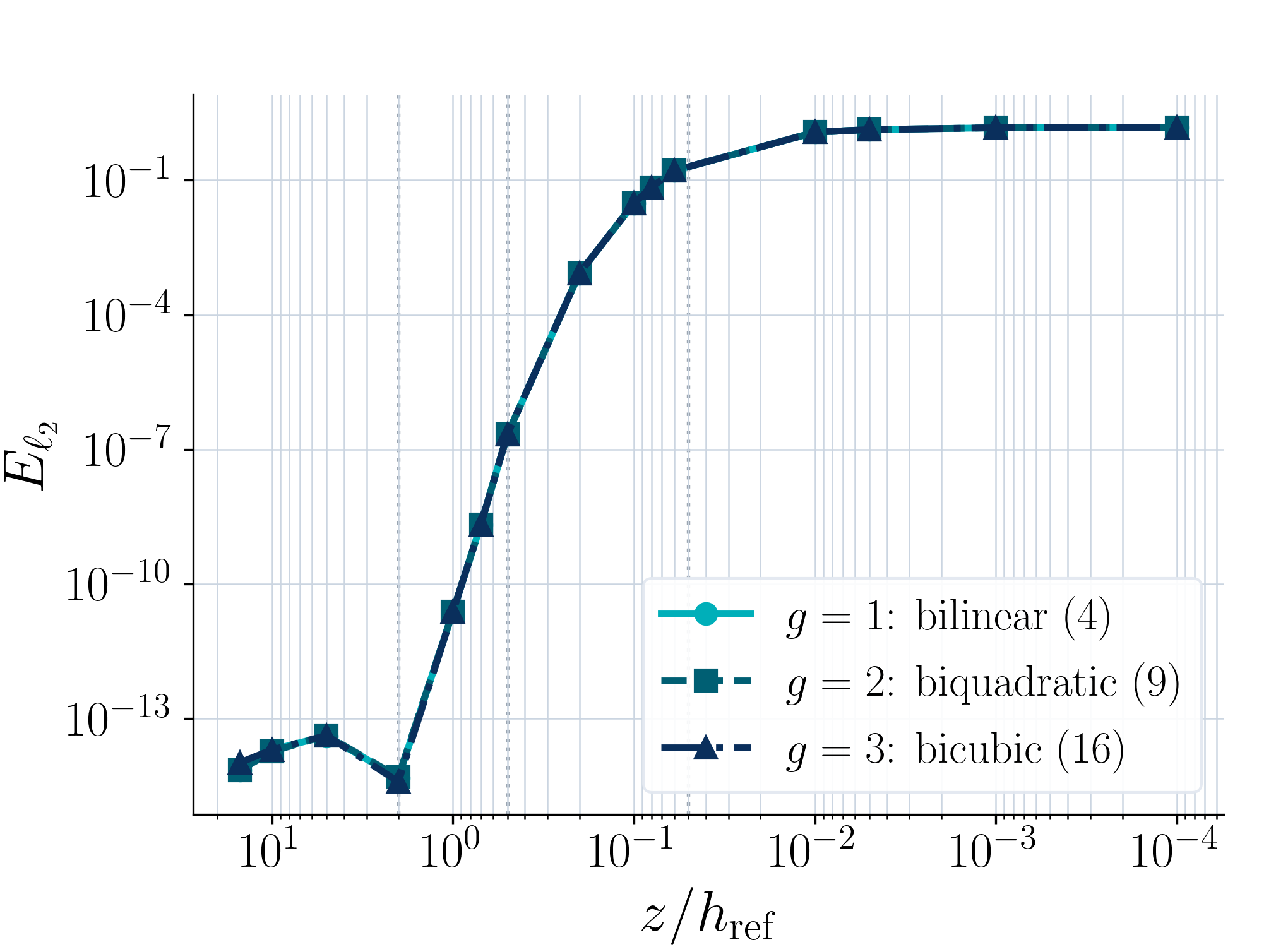}
  \caption{Surface-representation-order independence for the flat sheet:
relative error versus $z/h_{\rm ref}$ along the off-axis line for
$g=1,2,3$ at $N=128$ and $n_{\rm gl}=8$; all curves coincide.
\textcolor{red}{The number in parentheses in the legend is the count of
Lagrange nodes carried by each panel, $(g+1)^2$ --- four for the bilinear map,
nine for the biquadratic and sixteen for the bicubic. It measures the cost of
the geometry representation, which on this flat surface buys no accuracy.}}
  \label{fig:ps_geom}
\end{figure}

\subsection{Planar sheet: panel-driven QBX and the orphan fallback}
\label{sec:res_sheet_qbx}

We first evaluate the panel-driven QBX formulation of
Section~\ref{sec:qbx} on the planar sheet using a deliberately coarse
discretization to stress the near-singular evaluation. The surface is
discretized with $16\times16$ bilinear panels ($g=1$), giving
$h=0.625$, and the QBX parameters are
$r=h/2=0.3125$, $P=12$, $n_{\rm gl}=10$, with expansion centers placed on
the $+z$ side. Near-panel refinement is activated using
$\epsilon_{\rm tgt}=0.05$ and a $4\times4$ subdivision below the threshold
$1.5$ (Section~\ref{sec:accel_lc}). Because the panel count is even, the
origin coincides with a four-panel vertex, producing an orphan region. This
configuration is used to assess the fallback strategy described in
Section~\ref{sec:hybrid}.

Three evaluations are performed: a $241$-point line scan along $y=0$ with
$x\in[-1.5h,1.5h]$ at
$z_0\in\{0.5,1,1.5\}r$; a $161\times161$ coverage map at $z_0=r$; and an
on-axis scan with $z\in[0.02h,2.5h]$.

Because the panel-centred expansion balls are displaced by $r$ from the
surface, their union does not cover the regions above panel vertices. Targets
in these uncovered regions are classified as orphans
(Fig.~\ref{fig:qbx_handoff}). Along the line scan at $z_0=r$, $68\%$ of the
targets are orphans. Figure~\ref{fig:qbx_orphan} shows the QBX accuracy with
the orphan fallback enabled. The orphan targets achieve a lower error (median
$4.7\times10^{-15}$) than the covered targets (median
$1.6\times10^{-11}$), because the fallback places the expansion center directly
below the target, eliminating the horizontal center offset present for
panel-centered expansions. The resulting convergence ratios satisfy
$q\in[1.000,1.044]$, with the remaining error dominated by the
$q^{P+1}$ truncation tail. In this configuration, target-dependent
center placement matters more than geometric coverage. Compared
with direct Gauss--Legendre quadrature on the same scan
($2\times10^{-10}$--$2\times10^{-8}$), QBX improves the accuracy by three to
six orders of magnitude.

\begin{figure}[t]
  \centering
  \includegraphics[width=0.525\linewidth]{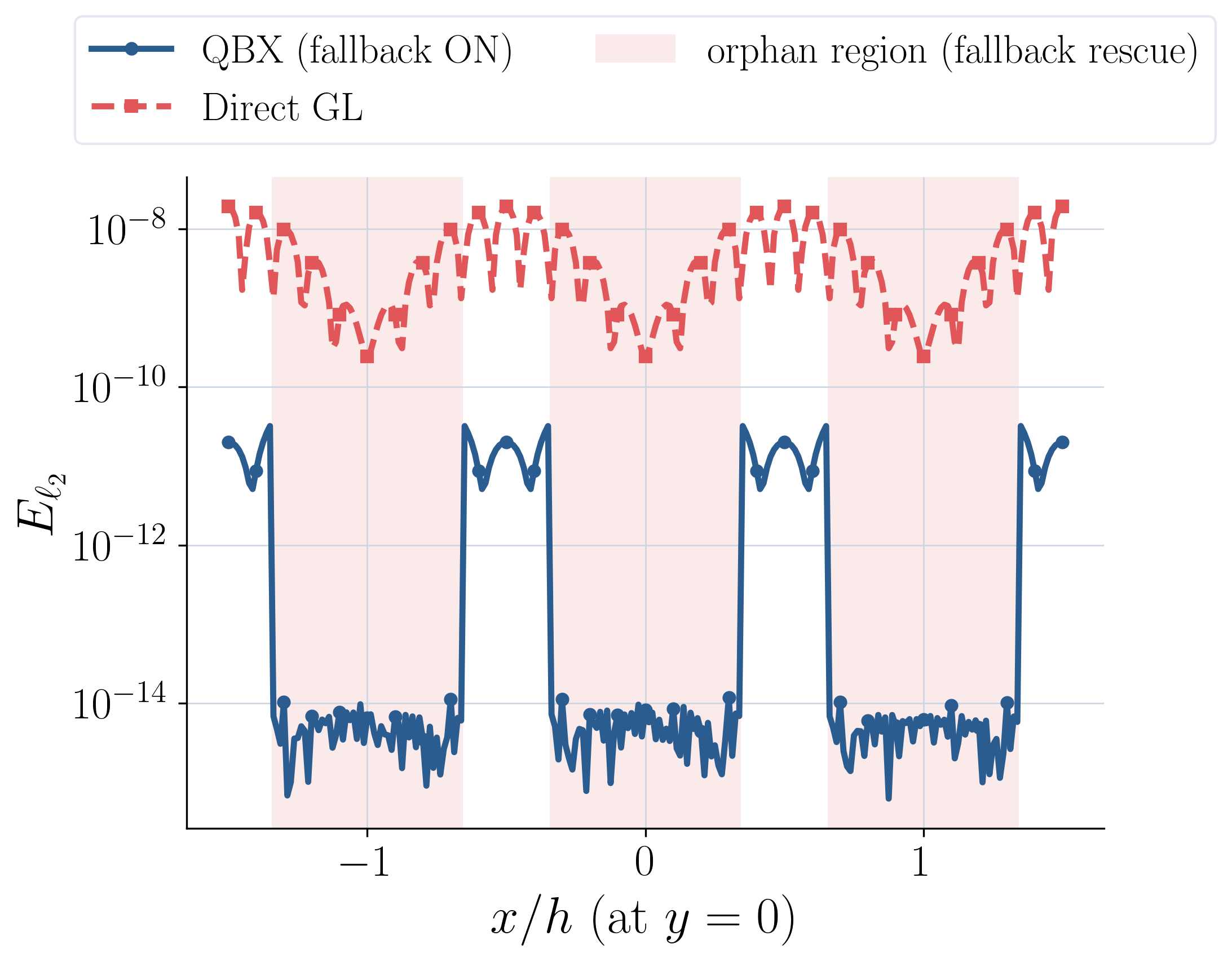}
  \hfill
  \includegraphics[width=0.465\linewidth]{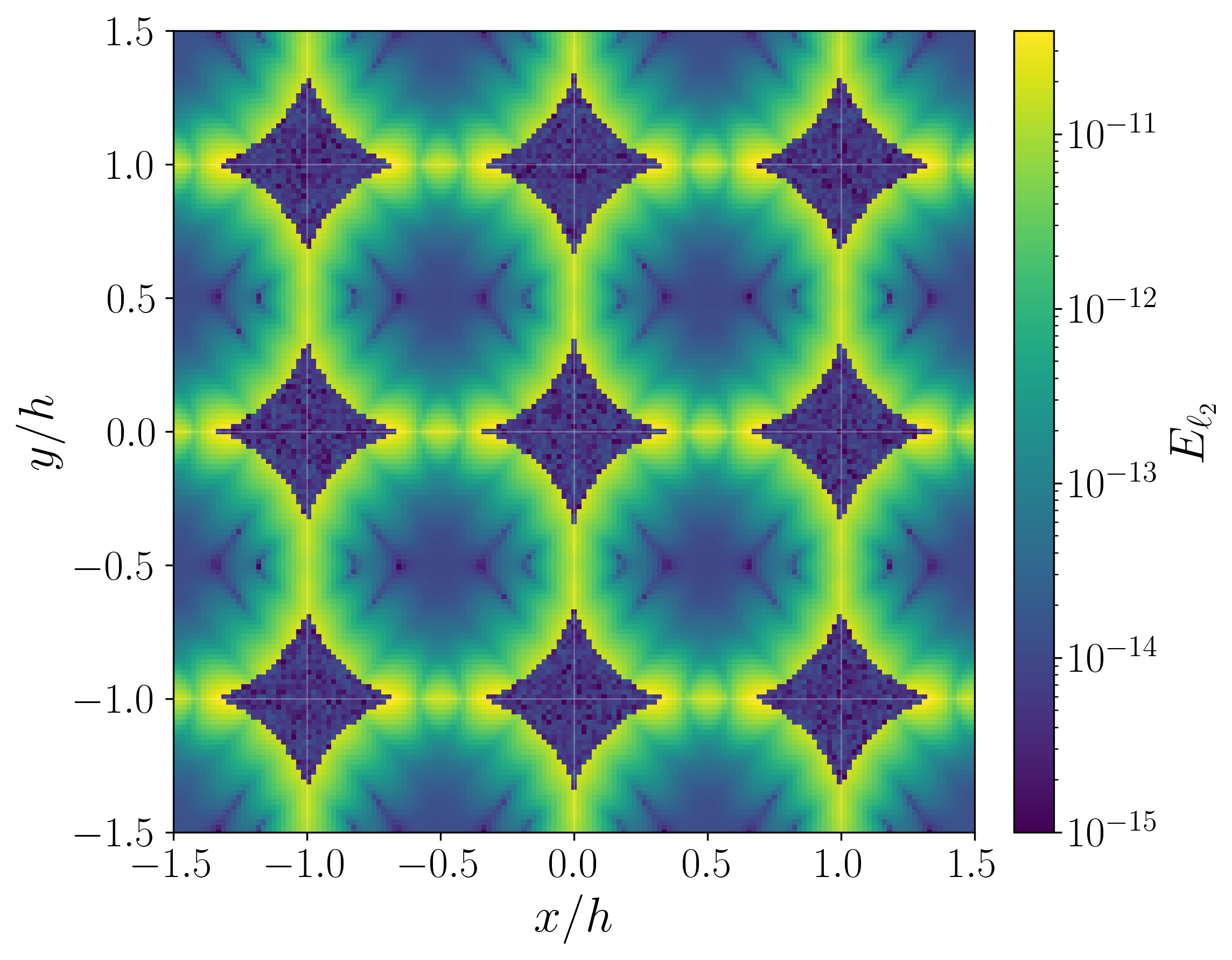}
  \caption{Accuracy with the orphan fallback enabled
($16\times16$ bilinear panels, $P=12$, $z_0=r$). \emph{Left:} relative error
along the line scan $y=0$ comparing QBX and direct Gauss--Legendre quadrature;
shaded regions indicate orphan intervals. \emph{Right:} relative error map over
the $161\times161$ target grid, shown on a logarithmic scale.}
  \label{fig:qbx_orphan}
\end{figure}

Figure~\ref{fig:qbx_handoff} illustrates the coverage geometry and the hybrid
handoff strategy. The left panel shows the coverage at $z_0=r$: covered
targets are contained within panel-centred expansion balls, whereas orphan
regions appear above panel vertices. The right panel reports the error along
the central orphan ray $(0,0,z)$. In the near field, QBX provides a substantial
accuracy improvement: at $z/h=0.02$, direct GL has an error of
$1.1\times10^{-1}$, while the rescued QBX evaluation achieves
$6.7\times10^{-12}$. The GL error decreases with increasing distance and
becomes smaller than the QBX error near $z/h\approx0.93$, whereas QBX remains
accurate within its expansion region. This crossover range
($z/h\approx0.9$--$1.0$) supports the choice of the dispatcher threshold
$\tau=1$ in \eqref{eq:dispatch}, assigning QBX to the near field and GL to
the far field.

\begin{figure}[t]
  \centering
  \includegraphics[width=0.465\linewidth]{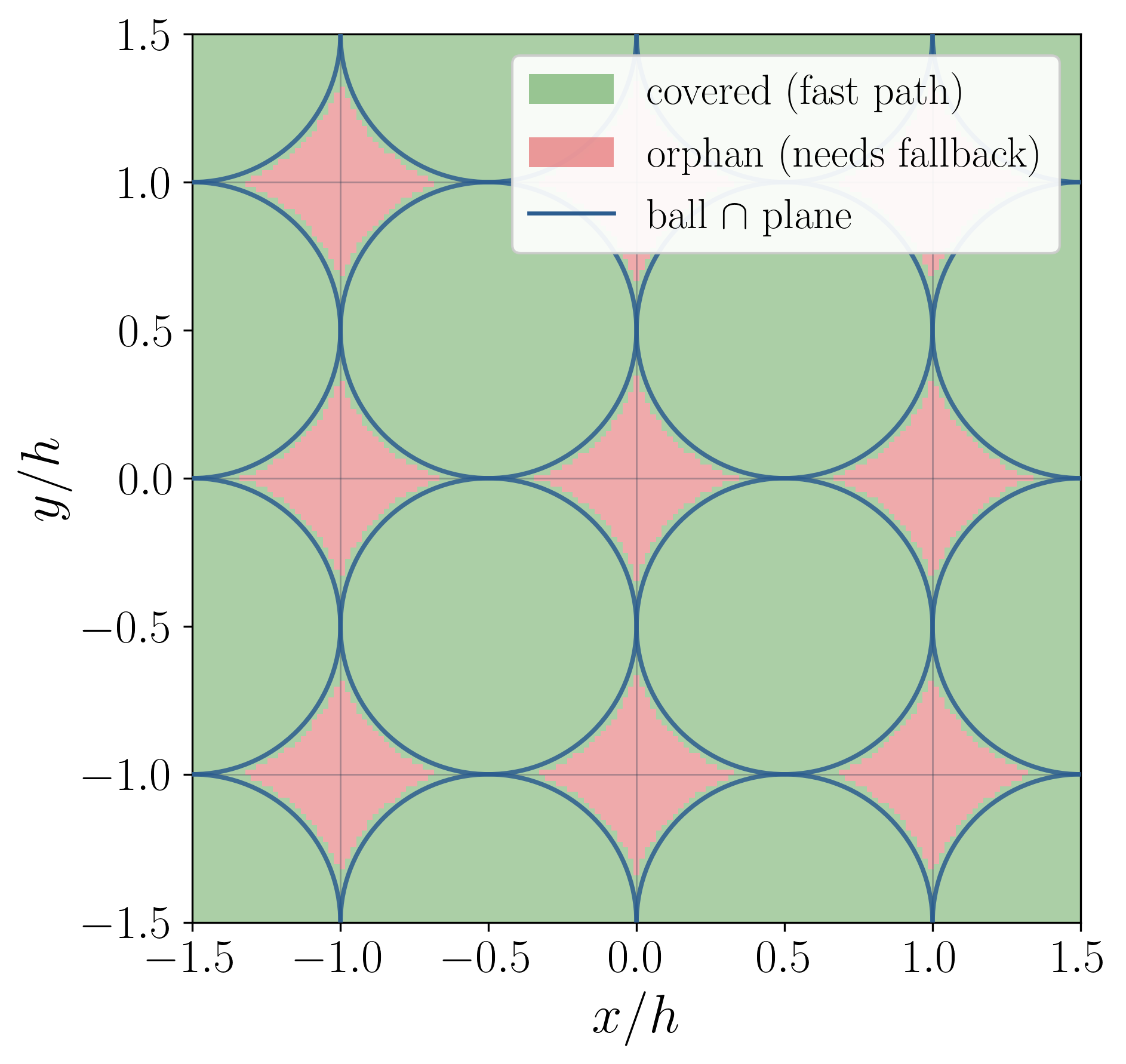}
  \hfill
  \includegraphics[width=0.525\linewidth]{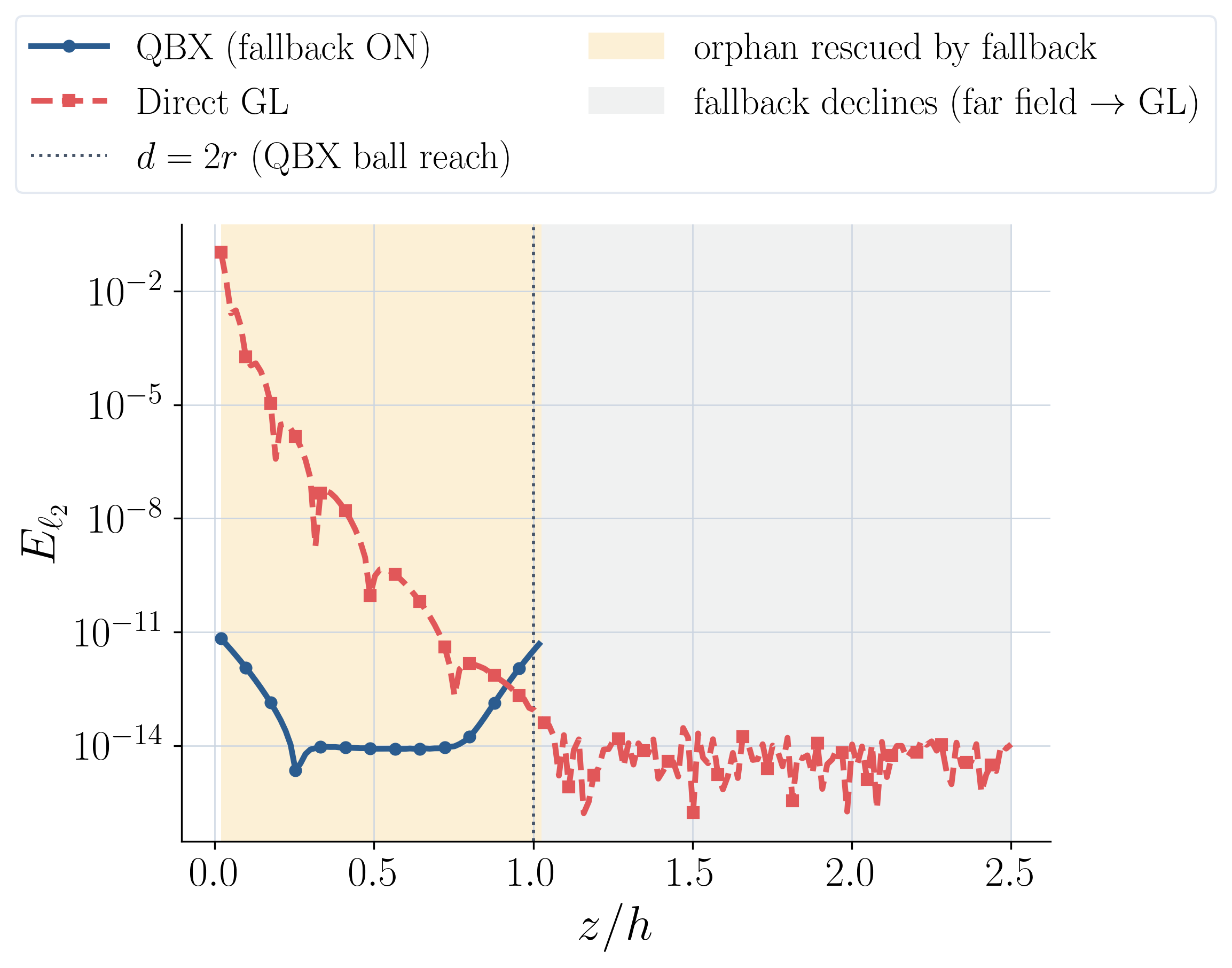}
  \caption{QBX coverage and orphan-target fallback. \emph{Left:} coverage of
the plane $z_0=r$ showing covered and orphan regions. \emph{Right:} relative
error along the orphan ray $(0,0,z)$ for rescued QBX and direct Gauss--Legendre
quadrature; the dotted line indicates the QBX reach $z=2r$.}
  \label{fig:qbx_handoff}
\end{figure}

\subsection{Rotating-ellipse virtual casing: a curved surface with an exact
reference field}
\label{sec:res_vcp_dipole}

We next assess the method on a curved virtual-casing surface using an exact
reference solution. The geometry is the rotating ellipse
\eqref{eq:ellipse_map} with $R_0=3$, $a=0.6$, $b=0.3$, and $\kappa=3/2$,
enclosing a unit point dipole $\mathbf{m}=\hat{\mathbf{z}}$ located at
$(3,0,0)$ (Fig.~\ref{fig:vcp_geom}). The analytical dipole field
\eqref{eq:dipole_field} is converted, through the virtual-casing principle
(Section~\ref{sec:problem}), into the Stratton--Chu representation
\eqref{eq:traces_app}. The extinction identity \eqref{eq:extinction} ensures
that the representation recovers the exact field in the limit of exact
surface integration, providing a reference for measuring numerical errors.
Because $S$ is not a flux surface, both terms in \eqref{eq:VCP} contribute.
Compared with the planar sheet, this configuration introduces surface
curvature and the associated geometric representation error
$O(h^{g+1})$ from \eqref{eq:errdec}, while preserving exact analytical
surface data. It therefore provides a controlled test of virtual-casing
evaluations in the near-surface regime, where target distances can be small
relative to the panel size.

The reference discretization uses $N\times N$ panels in the
$(\theta,\zeta)$ parameter space with $N=128$, resulting in $16{,}384$
cubic panels ($g=3$). The panel Gauss--Legendre rule uses
$n_{\rm gl}=10$, corresponding to $1{,}638{,}400$ source nodes. The hybrid
QBX evaluation employs expansion order $P=8$ with dispatcher threshold
$\tau=1$ in \eqref{eq:dispatch}. For comparison, the midpoint rule uses the
same $128\times128$ sampling grid with point sources located directly on the
exact surface. Owing to the doubly periodic geometry, this corresponds to the
periodic trapezoidal rule and retains spectral far-field accuracy without
introducing geometric representation errors. Targets are sampled along the
outward normal from the outboard anchor $(3.6,0,0)$ using $64$ logarithmically
spaced distances in the range
$d/h_{\rm panel}\in[10^{-4},10^{2}]$, where
$h_{\rm panel}=0.177$ denotes the local panel size at the anchor (see Fig.~\ref{fig:vcp_geom}).

Figure~\ref{fig:vcp_dipole} (left) compares the three methods. In the far
field, the dispatcher selects GL, and the hybrid QBX result coincides with the
GL evaluation for $d/h_{\rm panel}\ge1$. The combined method reaches a
distance-independent error floor of $8.5\times10^{-8}$, while the midpoint
rule converges to approximately $10^{-8}$. As the target approaches the
surface, the smooth quadratures lose accuracy: the GL error increases to
$1.2\times10^{-4}$ at $d/h_{\rm panel}=0.09$ and
$2.8\times10^{-2}$ at $d/h_{\rm panel}=0.01$. At
$d/h_{\rm panel}=10^{-4}$, midpoint and GL errors reach $0.58$ and $0.49$,
respectively, whereas QBX remains stable with an error of
$5.0\times10^{-5}$. Therefore, in the closest-to-surface regime, QBX improves
the accuracy by approximately four orders of magnitude.

The near-surface QBX error is limited by the geometric floor of
\eqref{eq:errdec} rather than by QBX truncation. This conclusion is supported
by three independent parameter studies along the same ray. Varying the
geometric order at fixed $N=128$ (Fig.~\ref{fig:vcp_dipole}, right) reduces
the near-surface error from $2.1\times10^{-3}$ for $g=1$ to
$4.7\times10^{-4}$ for $g=2$ and $5.0\times10^{-5}$ for $g=3$, with the
far-field error floors following the same trend. At fixed $g=3$, mesh
refinement further reduces the near-surface error from
$1.8\times10^{-2}$ to $3.7\times10^{-6}$ as $N$ increases from $32$ to
$256$. In contrast, increasing the QBX expansion order
$P\in\{6,8,10,12\}$ changes the error only marginally, from
$5.4\times10^{-5}$ to $4.3\times10^{-5}$, and does not affect the GL-dominated
region. Therefore, the $q^{P+1}$ truncation error is already below the
geometric discretization floor at $P=6$, and further near-surface improvement
requires increased geometric resolution rather than higher QBX order.

Finally, the accuracy is insensitive to the choice of ray anchor
(Fig.~\ref{fig:vcp_anchors}). Repeating the sweep from the inboard
$A_1=(2.4,0,0)$ and top $A_2=(3,0,0.6)$ anchors yields near-surface QBX
errors of $3.1\times10^{-5}$ and $1.4\times10^{-4}$, respectively, comparable
to the outboard value of $5.0\times10^{-5}$ despite variations in local
curvature and panel size. In contrast, the GL error remains approximately
$0.49$ for all three rays. The far-field increase observed on the inboard ray
is geometric: the ray passes through the central opening and intersects the
opposite surface, where the exact dipole field vanishes and the relative error
becomes ill-conditioned.

\begin{figure}[t]
  \centering
  \includegraphics[width=1.0\linewidth]{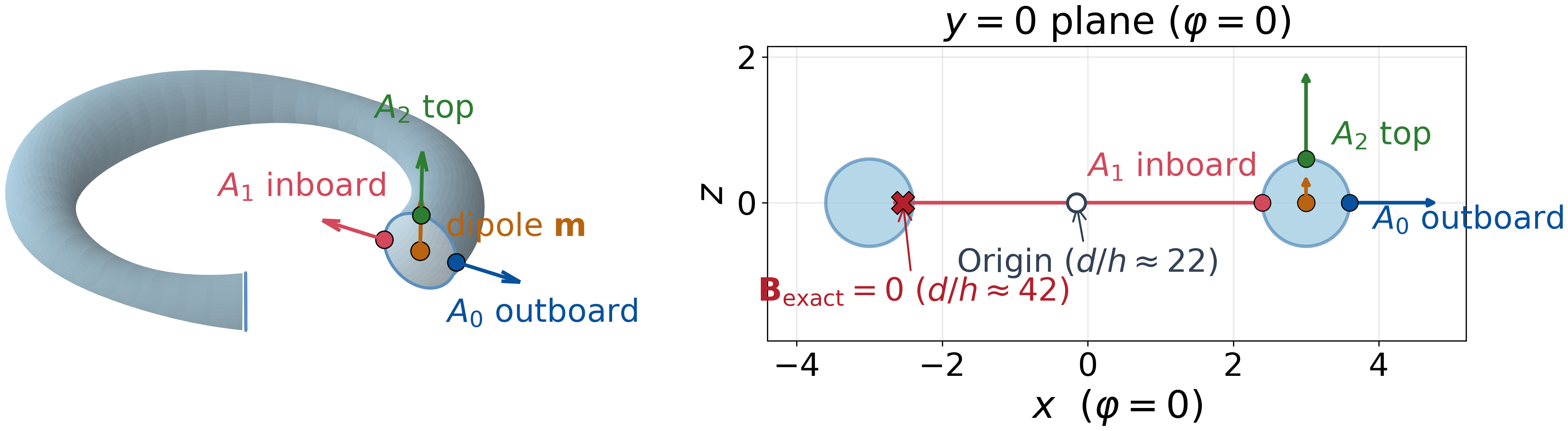}
\caption{\textcolor{red}{Geometry of the rotating-ellipse virtual-casing test.
\emph{Left:} Source surface together with the dipole $\mathbf{m}=(0,0,1)$ located at
$(R_0,0,0)$ and the three sweep test points $A_0$ (outboard), $A_1$ (inboard),
and $A_2$ (top), all lying in the $\varphi=0$ cross-section. Arrows indicate
the outward directions along which target points are swept.
\emph{Right:} The $y=0$ plane, which intersects the surface at
$\varphi=0$ (right lobe) and $\varphi=\pi$ (left lobe). Targets are swept from each
anchor outward to $d/h_{\rm panel}=100$.}}
  \label{fig:vcp_geom}
\end{figure}

\begin{figure}[t]
  \centering
  \includegraphics[width=0.495\linewidth]{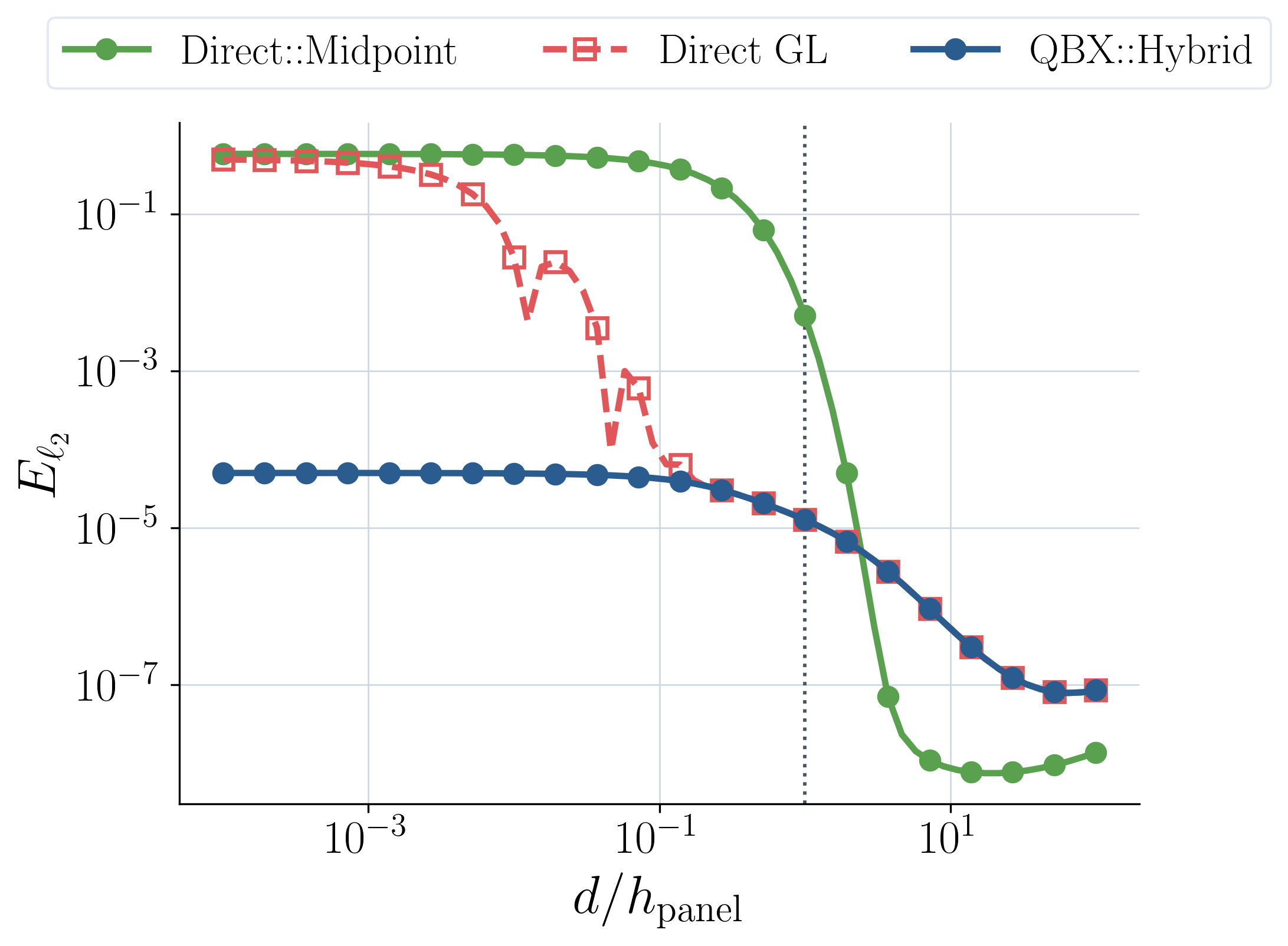}
  \hfill
  \includegraphics[width=0.495\linewidth]{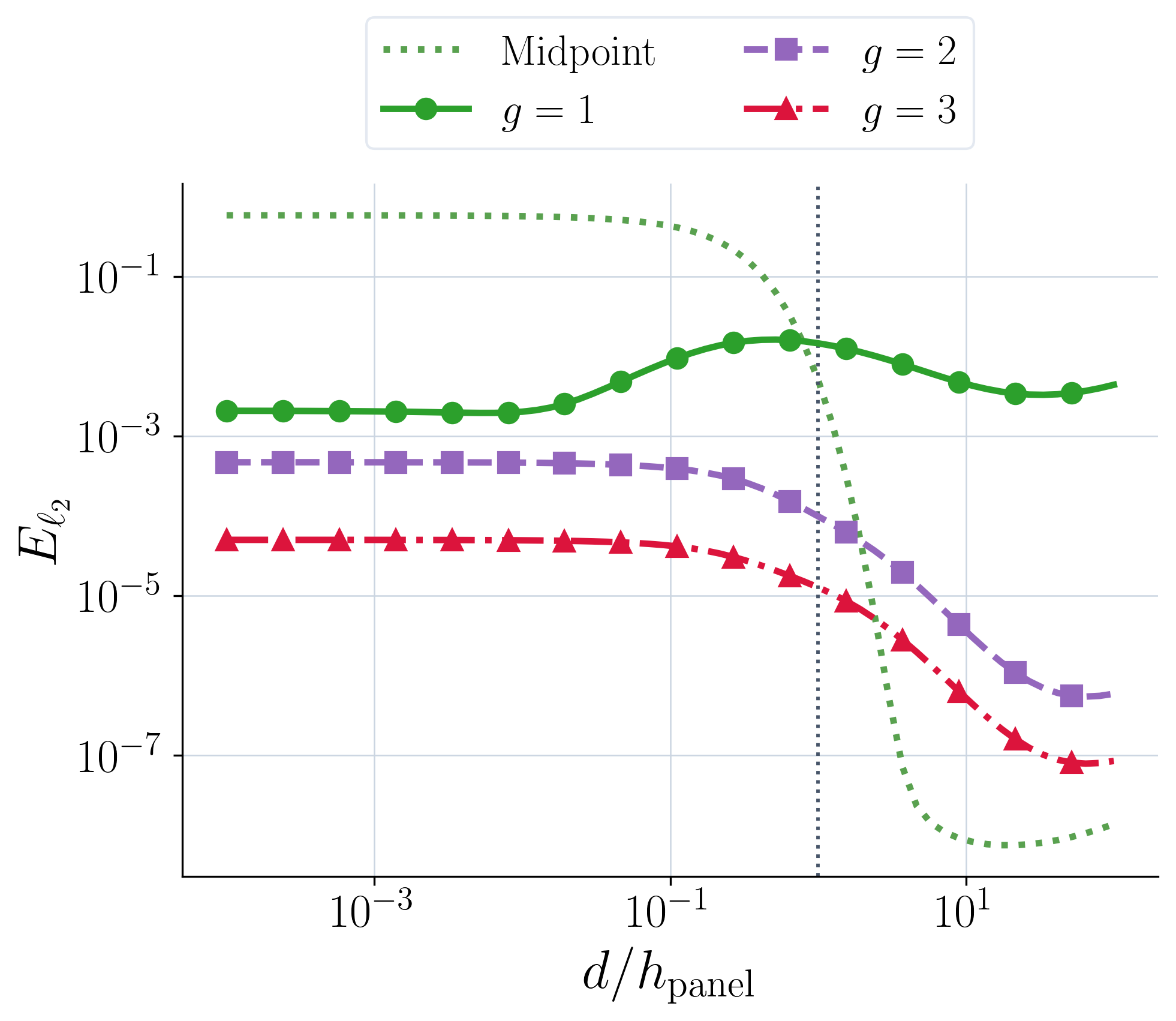}
\caption{Virtual casing on the rotating ellipse (outboard ray, $N=128$,
$n_{\rm gl}=10$). \emph{Left:} relative error $E_{\ell_2}$ versus
$d/h_{\rm panel}$ for $g=3$, $P=8$; the dotted line indicates the dispatcher
threshold $\tau=1$. \emph{Right:} hybrid QBX error for $g=1,2,3$ compared with
midpoint evaluation on the exact surface; the near-surface plateau is set by
the geometric discretization error $O(h^{g+1})$.}
  \label{fig:vcp_dipole}
\end{figure}

\begin{figure}[t]
  \centering
  \includegraphics[width=0.62\linewidth]{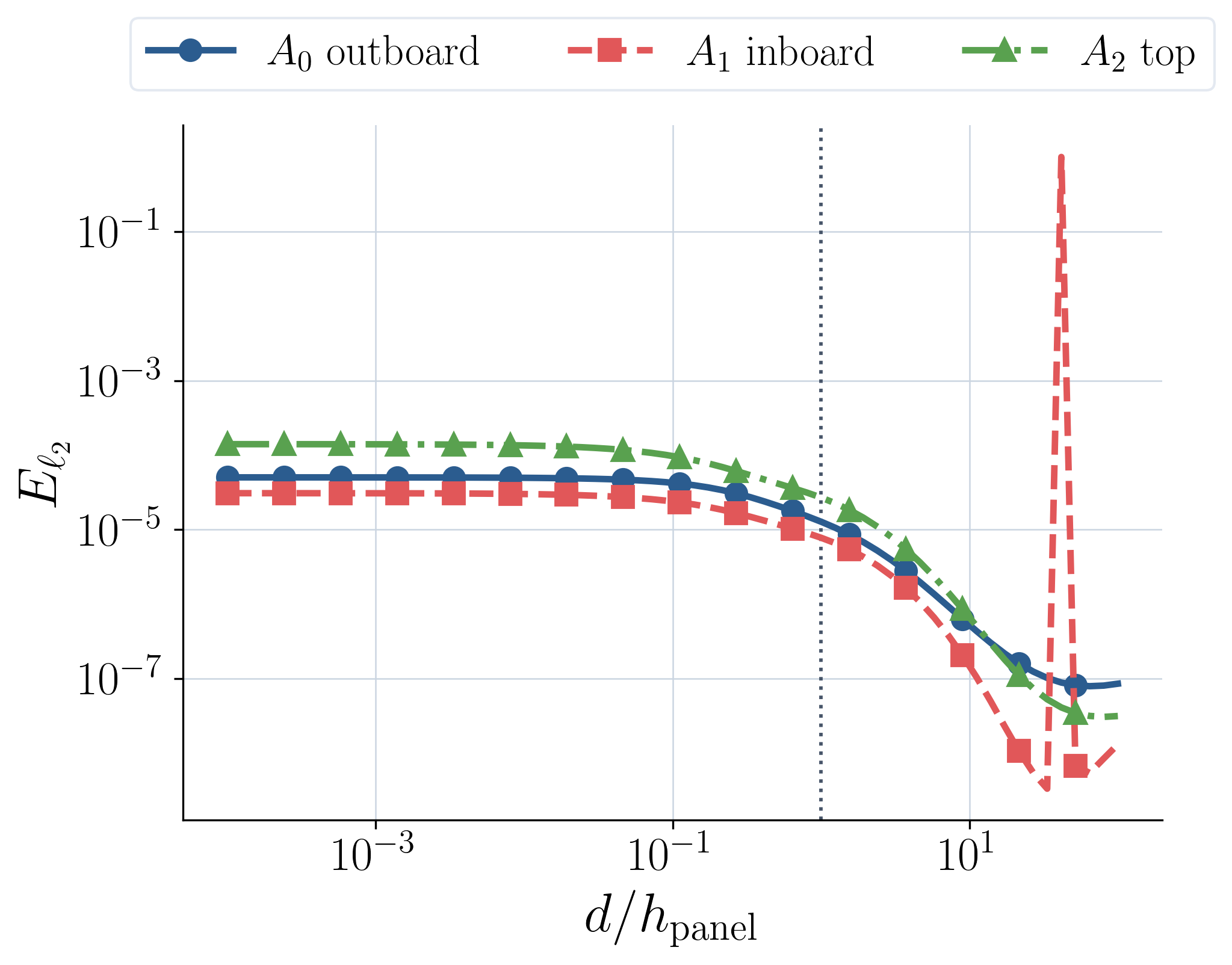}
\caption{QBX--Hybrid error along the three anchor rays of
Fig.~\ref{fig:vcp_geom} ($N=128$, $g=3$, $P=8$). Near-surface errors are
comparable for all anchors; the inboard ray exhibits a far-field spike at the
wall intersection where $\mathbf B_{\rm exact}=0$.}
  \label{fig:vcp_anchors}
\end{figure}

\subsection{QBX versus upsampled trapezoidal quadrature}
\label{sec:res_malhotra}

We adopt the benchmark introduced by Malhotra et al.~\cite{malhotra2019},
based on the rotating-ellipse geometry \eqref{eq:ellipse_map} with
$R_0=2$, $a=0.7$, $b=1.0$, and $\kappa=3/2$, together with the analytical
field generated by a circular current loop of radius $5$
\eqref{eq:loop_field}. The loop field is converted into Stratton--Chu
surface traces and reconstructed inside $\Omega$. By the extinction identity
\eqref{eq:extinction}, the resulting virtual-casing field satisfies
$\mathbf B_{\rm VCP}=\mathbf B_0$, corresponding to the vacuum
($\lambda=0$) configuration considered in~\cite{malhotra2019}.

We compare two off-surface evaluation methods using identical surface data:
the hybrid QBX method with $64\times64$ cubic panels ($g=3$),
$n_{\rm gl}=10$, and $P=8$, and the global periodic trapezoidal quadrature of
~\cite{malhotra2019} on an upsampled uniform grid. For the latter, the surface
densities are analytically resampled to remove interpolation errors. The
comparison therefore isolates the off-surface quadrature accuracy rather than
errors associated with the boundary-integral discretization.

In the far field, the periodic trapezoidal rule provides higher accuracy
(Fig.~\ref{fig:mal_conv}, left): for an eight-point interior target cloud, its
error decreases spectrally from $6.3\times10^{-5}$ to
$1.1\times10^{-14}$ as the grid is refined from $64^2$ to $256^2$ nodes. In
contrast, the panel-based evaluation is limited by the $O(h^4)$ geometric
discretization floor. The behavior changes in the near-surface regime
(Fig.~\ref{fig:mal_headtohead}, left). At equal surface sampling
($65{,}536$ nodes), the two methods intersect at approximately
$d/h=0.3$: the trapezoidal error increases from $1.8\times10^{-7}$ at
$d/h=1$ to $0.50$ at \textcolor{red}{$d/h=10^{-2}$}, whereas QBX remains in the range
$2$--$5\times10^{-6}$. The cost comparison in
Fig.~\ref{fig:mal_headtohead} (right) reflects the different scaling of the
two approaches. Maintaining trapezoidal accuracy requires increasing the
surface resolution according to $N_s\sim(h/d)^2$
~\cite{malhotra2019,malhotra2020}, while QBX maintains a fixed
$65{,}536$-node discretization independent of the target distance. Even with
adaptive upsampling limited to $1024^2$ nodes, the trapezoidal approach
eventually reaches $O(1)$ error.

The remaining QBX truncation error follows the predicted convergence behavior
(Fig.~\ref{fig:mal_conv}, right). Self-convergence with respect to a high
reference order ($P_{\rm ref}=18$) isolates the truncation contribution by
removing the common geometric and quadrature errors. At $d/h=0.5$, the error
decreases as $q^{P+1}$ with fitted $q\approx0.46$, reaching
$3.5\times10^{-14}$ at $P=16$. At the target distance
$d/h=0.1$ used in the main simulations, the convergence factor is much smaller
($q\approx0.02$), and convergence is achieved by approximately $P=4$. The
total error is then limited by the geometric discretization floor
($7\times10^{-9}$ on a $256\times256$-panel mesh), consistent with the
$P$-insensitivity observed in Section~\ref{sec:res_vcp_dipole}.

These results demonstrate the complementary behavior of global upsampling and
hybrid QBX for off-surface field evaluation. Global upsampling, as used in
~\cite{malhotra2019,malhotra2020}, provides spectral accuracy in the far field
but requires a surface resolution that increases as $(h/d)^2$ in the
near-surface regime. The hybrid QBX formulation avoids this
distance-dependent refinement by maintaining a fixed discretization while
achieving $2$--$5\times10^{-6}$ accuracy over the tested near-surface range
with $65{,}536$ surface nodes. Therefore, near-surface evaluations can be
performed without the progressive surface upsampling required by global
trapezoidal quadrature.

\begin{figure}[t]
  \centering
  \includegraphics[width=0.495\linewidth]{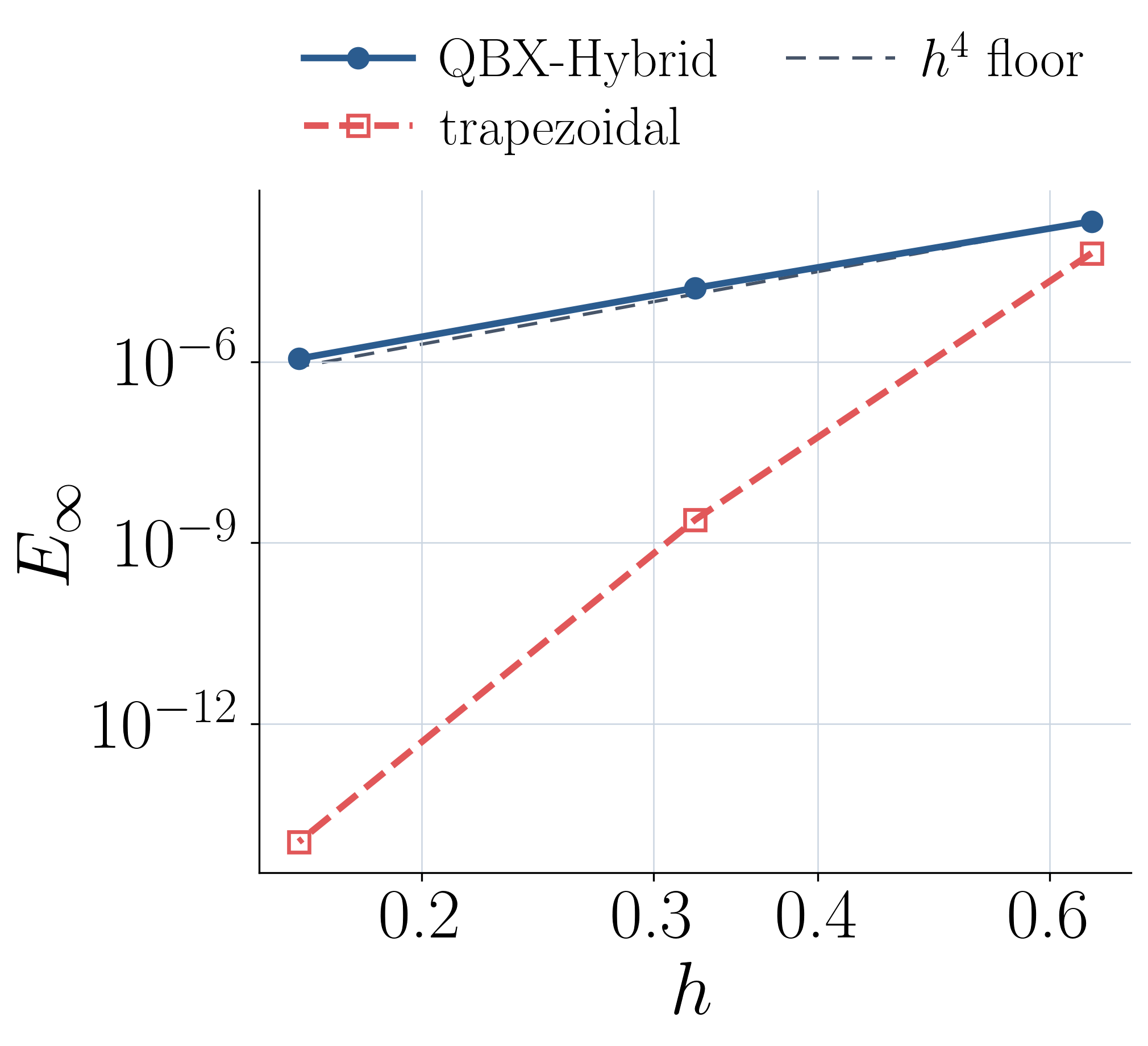}
  \hfill
  \includegraphics[width=0.495\linewidth]{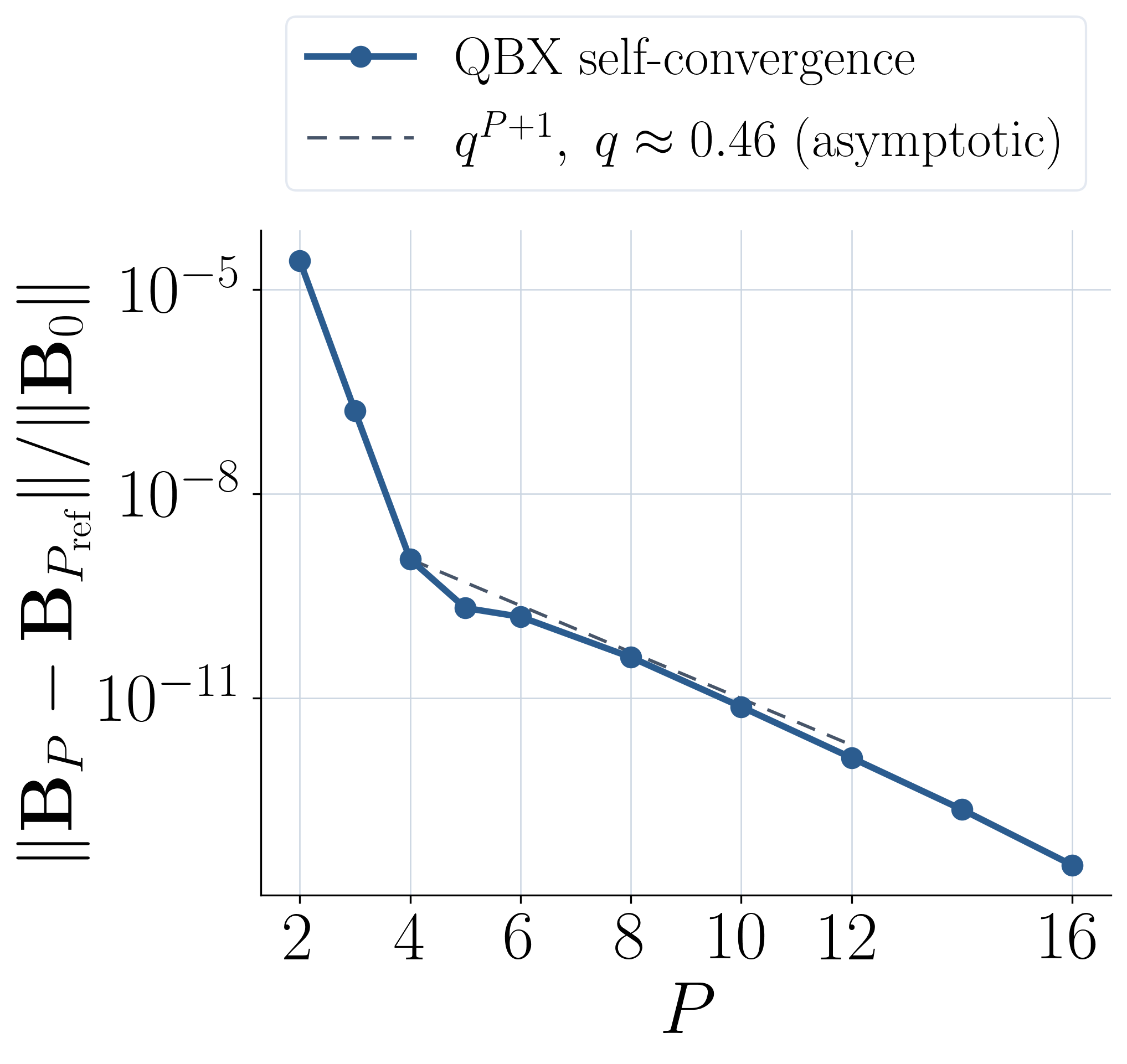}
\caption{Convergence on the Malhotra et al.~\cite{malhotra2020} benchmark.
\emph{Left:} far-field convergence versus panel size $h$; the
\textcolor{red}{upsampled} trapezoidal quadrature is spectral while the panel
method is limited by the \textcolor{red}{$h^4$} geometric floor.
\emph{Right:} QBX truncation convergence at $d/h=0.5$ with
$P_{\rm ref}=18$. \textcolor{red}{Here and in the following figures, curves
labelled \emph{QBX-Hybrid} denote the target-driven dispatcher of
\S\ref{sec:hybrid} (Table~\ref{tab:terminology}).}}
  \label{fig:mal_conv}
\end{figure}

\begin{figure}[t]
  \centering
  \includegraphics[width=0.495\linewidth]{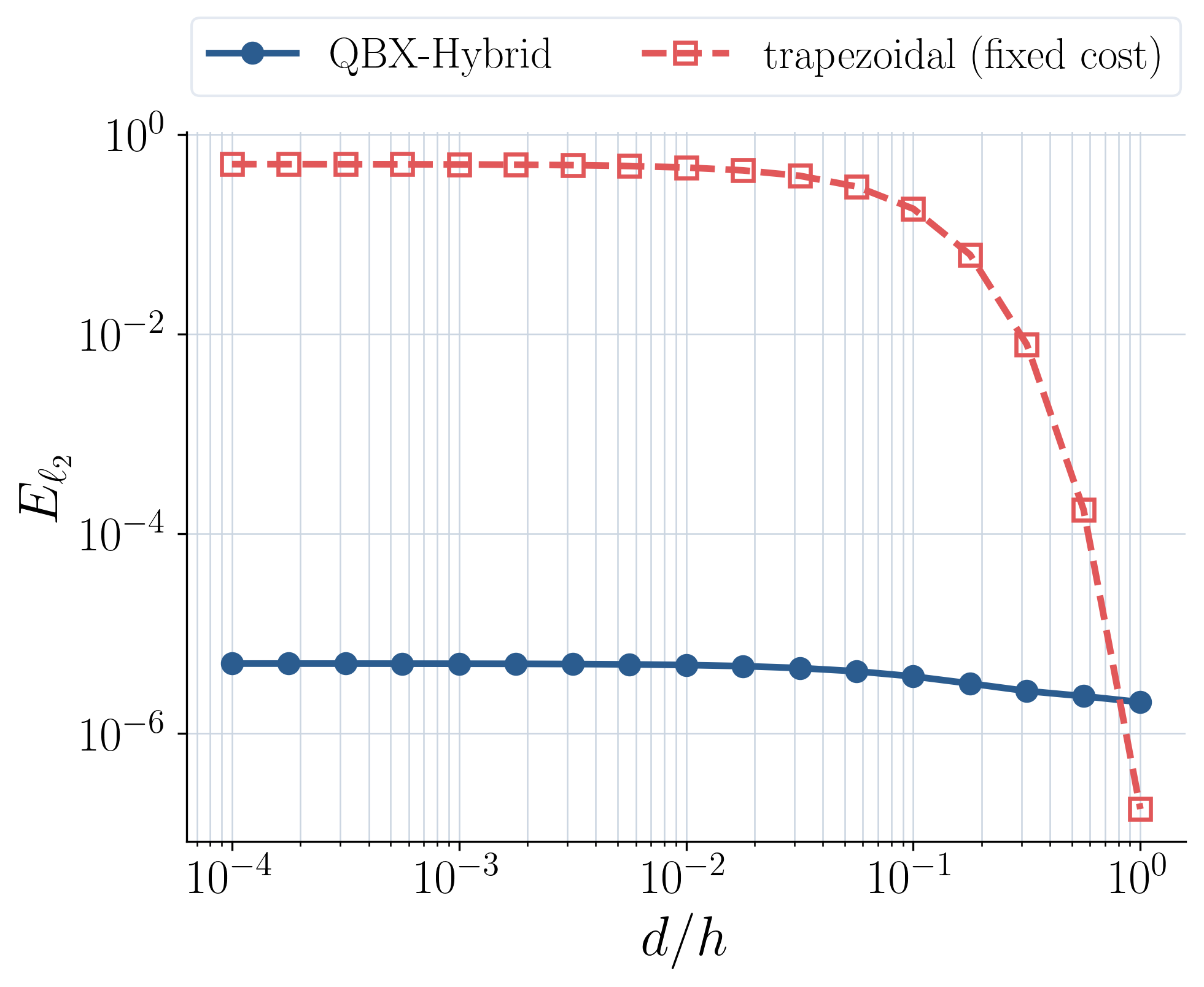}
  \hfill
  \includegraphics[width=0.495\linewidth]{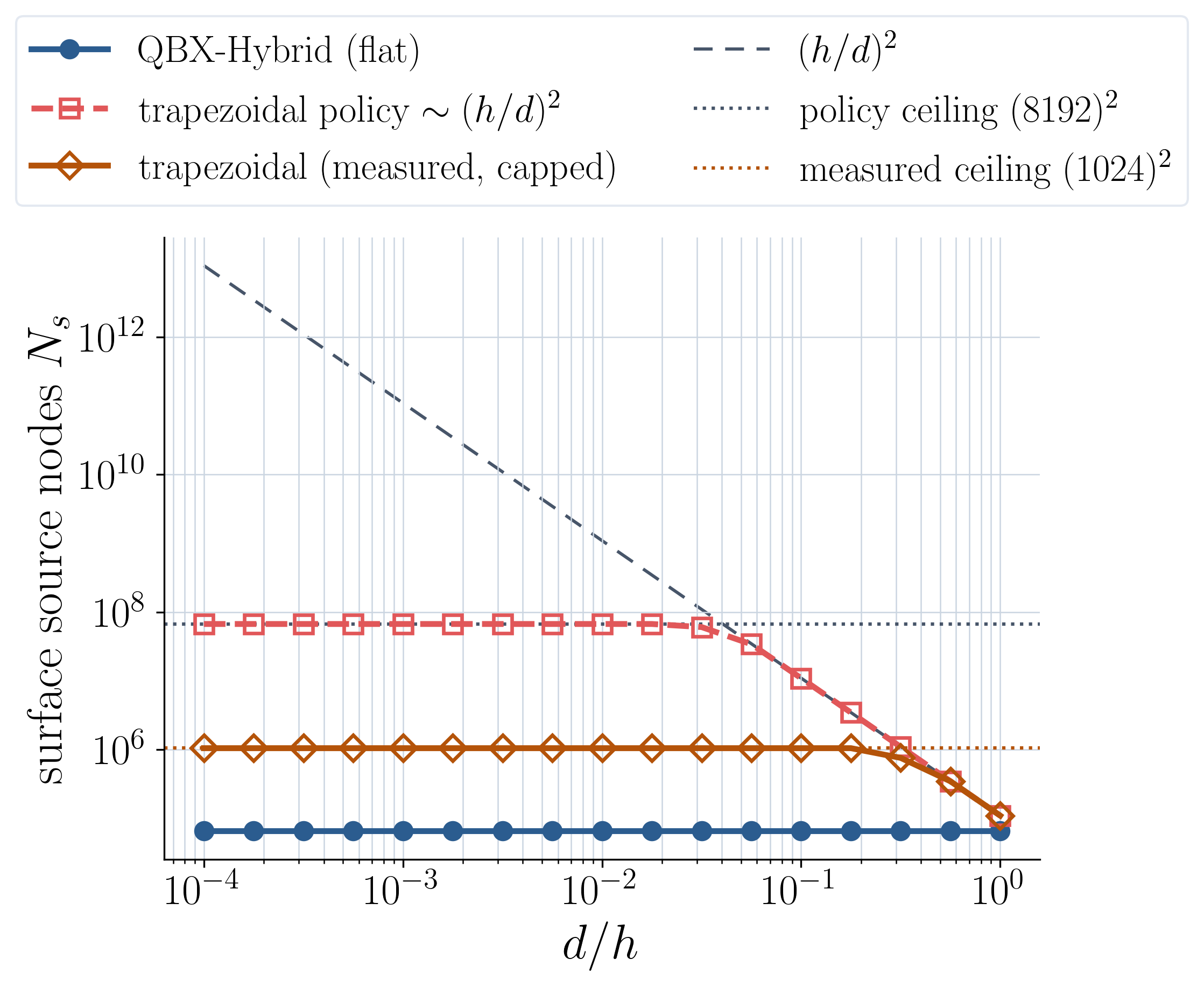}
  \caption{Malhotra benchmark in the near-surface regime
$d/h\in[10^{-4},1]$ (outboard anchor $(2.7,0,0)$, $h=0.161$).
\emph{Left:} relative error for hybrid QBX and upsampled trapezoidal
quadrature at matched sampling ($65{,}536$ nodes). \emph{Right:} surface-node
count versus distance, showing the distance-independent QBX cost.}
  \label{fig:mal_headtohead}
\end{figure}

\subsection{Comparison of evaluation strategies}
\label{sec:res_paths}

We compare now the different strategies available for near-surface field
evaluation: the panel-driven QBX formulation of Section~\ref{sec:qbx}, the
target-driven hybrid dispatcher of Section~\ref{sec:hybrid}, and the two
acceleration mechanisms introduced in Section~\ref{sec:acceleration}: local
correction and target-specific expansion (TSE). We first verify the
equivalence of the two evaluation paths and then demonstrate that the two
accelerations act independently, allowing their performance benefits to be
combined.

The panel-driven and target-driven approaches agree in their common
application region. For the toroidal-coil test (exact exterior field zero),
both methods produce identical errors of $3.1\times10^{-6}$ at shared
near-surface targets. Their computational costs, however, differ
substantially: the panel-driven approach requires an $O(N_p^2)$ preprocessing
stage, whereas the target-driven formulation scales linearly with $N_p$ for a
given target set. Consequently, the measured speed-up increases approximately
as $N_p^{0.95}$, from $4.5\times$ at $961$ panels to $66\times$ at
$16{,}129$ panels (Fig.~\ref{fig:paths_driver}). The target-driven hybrid
formulation is therefore adopted as the default evaluation problems.

\begin{figure}[t]
  \centering
  \textcolor{red}{\includegraphics[width=0.495\linewidth]{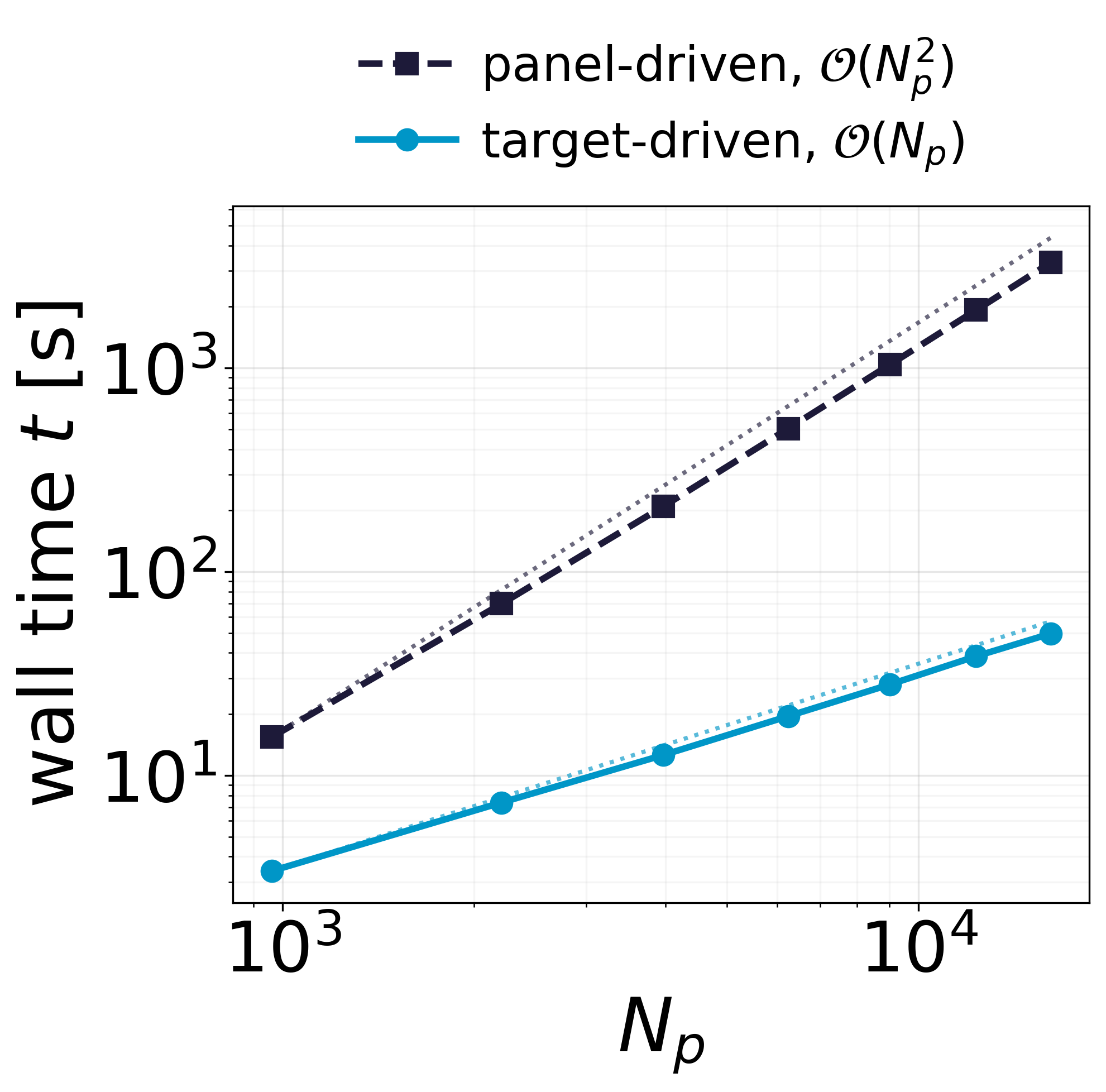}}
  \hfill
  \textcolor{red}{\includegraphics[width=0.495\linewidth]{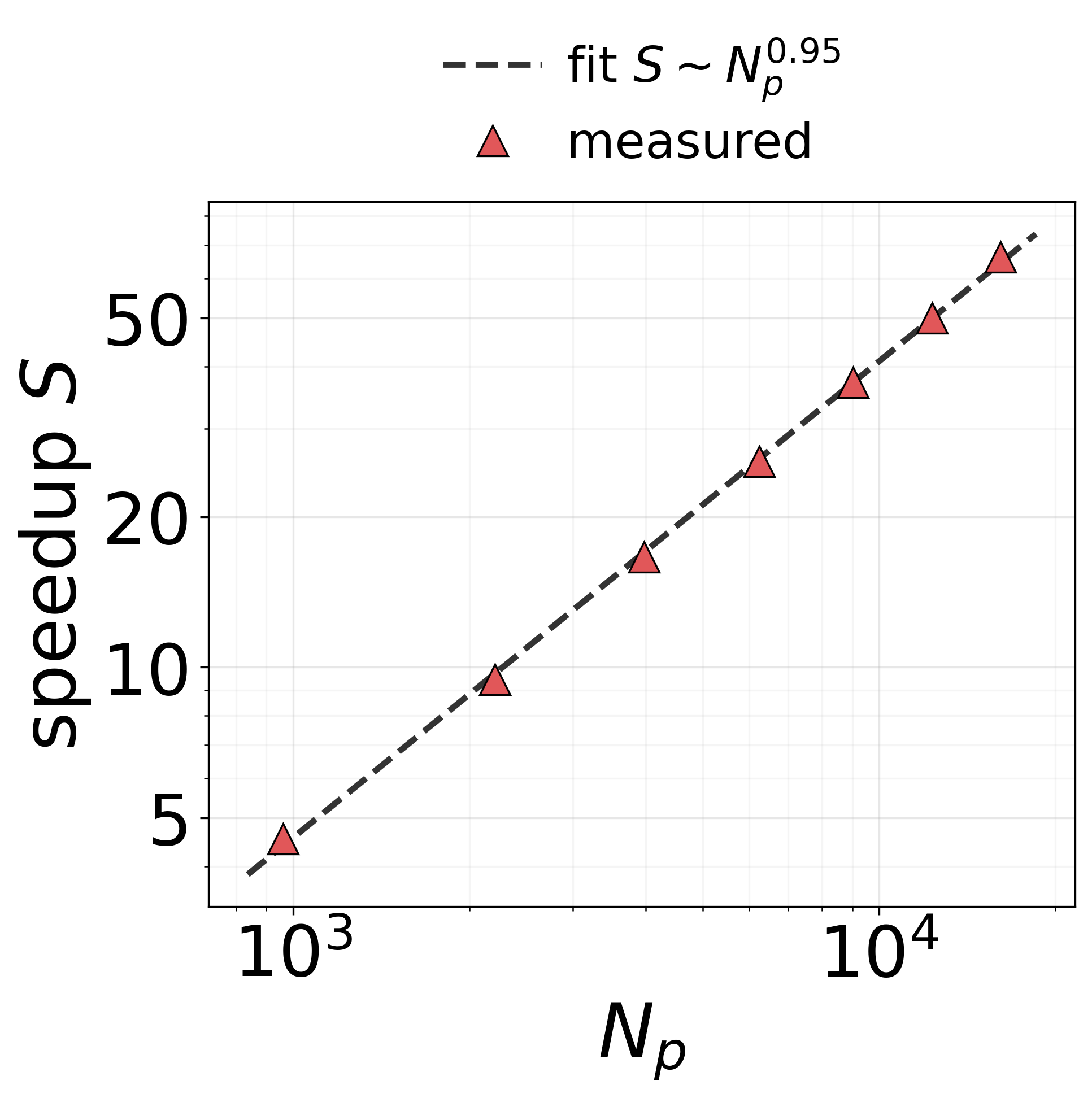}}
  \caption{\textcolor{red}{The two drivers of
Table~\ref{tab:terminology} on the toroidal coil: the \emph{panel-driven}
formulation of \S\ref{sec:qbx} against the \emph{target-driven} dispatcher of
\S\ref{sec:hybrid}.
\emph{Left:} wall time $t$ against the panel count $N_p$; the
dotted guides are the expected sclaings: $\mathcal{O}(N_p^{2})$ and $\mathcal{O}(N_p)$. \emph{Right:} the resulting speed-up
$S=t_{\rm panel}/t_{\rm target}$, which grows from $4.5\times$ at $N_p=961$ to
$66\times$ at $N_p=16{,}129$ and follows the empirical fit
$S\sim N_p^{0.95}$.}}
  \label{fig:paths_driver}
\end{figure}

The two accelerations are verified on the Malhotra et al.~\cite{malhotra2019,malhotra2020}
configuration of Section~\ref{sec:res_malhotra}. The target-specific
contraction is algebraically equivalent to the coefficient-based evaluation:
using identical expansion centers, the maximum relative difference remains
below $8\times10^{-14}$ for all tested fields, and the convergence curves in
Fig.~\ref{fig:paths_equiv} (left) are indistinguishable. In contrast, the
local-correction split introduces a controlled $\alpha$-dependent deviation
from the full-QBX reference. This deviation decreases from
$1.9\times10^{-6}$ at $\alpha=1.5$ to $1.2\times10^{-12}$ at
$\alpha=8$ (Fig.~\ref{fig:paths_cost}, left), with the same trend observed
for both coefficient-based and TSE evaluations.

\begin{table}[t]
\centering
\small
\caption{The $2\times2$ acceleration factorization on the Malhotra et al.~\cite{malhotra2020}
benchmark ($\alpha=3$). Parameters: $48\times48$ cubic panels,
$n_{\rm gl}=8$, $P=8$, and $256$ targets at $d/h=0.3$. Wall-times are averaged
over three OpenMP runs; errors are $E_\infty$ relative to the full-coefficient
baseline.}
\label{tab:paths2x2}
\setlength{\tabcolsep}{6pt}
\begin{tabular}{lrrr}
\hline
Evaluation path & wall-time [ms] & speed-up & deviation \\
\hline
full coefficient (baseline) & $2507$ & $1.0$ & --- \\
full TSE                    & $426$  & $5.9$ & $3.8\times10^{-14}$ \\
local correction, coefficient & $464$ & $5.4$ & $4.5\times10^{-9}$ \\
local correction, TSE       & $368$  & $6.8$ & $4.5\times10^{-9}$ \\
\hline
\end{tabular}
\end{table}

Table~\ref{tab:paths2x2} combines the two accelerations at fixed
$(P,\alpha)$. Each acceleration alone provides a $5$--$6\times$ speed-up,
while their combination reaches $6.8\times$ with an error three orders of
magnitude below the $4\times10^{-6}$ discretization error floor.
Figure~\ref{fig:paths_equiv} (right) illustrates the complementary nature of
the two optimizations: TSE removes the $(P+1)^2$ coefficient evaluation cost,
whereas local correction reduces the number of corrected interactions. The
advantage of combining both mechanisms increases with mesh size
(Fig.~\ref{fig:paths_cost}, right), with split-TSE outperforming full-TSE
beyond approximately $10^3$ panels.

The same accuracy--cost trade-off is observed for the toroidal geometry of
Section~\ref{sec:accel_lc}: $\alpha=3$ provides a $3.6\times$ speed-up with a
deviation of $5.8\times10^{-8}$, while stricter accuracy requirements reduce
the speed-up to $1.6\times$. Since absolute timings depend on the hardware
configuration (Section~\ref{sec:res_gpu}), the reported speed-ups and error
levels represent the transferable performance characteristics.

Overall, these results show that the hybrid QBX formulation retains
near-surface accuracy while the target-specific contraction and local
correction significantly reduce its computational cost. The resulting
evaluation strategy avoids both the near-surface accuracy degradation of
smooth quadrature and the distance-dependent refinement required by
upsampled virtual-casing approaches.

\begin{figure}[t]
  \centering
  \includegraphics[width=0.495\linewidth]{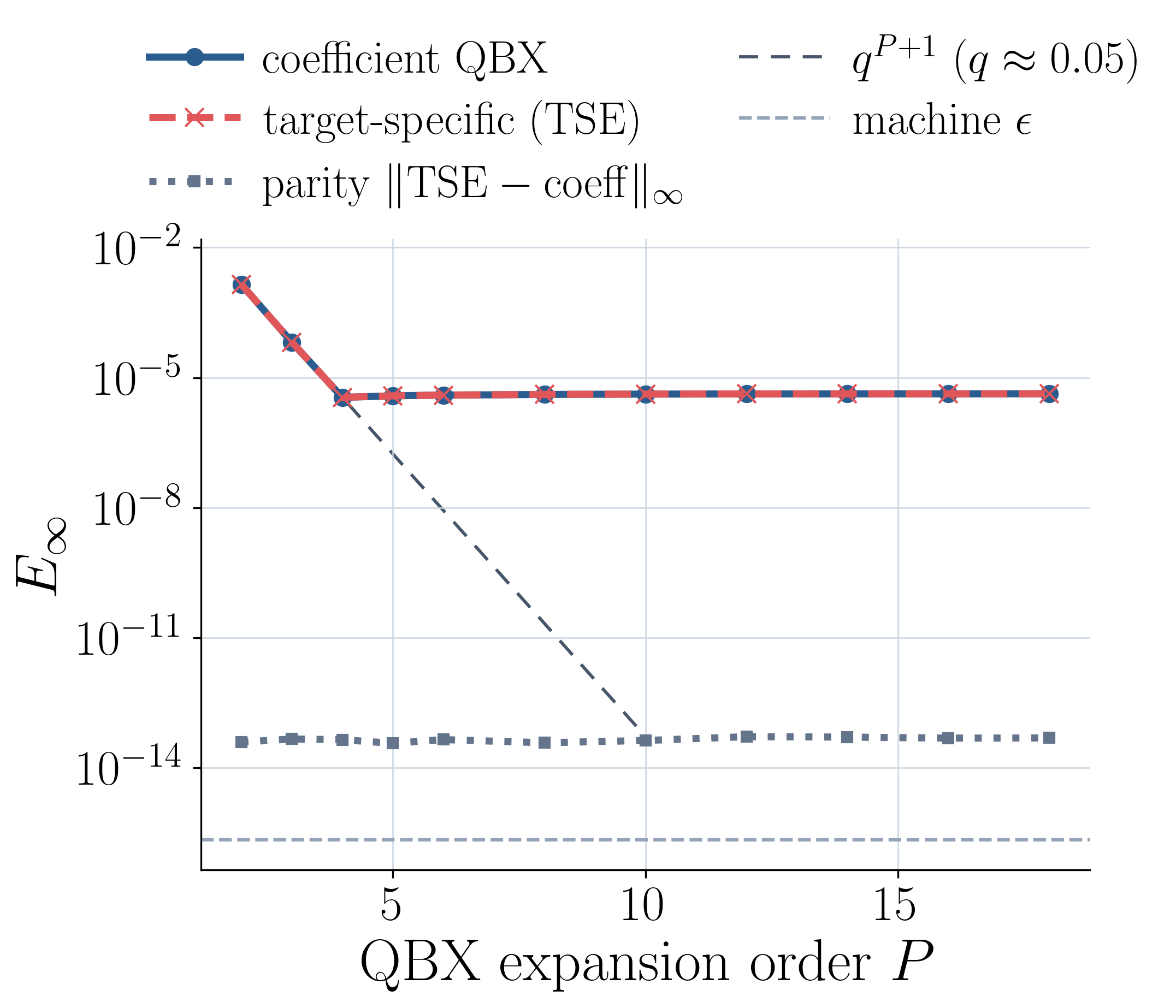}
  \hfill
  \includegraphics[width=0.495\linewidth]{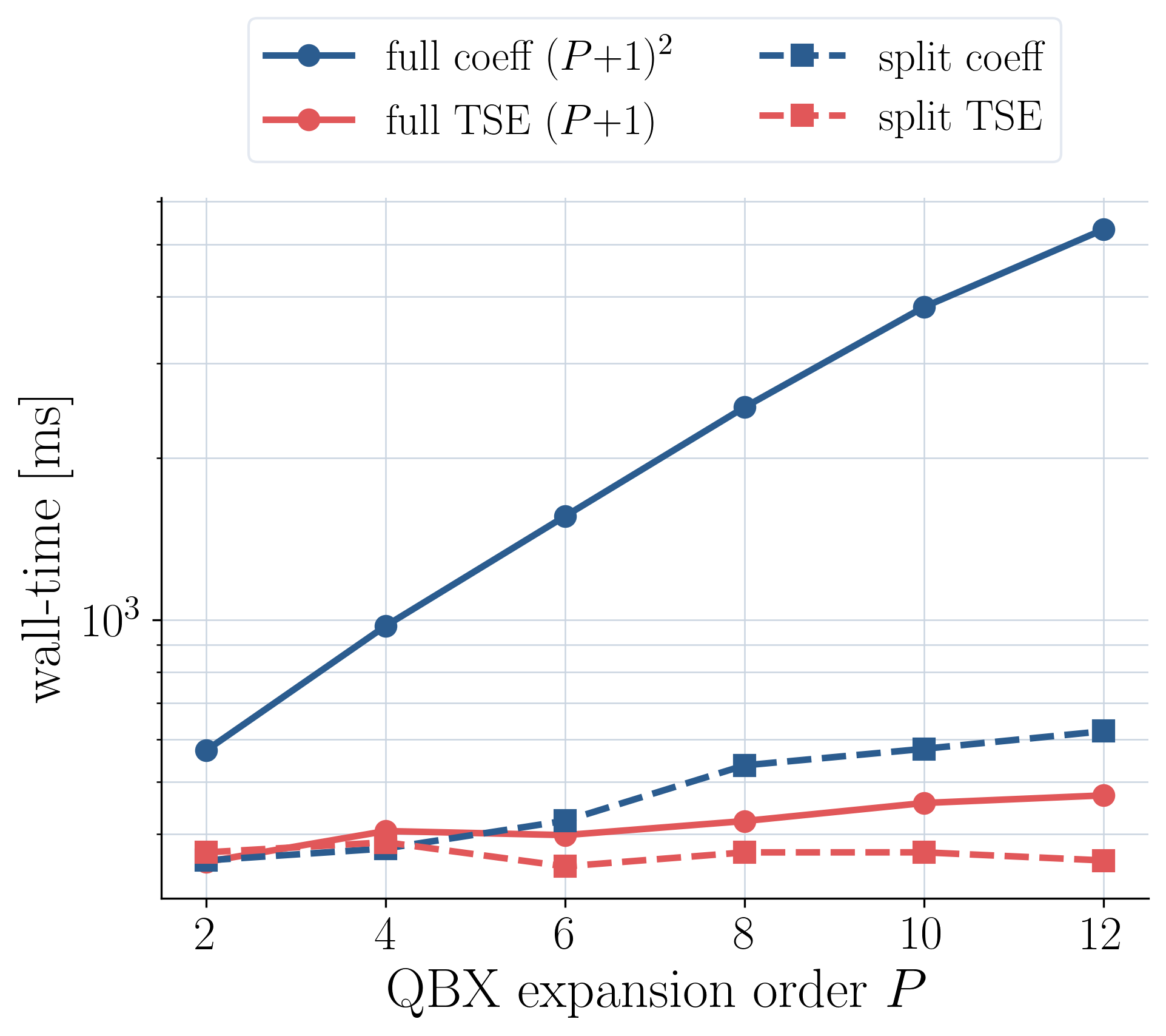}
\caption{Target-specific and coefficient QBX on the Malhotra benchmark.
\emph{Left:} $E_\infty$ relative to $\mathbf B_0$ versus $P$; both evaluations
agree to round-off and have identical error floors. \emph{Right:} wall-time
versus $P$ for the four paths of Table~\ref{tab:paths2x2}, showing the reduced
$P$-scaling of TSE compared with coefficient evaluation.}
  \label{fig:paths_equiv}
\end{figure}

\begin{figure}[t]
  \centering
  \includegraphics[width=0.495\linewidth]{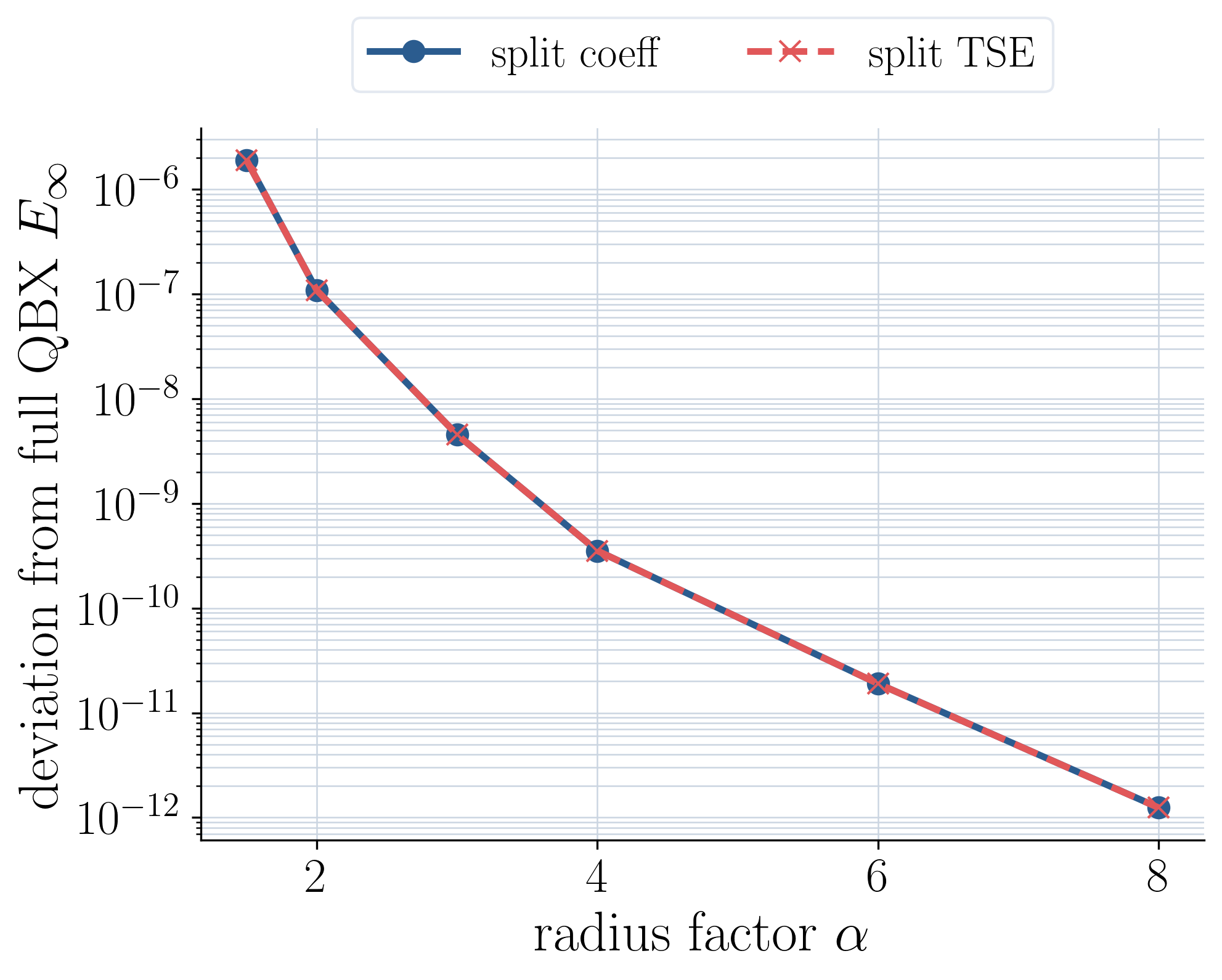}
  \hfill
  \includegraphics[width=0.495\linewidth]{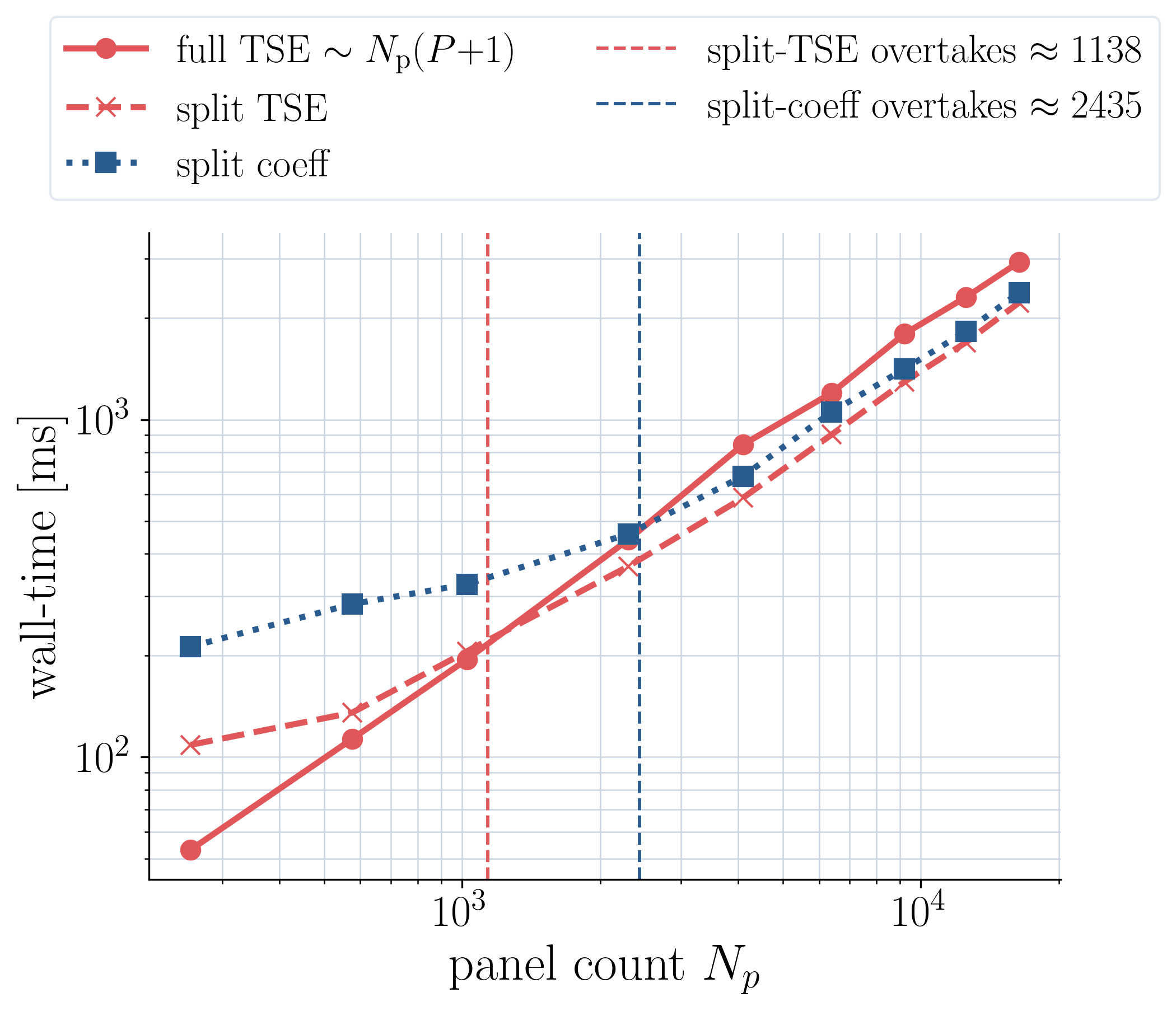}
  \caption{The local-correction split on the Malhotra benchmark.
\emph{Left:} deviation from the full-QBX baseline versus $\alpha$ for
coefficient and TSE paths. \emph{Right:} wall-time versus number of panels at
fixed target count.}
  \label{fig:paths_cost}
\end{figure}

\subsection{REGCOIL sheet currents and two-sided QBX}
\label{sec:res_w7x}
We finally consider a realistic application to the Wendelstein~7-X stellarator
~\cite{klinger2019}. The VMEC vacuum equilibrium
~\cite{hirshman1983vmec} is converted by CLEBSCH (\ref{app:CLEBSCH_Kokkos})
into the self-contained HDF5 representation required by AMPERe. The
sheet-current source is generated internally using the Kokkos implementation
of REGCOIL~\cite{landreman2017regcoil} (see ~\ref{app:REGCOIL_Kokkos}), with the regularization parameter
$\lambda$ varied from the NESCOIL~\cite{merkel1987} limit to strongly
regularized configurations.

Figure~\ref{fig:w7x_regcoil} verifies the generated current distribution
against the VMEC reference field. Increasing the current-potential resolution
$M=m_{\rm pol}=n_{\rm tor}$ reduces the relative RMS error until saturation
near $M=16$, reaching $4\times10^{-5}$ for $RB_\phi$ and
$(1.9,2.4)\times10^{-4}$ for $(B_Z,B_R)$. At $M=24$, the remaining residual
is spatially featureless, with the apparent increase at the interior field
null resulting from the ill-conditioning of relative errors. The AMPERe
REGCOIL implementation agrees with the reference Fortran solver within
$0.001\%$ RMS as detailed in \ref{app:REGCOIL_Kokkos}. The complete workflow, including equilibrium conversion, source
construction, panelization, and field evaluation, is executed within AMPERe framework.

The evaluation study uses the REGCOIL solution at the L-curve knee on a
conformal winding surface offset by $0.30$~m from the LCFS
($m_{\rm pol}=n_{\rm tor}=24$, $n_{\rm fp}=5$). The sheet current
$\mathbf K=\hat{\mathbf n}\times\nabla\Phi$ is discretized using
$40\times200=8000$ cubic panels generated from a $120\times600$ parameter
grid. Targets consist of a $32\times32\times10$ cylindrical grid containing
$10{,}240$ points that straddle the winding surface
(left panel of Fig.~\ref{fig:w7x_conformal}), including the region between the
winding surface and the LCFS.

Since targets are located on both sides of the current sheet, the one-sided
assumption of Section~\ref{sec:hybrid} is extended by assigning each target
to the appropriate QBX expansion side based on its location relative to the
surface. This two-sided treatment reduces the gap-region error from $0.103$
for a single global side choice to $8.5\times10^{-4}$
(Fig.~\ref{fig:w7x_jump}, right). Along the outboard normal, both sides exhibit the
same near-surface behavior: at $d/h=10^{-4}$, direct GL gives errors of
$9.13$ and $0.48$, whereas hybrid QBX with $P=8$ achieves
$1.5\times10^{-3}$ and $1.0\times10^{-4}$ (Fig.~\ref{fig:w7x_conformal}, right). Increasing the expansion order
reduces the QBX error to $3\times10^{-9}$ at $P=16$.
The apparent error discontinuity at $d/h=1$ results from the dispatcher
switch: beyond this distance both methods use the same GL evaluation and
therefore differ only by round-off. The QBX-active region is therefore
$d/h<1$.
\textcolor{red}{The shape of the error curves within this region follows
directly from the geometry of the QBX expansion. The convergence ball has
radius $h/2$, so that a target at $d/h=0.5$ coincides with the expansion
centre, where the truncation error reaches machine precision. As the target
approaches the boundary of the convergence ball, $d/h\to1$, the convergence
of the $P=8$ expansion deteriorates and its error increases. The GL error
shows the opposite trend, decreasing with distance from the surface and
becoming competitive with QBX away from the near-surface region. The two
curves therefore cross at $d/h\approx0.6$, beyond which GL is more accurate
up to the threshold at $d/h=1$. This crossover does not compromise the
accuracy of the hybrid strategy: the $P=8$ QBX error remains below $10^{-5}$
throughout the interval, whereas lowering the threshold would expose
near-surface targets to GL, for which the error can become $O(1)$. The small
increase in the GL error immediately below $d/h=1$ is instead associated
with the reference solution: the $P=16$ expansion is itself least accurate
near the boundary of the convergence ball, and the relative error there
therefore reflects, in part, the error of the reference. Overall, the
observed features---the minimum at the expansion centre, the deterioration
of the QBX expansion toward the boundary of its convergence ball, and the
vanishing contribution of QBX beyond the switching threshold---are
consistent with the expected convergence behavior and provide a useful
consistency check of the implementation.}

Finally, the BVH-based local correction reproduces the full hybrid result with
a maximum deviation of $1.8\times10^{-12}$ while reducing the
$10{,}240$-target evaluation time from $239$~s to $61$~s
(Fig.~\ref{fig:w7x_accel}, right). The deviation from the unaccelerated hybrid
evaluation decreases spectrally with the radius factor $\alpha$, with larger
$P$ allowing tighter agreement, until reaching the assembly round-off floor (Fig.~\ref{fig:w7x_accel}, left).
The accelerated cost remains nearly independent of $\alpha$ and stays below
the corresponding unaccelerated evaluations.

Because the source is a genuine current sheet, it provides a physical
validation beyond self-convergence: the two-sided field must satisfy
$\hat{\mathbf n}\times(\mathbf B_{\rm out}-\mathbf B_{\rm in})=\mu_0\mathbf K$
with a continuous normal component. The evaluation reproduces the tangential
jump to $1.3\%$ of $|\mu_0\mathbf K|$ and normal continuity to
$9.3\times10^{-7}$ (Fig.~\ref{fig:w7x_jump}, left), recovering the coil-design
sheet current from its own field.

The same framework is applied to the virtual-casing configuration of
Table~\ref{tab:configs}. The total equilibrium field on the W7-X LCFS is
converted into Stratton--Chu surface traces using $20{,}480$ quadratic panels
generated from the $128\times640$ parameter grid. Since the LCFS is a
stellarator flux surface, $\sigma=\hat{\mathbf n}\cdot\mathbf B$ vanishes to
$10^{-19}$, and the virtual-casing field is therefore determined entirely by
the $\mathbf K$ contribution.

Along the outboard normal (Fig.~\ref{fig:w7x_vcp}), direct GL loses accuracy
in the near-surface regime, reaching a relative error of $2\times10^{3}$ at
$d/h=10^{-4}$. In contrast, the hybrid evaluator remains stable, with errors
of $8.9\times10^{-3}$ at $P=8$ and $1.3\times10^{-3}$ at $P=12$. The reference
solution is obtained from a higher-order QBX evaluation on the same panels,
so the reported error measures the QBX truncation and discretization effects.
For $d/h>1$, both evaluations switch to the GL branch and therefore differ
only at round-off level.

Moreover, the reconstructed field preserves the equilibrium symmetries to near
machine precision (Fig.~\ref{fig:w7x_sym}). The five-fold field-period
residual is $6.6\times10^{-10}$ rms
($6.6\times10^{-9}$ at the worst exterior point), while the stellarator-
symmetry residual is $8.1\times10^{-11}$ rms
($2.1\times10^{-9}$ maximum). The divergence error remains at the far-field
machine-precision level. These symmetries are not explicitly imposed in the
evaluator: the surface discretization, BVH construction, and QBX expansions
do not encode $n_{\rm fp}$ or stellarator symmetry. Their recovery at the
$10^{-9}$ level therefore provides an end-to-end verification of the complete
pipeline, from surface traces through quadrature evaluation and acceleration.

The reconstructed field magnitude also illustrates the physical scale of the
virtual-casing contribution (Fig.~\ref{fig:w7x_vcp}, right). For this vacuum
equilibrium, the virtual-casing field is associated with the residual
toroidal current \textcolor{red}{(${\approx}194$~A)}, whose characteristic scale
$\mu_0 I/(2\pi a)\approx10^{-4}$~T is consistent with the reconstructed
surface field magnitude of approximately $6\times10^{-4}$~T. Along the
evaluation sweep, the magnitude decreases monotonically to
$5\times10^{-7}$~T. This scale is small compared with the $2.5$~T W7-X vacuum
field, whose experimental verification reached better than one part in
$10^{5}$ accuracy~\cite{pedersen2016}. The numerical challenge is therefore
the accurate reconstruction of a small residual field: where the true field
is $10^{-4}$~T, direct GL produces spurious values exceeding $1$~T near the
surface, whereas the hybrid evaluator follows the reference smoothly across
the near-surface region. We expect that, in a finite-$\beta$ equilibrium, the same method can be
applied directly to recover the plasma field with the accuracy established
above; an investigation left for future work.

\begin{figure}[t]
  \centering
  \includegraphics[width=0.53\linewidth]{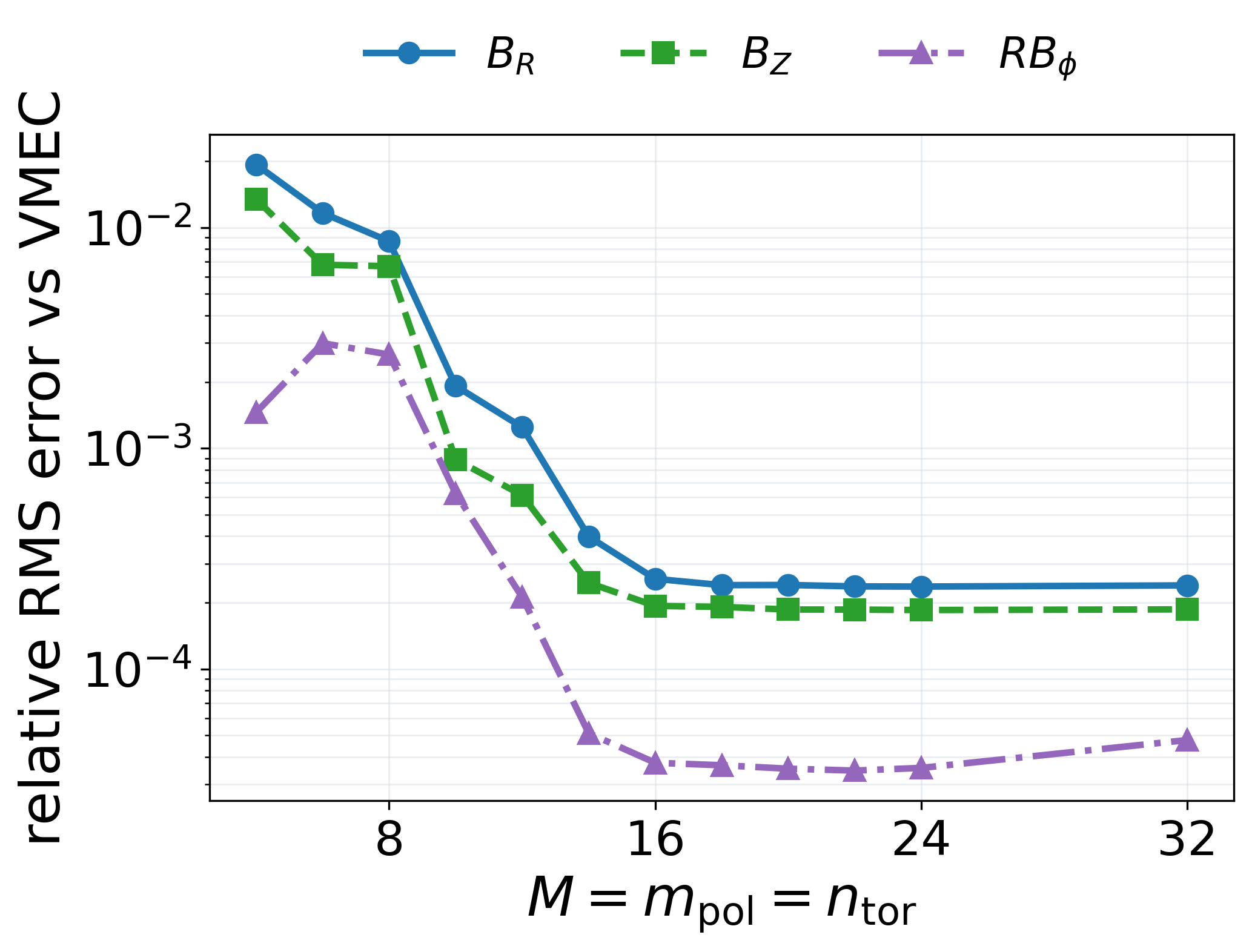}
  \hfill
  \includegraphics[width=0.44\linewidth]{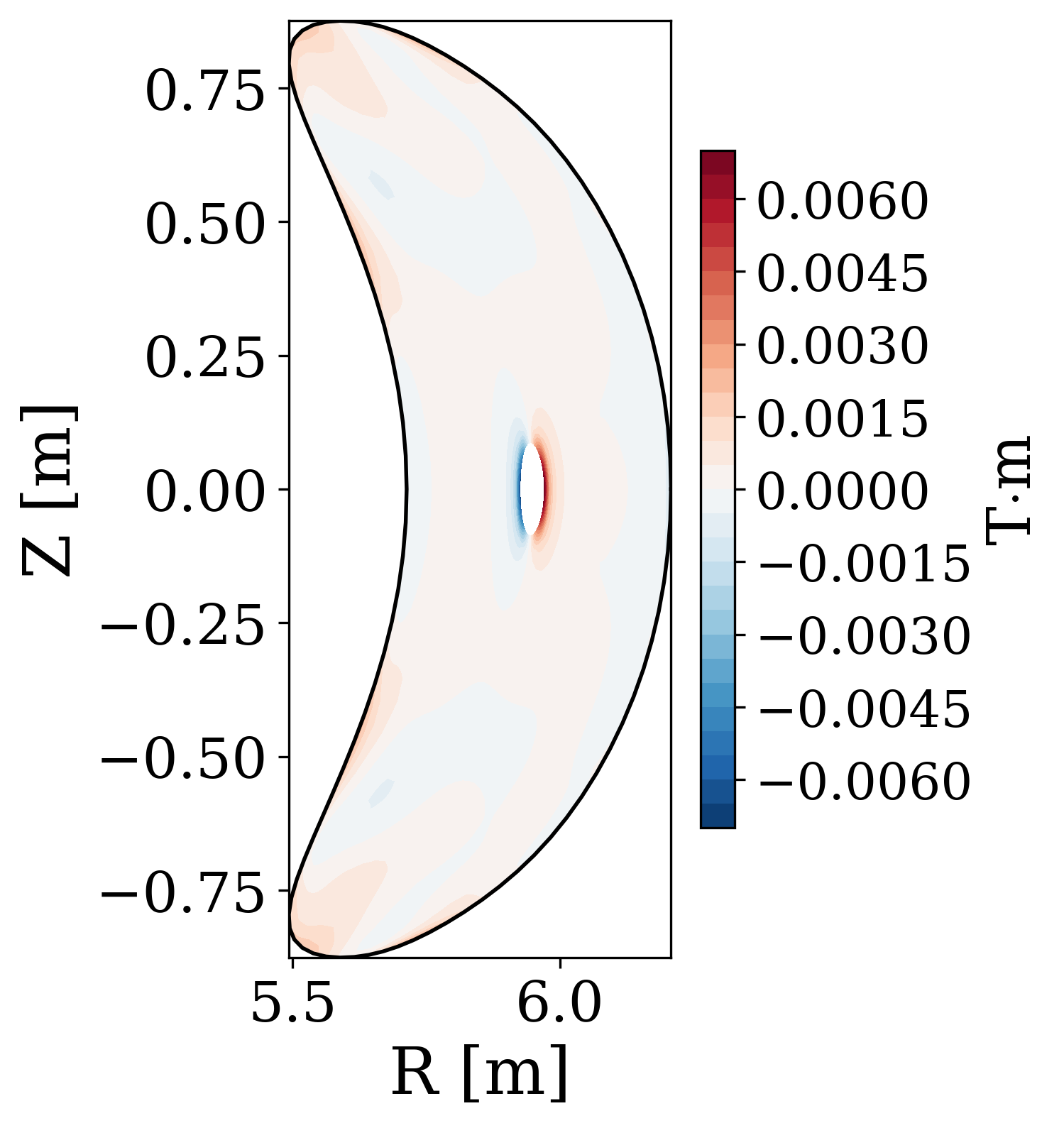}
  \caption{Equilibrium-field reproduction by the W7-X sheet-current source.
\emph{Left:} relative RMS error with respect to VMEC at interior cross-section
points versus current-potential resolution $M$ ($\lambda\to0$). The error
saturates near $M=16$, reaching $4\times10^{-5}$ for $RB_\phi$ and
$\sim2\times10^{-4}$ for $B_R$ and $B_Z$. \emph{Right:} $RB_\phi$ residual map
at $M=24$ on the bean cross-section (RMS $=0.004\%$); the remaining structure
is localized near the interior field null.}
  \label{fig:w7x_regcoil}
\end{figure}

\begin{figure}[t]
  \centering
  \includegraphics[width=0.465\linewidth]{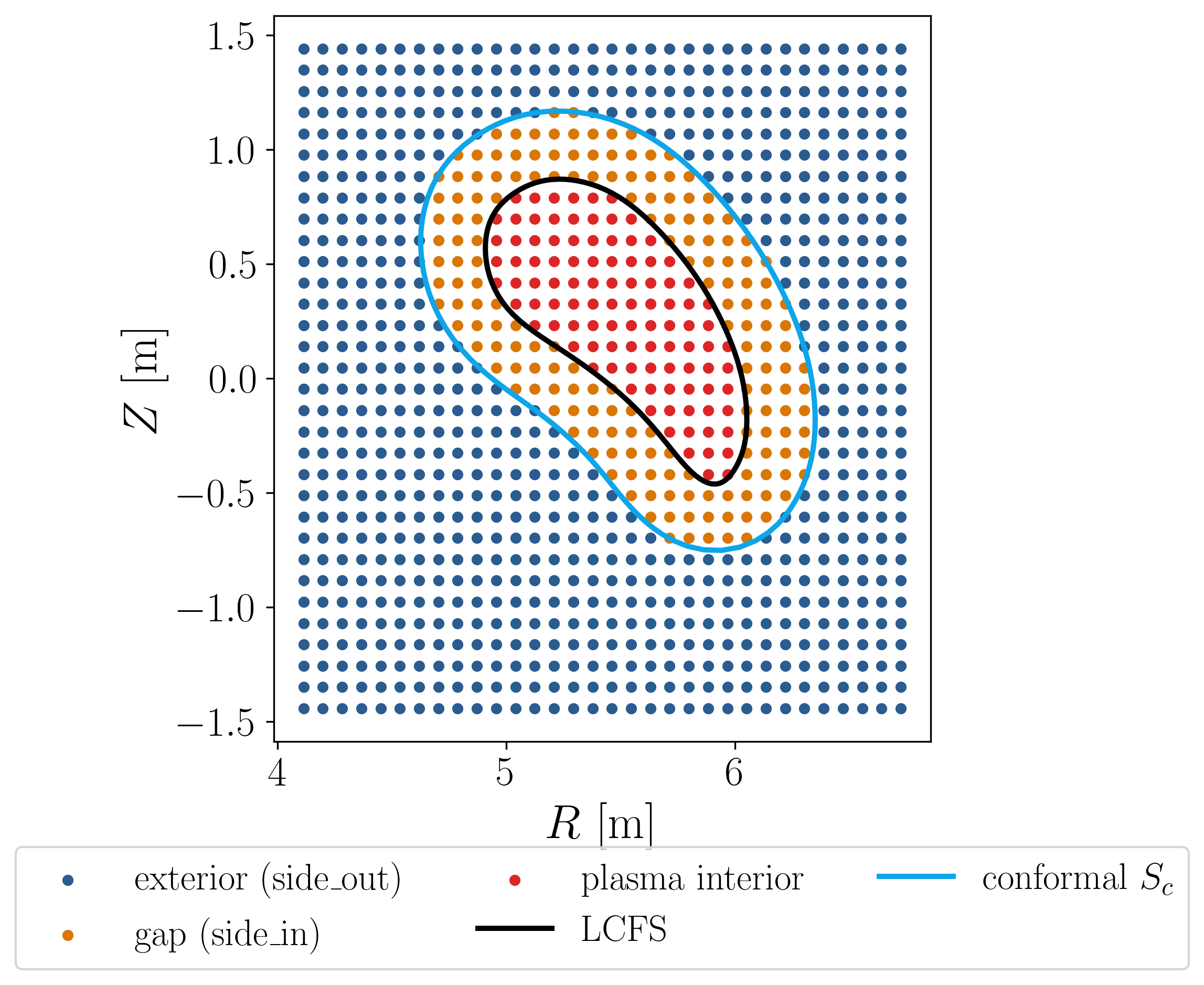}
  \hfill
  \includegraphics[width=0.525\linewidth]{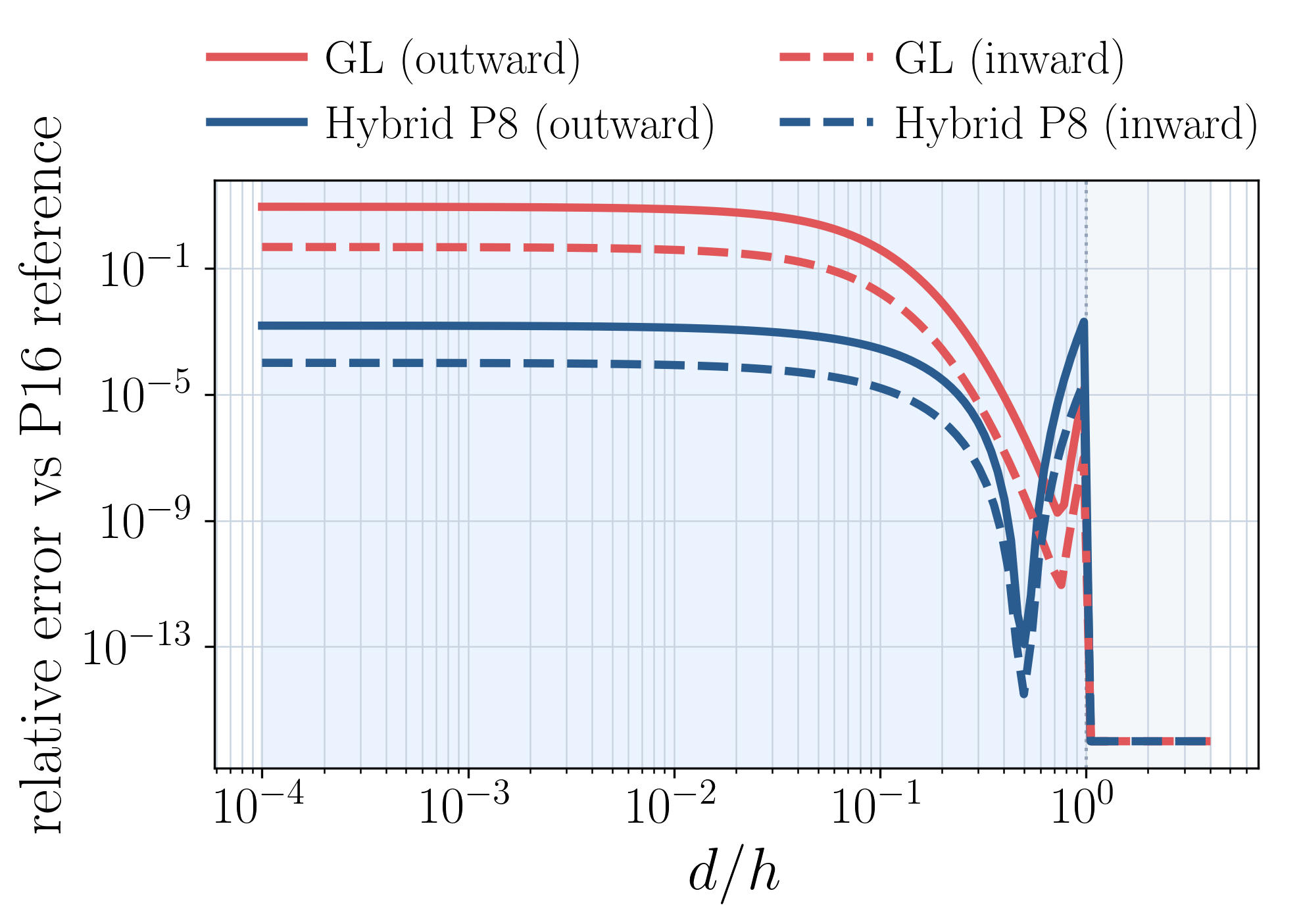}
  \caption{Two-sided QBX on the conformal winding surface.
\emph{Left:} target-side classification relative to the winding surface and
LCFS. \emph{Right:} relative error versus $d/h$ along the outboard normal:
direct GL versus two-sided hybrid QBX ($P=8$) with a $P=16$ reference.
\textcolor{red}{The shaded band $d/h<1$ is the QBX-active (near) region. For
$d/h\ge1$ the hybrid evaluator and the $P=16$ reference fall back to the same
direct-GL branch, so their difference vanishes identically and the curves drop
to round-off.}}
  \label{fig:w7x_conformal}
\end{figure}

\begin{figure}[t]
  \centering
  \includegraphics[width=0.495\linewidth]{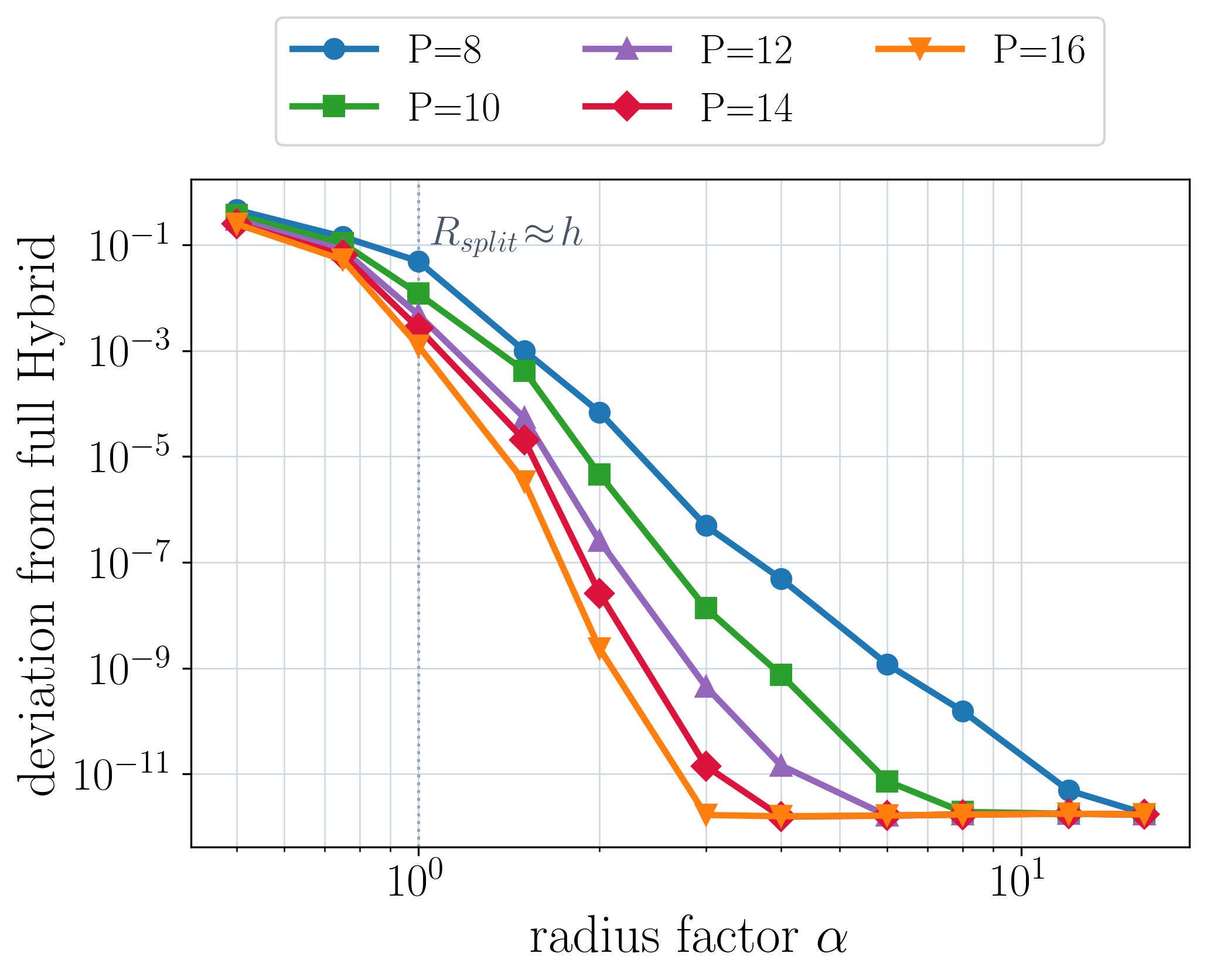}
  \hfill
  \includegraphics[width=0.495\linewidth]{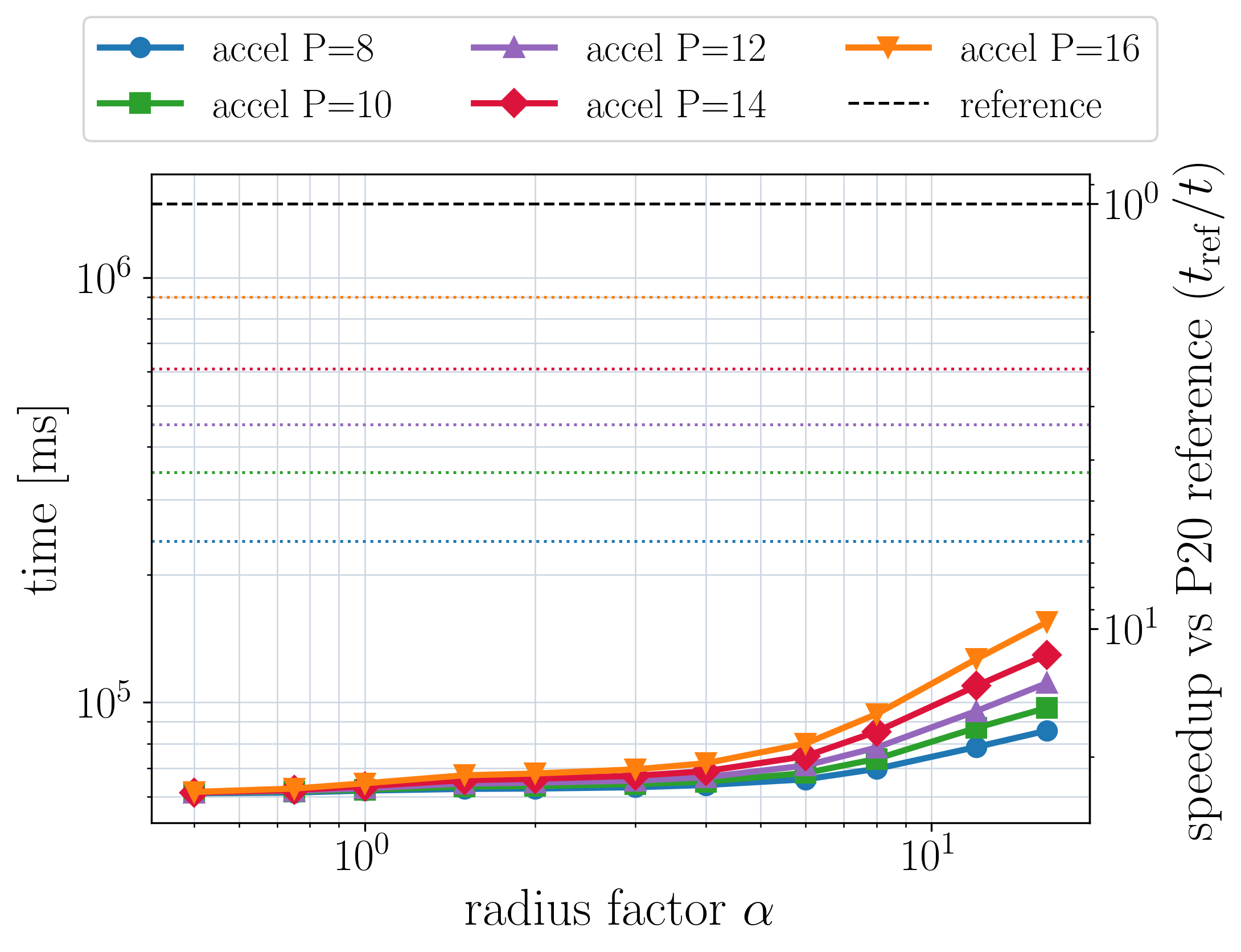}
  \caption{BVH local correction on the conformal W7-X problem
($R_{\rm split}=\alpha h$). \emph{Left:} maximum deviation of the accelerated
evaluation from the full hybrid result at fixed $P$ versus $\alpha$, showing
spectral convergence to the $\sim10^{-12}$ round-off floor. \emph{Right:}
wall-time versus $\alpha$ for the accelerated evaluation compared with the
full-hybrid baselines and the $P=20$ reference; the right axis shows the
corresponding speed-up.}
  \label{fig:w7x_accel}
\end{figure}

\begin{figure}[t]
  \centering
  \includegraphics[width=0.495\linewidth]{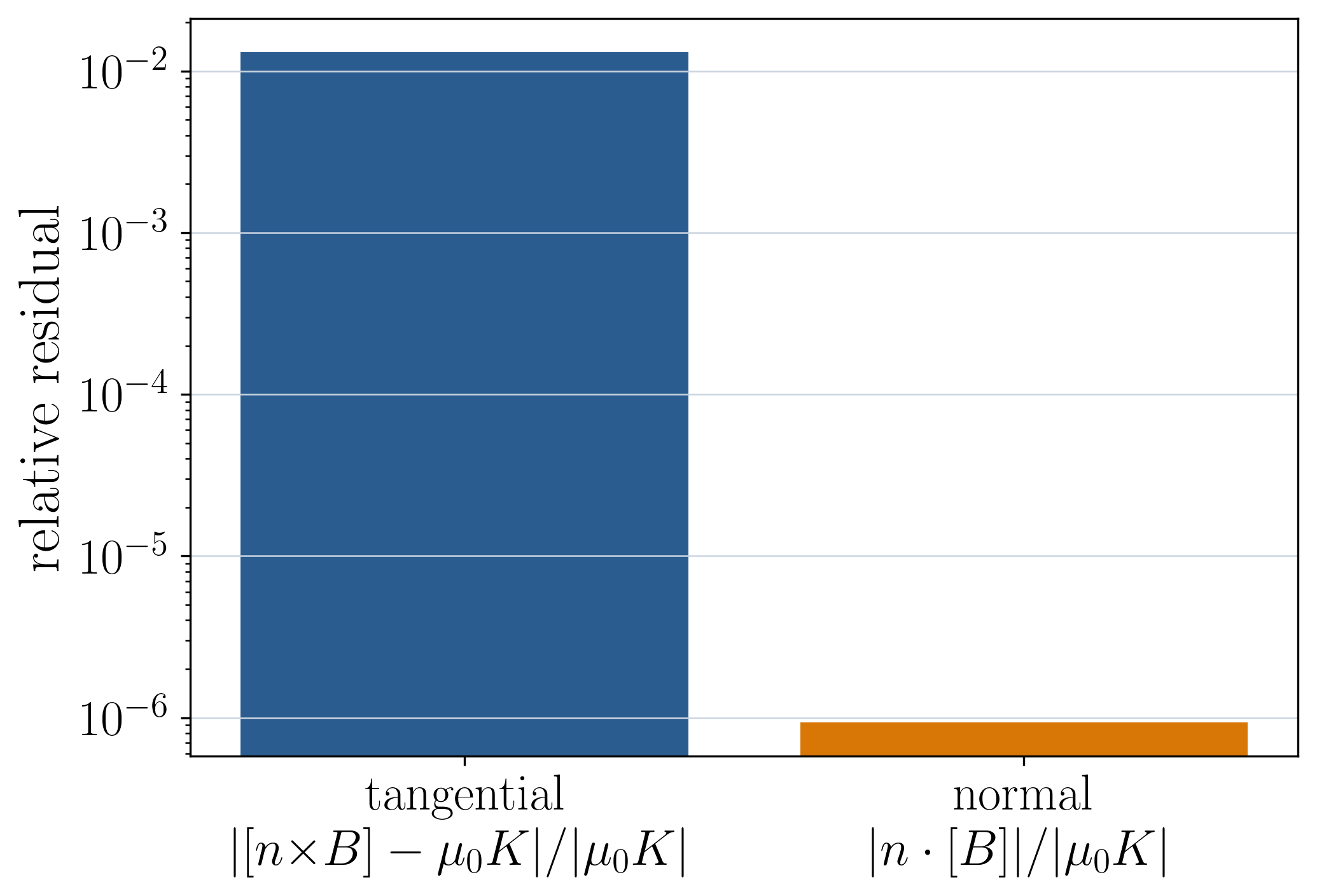}
  \hfill
  \includegraphics[width=0.495\linewidth]{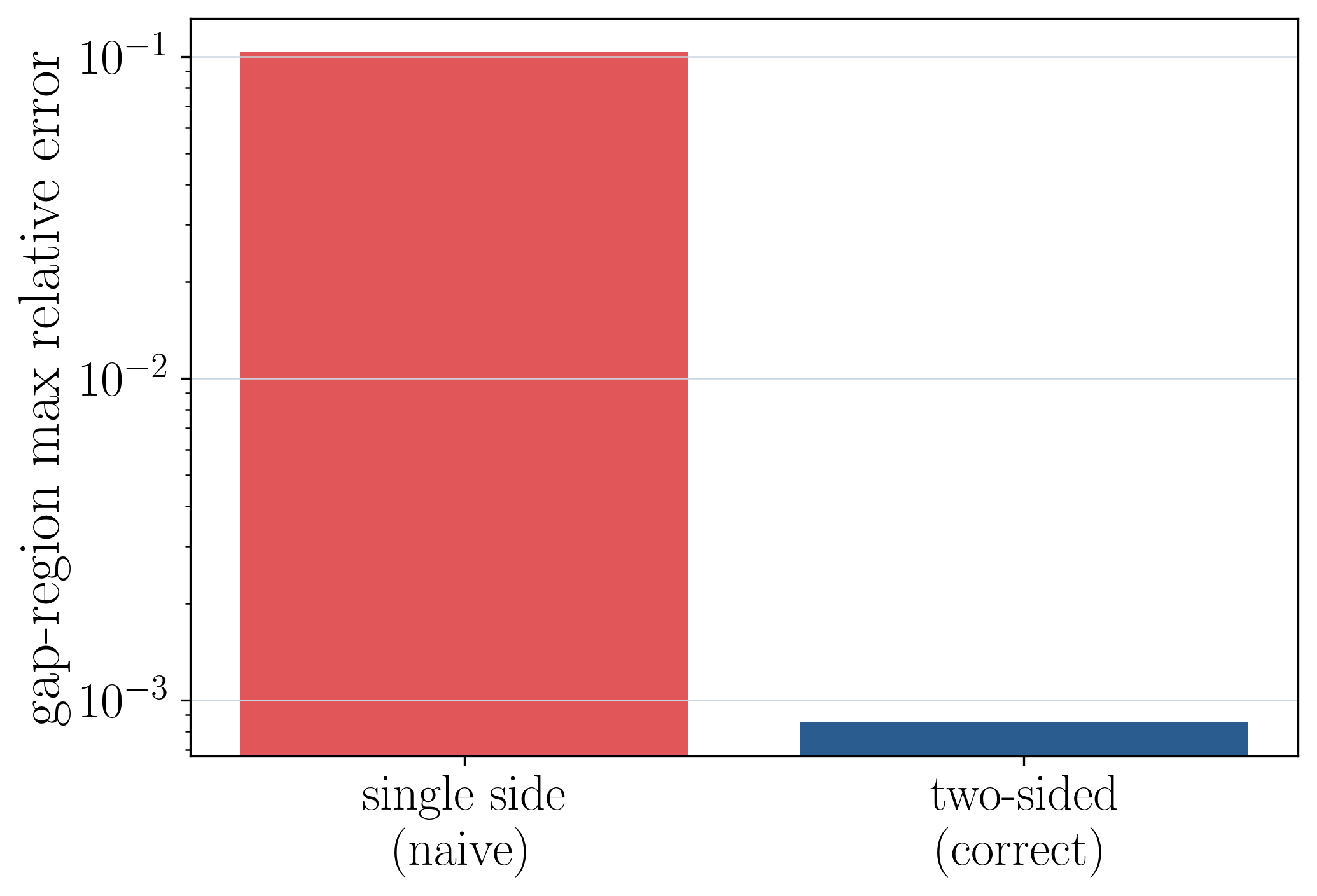}
  \caption{Verification of the two-sided QBX treatment on the winding surface.
\emph{Left:} recovery of the sheet jump condition
$\hat{\mathbf n}\times[\mathbf B]=\mu_0\mathbf K$ and continuity of the normal
field component. \emph{Right:} gap-region error for global versus per-target
expansion-side selection.}
  \label{fig:w7x_jump}
\end{figure}

\begin{figure}[t]
  \centering
  \includegraphics[width=0.495\linewidth]{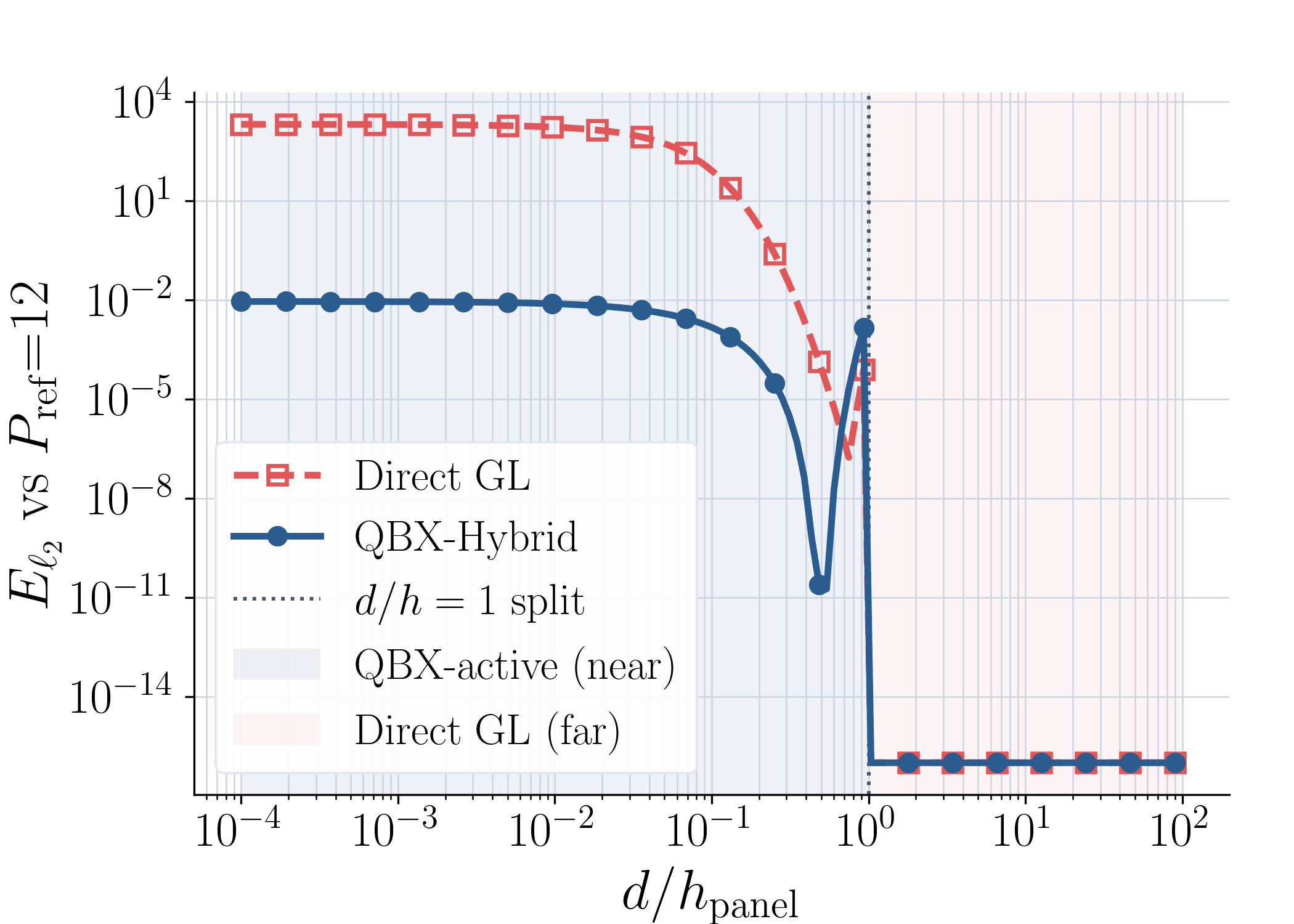}
  \hfill
  \includegraphics[width=0.495\linewidth]{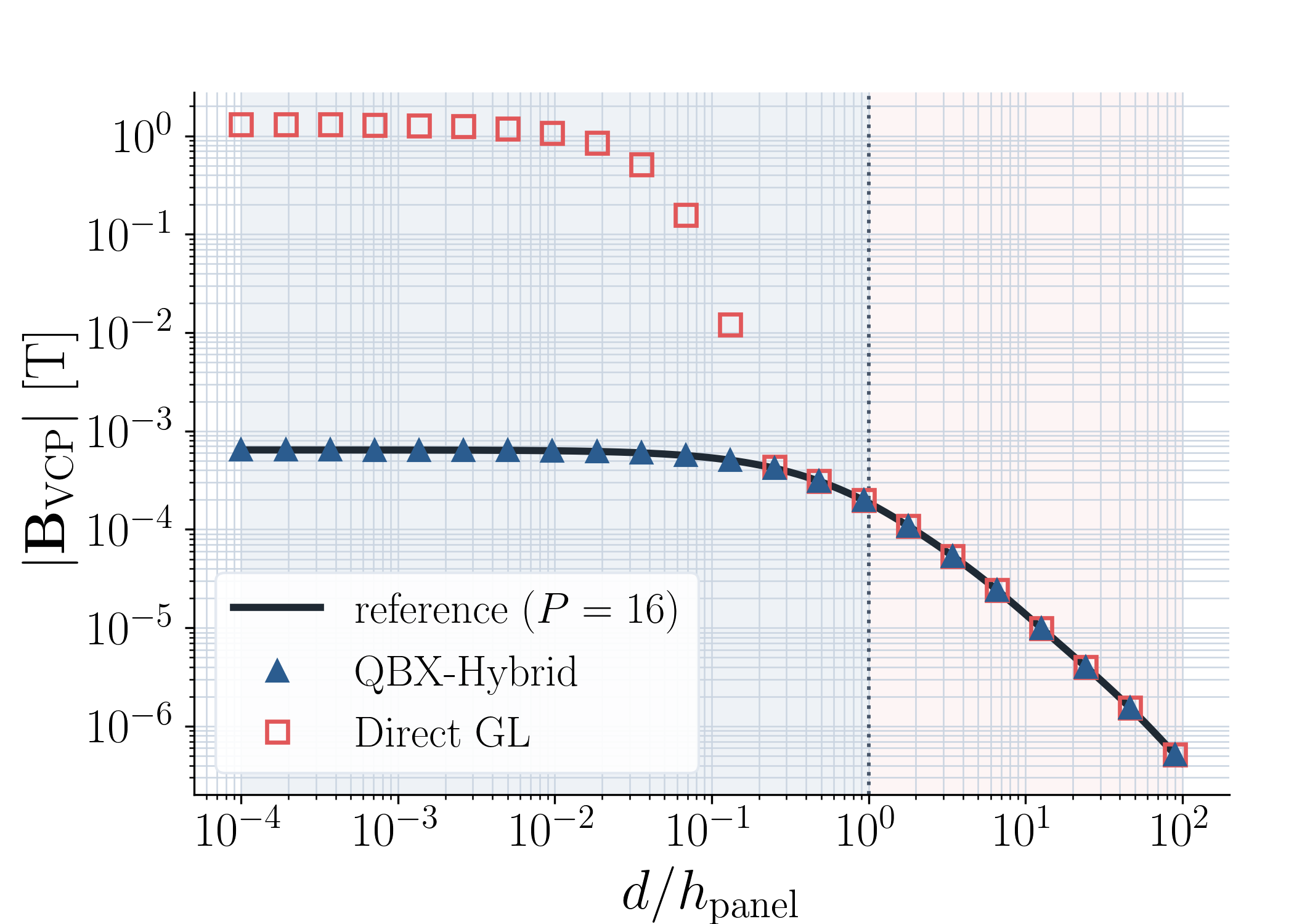}
  \caption{Virtual casing on the W7-X LCFS  along
the outboard panel normal. \emph{Left:} relative error versus
$d/h_{\rm panel}$. \emph{Right:} reconstructed virtual-casing field
magnitude for $P=12$ with a $P=16$ reference. The hybrid evaluator resolves the
$\sim10^{-4}$~T vacuum-residual field.}
  \label{fig:w7x_vcp}
\end{figure}

\begin{figure}[t]
  \centering
  \includegraphics[width=0.56\linewidth]{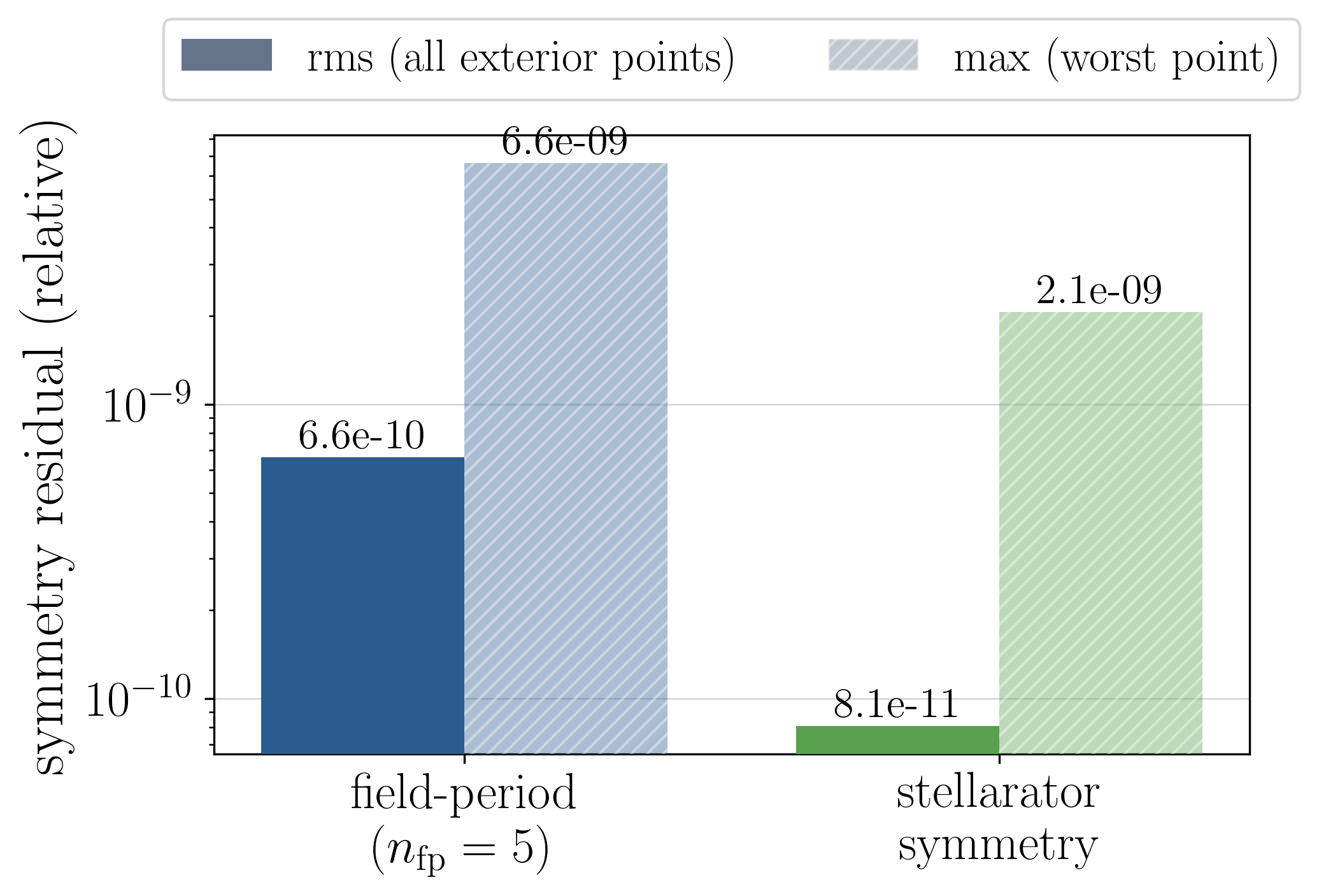}.
  \caption{Equilibrium symmetries of the evaluated W7-X exterior field:
relative residuals of the five-fold field-period and stellarator symmetries
over exterior grid points.}
  \label{fig:w7x_sym}
\end{figure}

\subsection{Structure preservation}
\label{sec:res_divfree}

Proposition~\ref{prop:divfree} shows that the QBX $\mathbf K$-term field is
solenoidal by construction at every truncation order and provides an analytic
Jacobian from the same local expansion. This property is verified on the
Malhotra configuration (exact reference) and the W7-X LCFS (production
geometry) by comparing the analytic divergence with finite-difference
approximations of the evaluated field.

For the Malhotra configuration (Fig.~\ref{fig:divfree}, left), the analytic
divergence remains at round-off level ($\sim10^{-20}$), whereas finite
differences exhibit the expected truncation--round-off trade-off. The optimal
steps give minimum errors of $7\times10^{-15}$ for FD-4 and
$2\times10^{-15}$ for FD-6. These values are obtained only for very small
spatial steps and are not practical for production solvers due to the
associated computational cost.

For the W7-X near-surface shell (Fig.~\ref{fig:divfree}, right), using the
normalized metric $|\nabla\cdot\mathbf B|a/|\mathbf B|$, the analytic
divergence has an rms value of $1.1\times10^{-13}$ and a maximum of
$9.1\times10^{-13}$, compared with $2.1\times10^{-3}$ for FD-2 and
$1.8\times10^{-7}$ for FD-4. The observed discrepancy is therefore entirely
attributed to finite-difference error rather than to a violation of
$\nabla\cdot\mathbf B=0$ by the evaluator.

For applications requiring field derivatives, such as field-line tracing or
stability calculations, the analytic Jacobian provides derivatives directly
from the QBX expansion, avoiding the step-size selection and accuracy
limitations associated with finite differences.

The curl provides the complementary property discussed in
Remark~\ref{rem:curl}. It is evaluated analytically from the same local
expansion,
$\nabla\times\mathbf B_P=\nabla(\nabla\cdot\mathbf A_P)$, and agrees with a
fourth-order finite-difference Jacobian to $2\times10^{-12}$. For a closed-loop
source, the result is at round-off level ($5\times10^{-16}$) for all tested
$P$. Since the Stratton--Chu surface trace is analytically divergence-free, the
curl is expected to reflect the accuracy of the discrete source
representation rather than QBX truncation, as confirmed in
Fig.~\ref{fig:curl}.

For the Malhotra configuration, the normalized curl
$|\nabla\times\mathbf B|\,r_{\rm QBX}/|\mathbf B|$ is nearly independent of
$P$ ($\sim6\times10^{-7}$), but decreases strongly with mesh refinement:
$4.8\times10^{-4}$, $3.1\times10^{-5}$, and $6.8\times10^{-7}$ for increasing
resolution. Reducing the panel quadrature order to $n_{\rm gl}=4$ increases
the error to $4\times10^{-6}$. For the W7-X near-surface shell, the normalized
curl $|\nabla\times\mathbf B|\,a/B_0$ has an rms value of
$5.3\times10^{-5}$ and a maximum of $2.2\times10^{-4}$, again with negligible
dependence on $P$, while the divergence remains at round-off level. The curl
therefore quantifies the accuracy of the discrete source data, whereas the
divergence verifies the solenoidal construction of the evaluator.

\begin{figure}[t]
  \centering
  \includegraphics[width=0.495\linewidth]{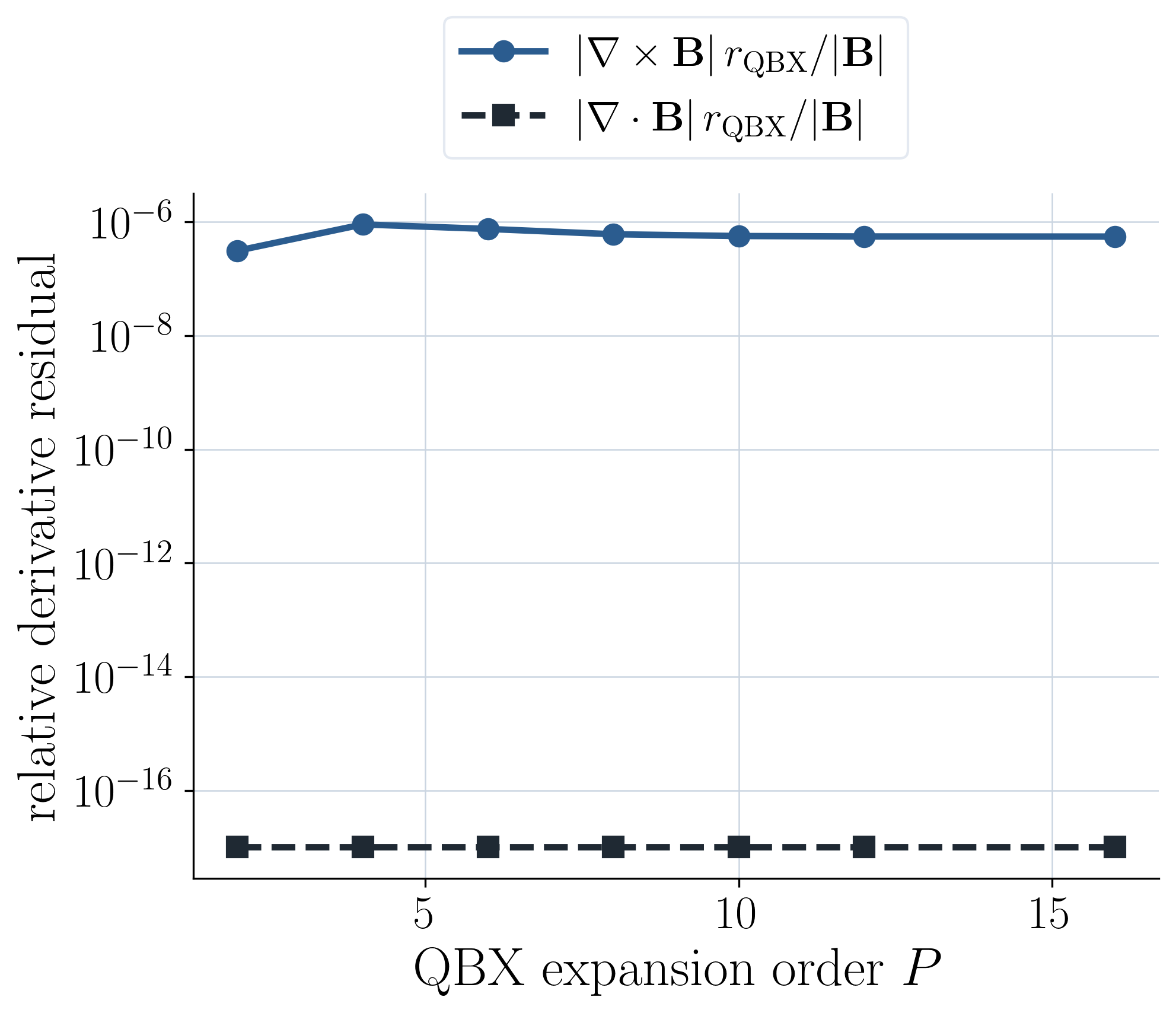}
  \hfill
  \includegraphics[width=0.495\linewidth]{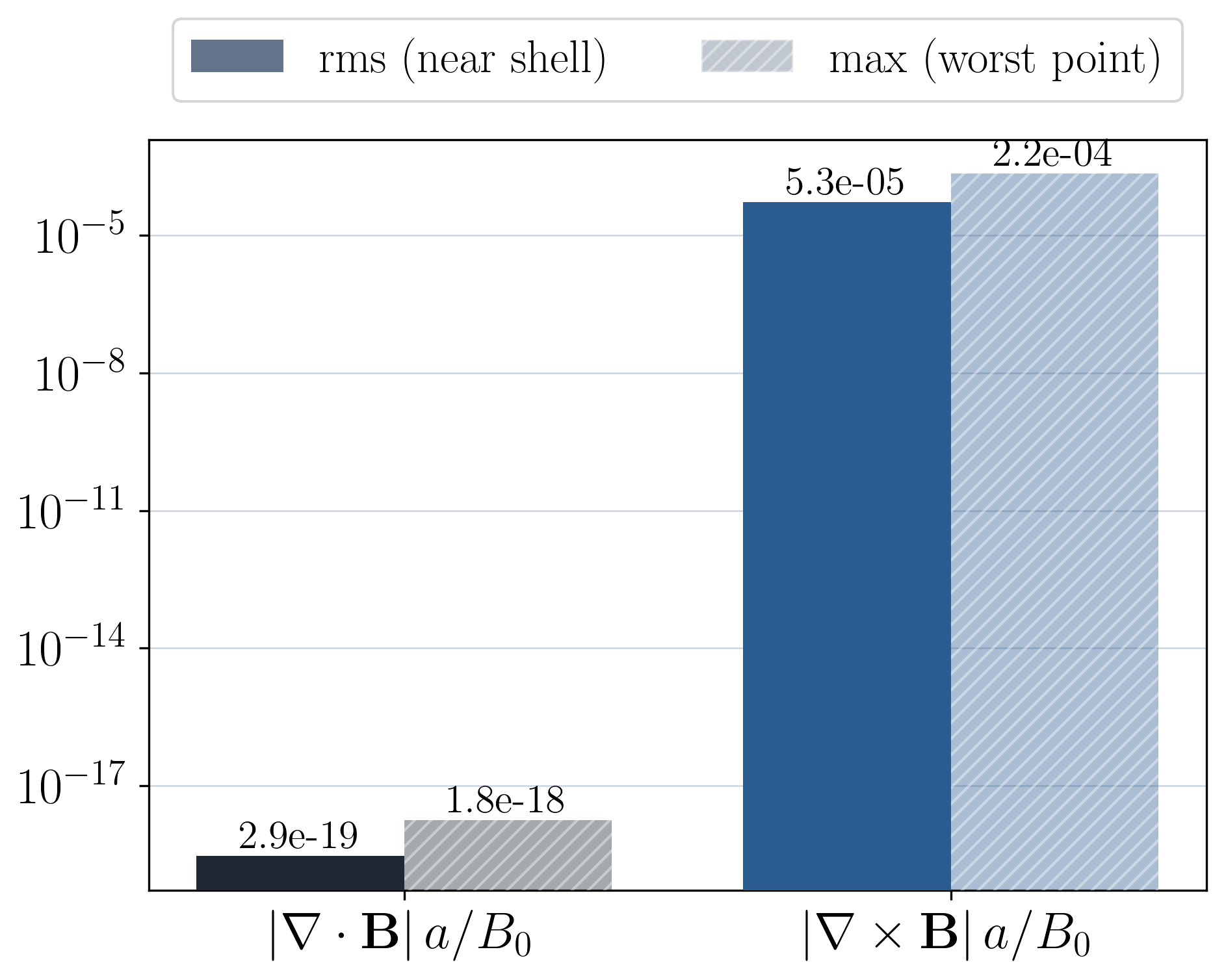}
  \caption{Div/curl asymmetry. Both derivatives are
computed analytically from the same QBX expansions. \emph{Left:} Malhotra
case. \emph{Right:} W7-X current sheet.}
  \label{fig:curl}
\end{figure}

\begin{figure}[t]
  \centering
  \includegraphics[width=0.495\linewidth]{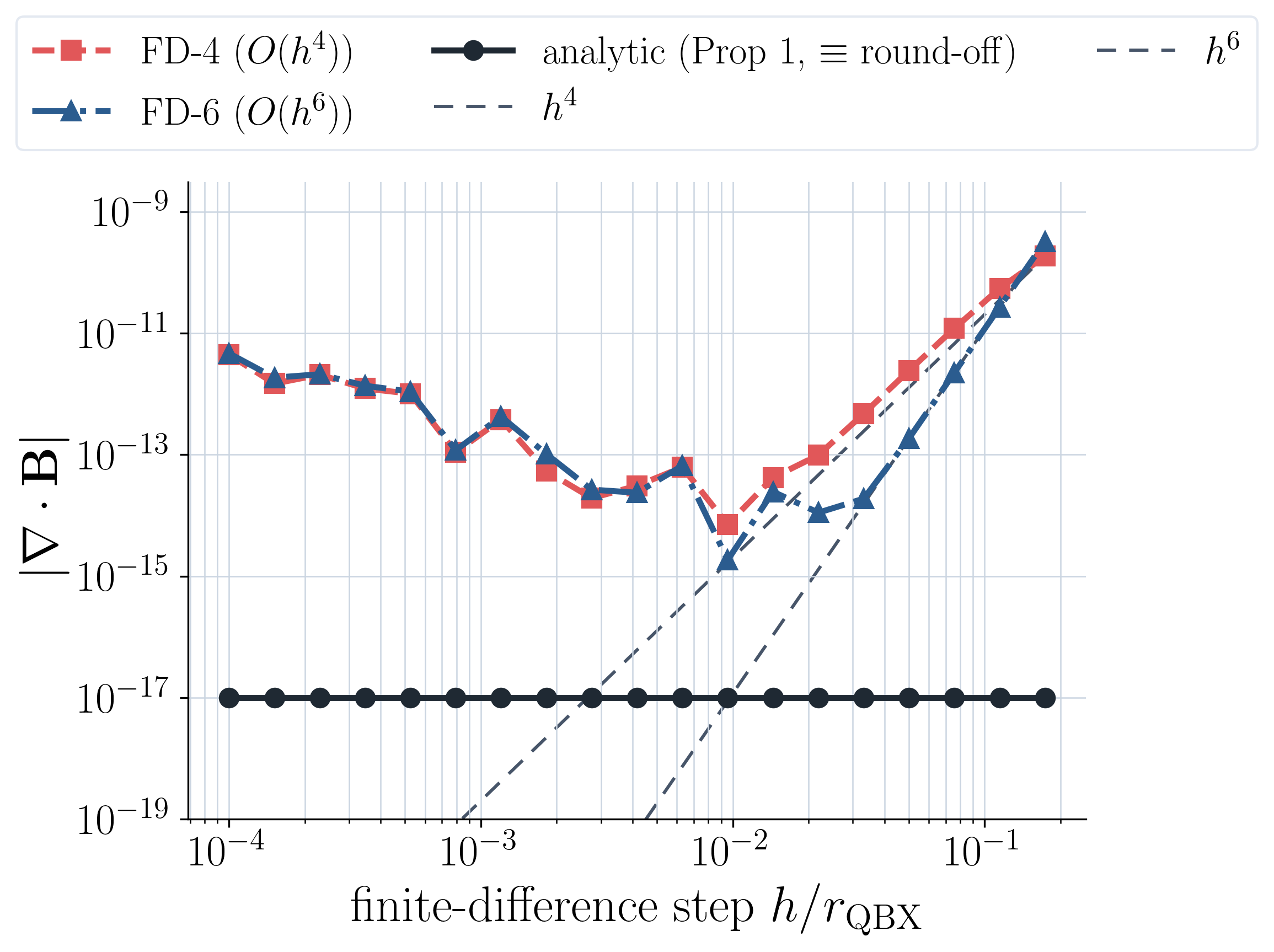}
  \hfill
  \includegraphics[width=0.495\linewidth]{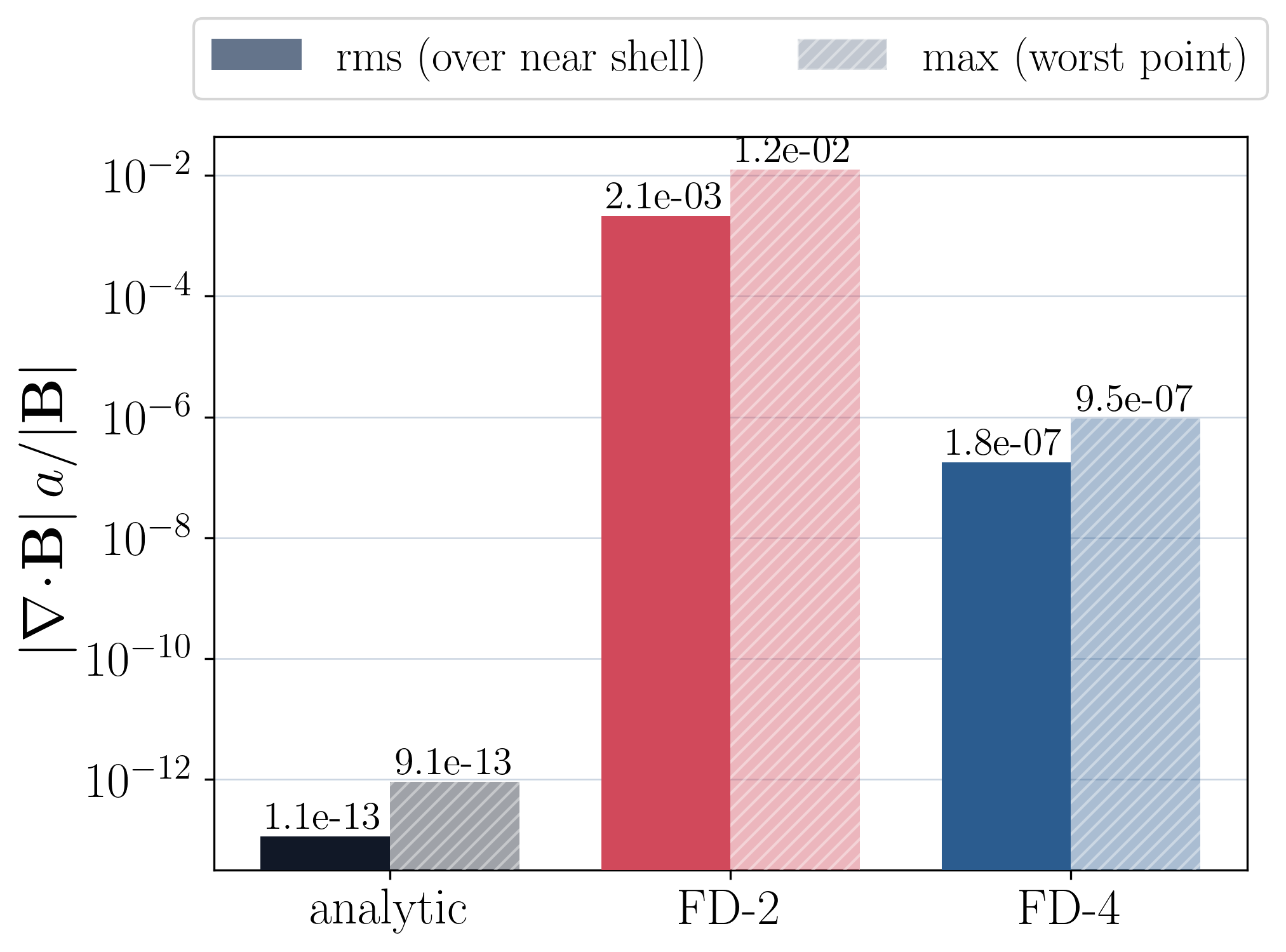}
  \caption{Divergence-free evaluation from a single QBX expansion. \emph{Left:}
Malhotra configuration: analytic divergence at round-off compared with FD-4/FD-6
convergence curves. \emph{Right:} W7-X near-surface shell.}
  \label{fig:divfree}
\end{figure}

\subsection{Performance portability: the same evaluator on CPU and GPU}
\label{sec:res_gpu}

All CPU timings use the OpenMP backend on a dual-socket Intel Xeon
E5-2640~v4 node (40 threads), while GPU timings use the \emph{same source
code} through the Kokkos CUDA backend~\cite{Kokkos_2014} on a single NVIDIA
RTX~2080~Ti. On the production W7-X problem of
Section~\ref{sec:res_w7x}, the CPU and GPU produce identical physics: TSE and
coefficient paths agree to round-off ($3\times10^{-12}$--$3\times10^{-11}$),
and the local-correction deviations are identical to every reported digit.

Figure~\ref{fig:w7x_gpu} shows the performance crossover. GPU execution is
launch-limited for small target sets, while CPU time scales linearly with
workload; the split-TSE path crosses over at
$N_t\approx4\times10^3$ targets and reaches $168$~s versus $564$~s at
$65{,}536$ targets. Because both backends execute the same source, the choice
of CPU or GPU depends only on workload.

The per-path timings highlight the algorithmic accelerations. The
full-coefficient path is register-limited and gains essentially nothing from
the GPU ($0.9\times$), whereas the accelerated paths achieve speed-ups of
$1.7\times$ (full TSE), $3.2\times$ (split coefficient), and
$3.4\times$ (split TSE). Thus the target-specific contraction and local
correction not only accelerate the CPU implementation but also make efficient
GPU execution possible.

No GPU-specific optimizations were applied. The reported performance therefore
reflects portability rather than hardware tuning, leaving occupancy
optimization, batching, mixed precision, and multi-GPU execution as future
work.

\begin{figure}[t]
  \centering
  \includegraphics[width=0.495\linewidth]{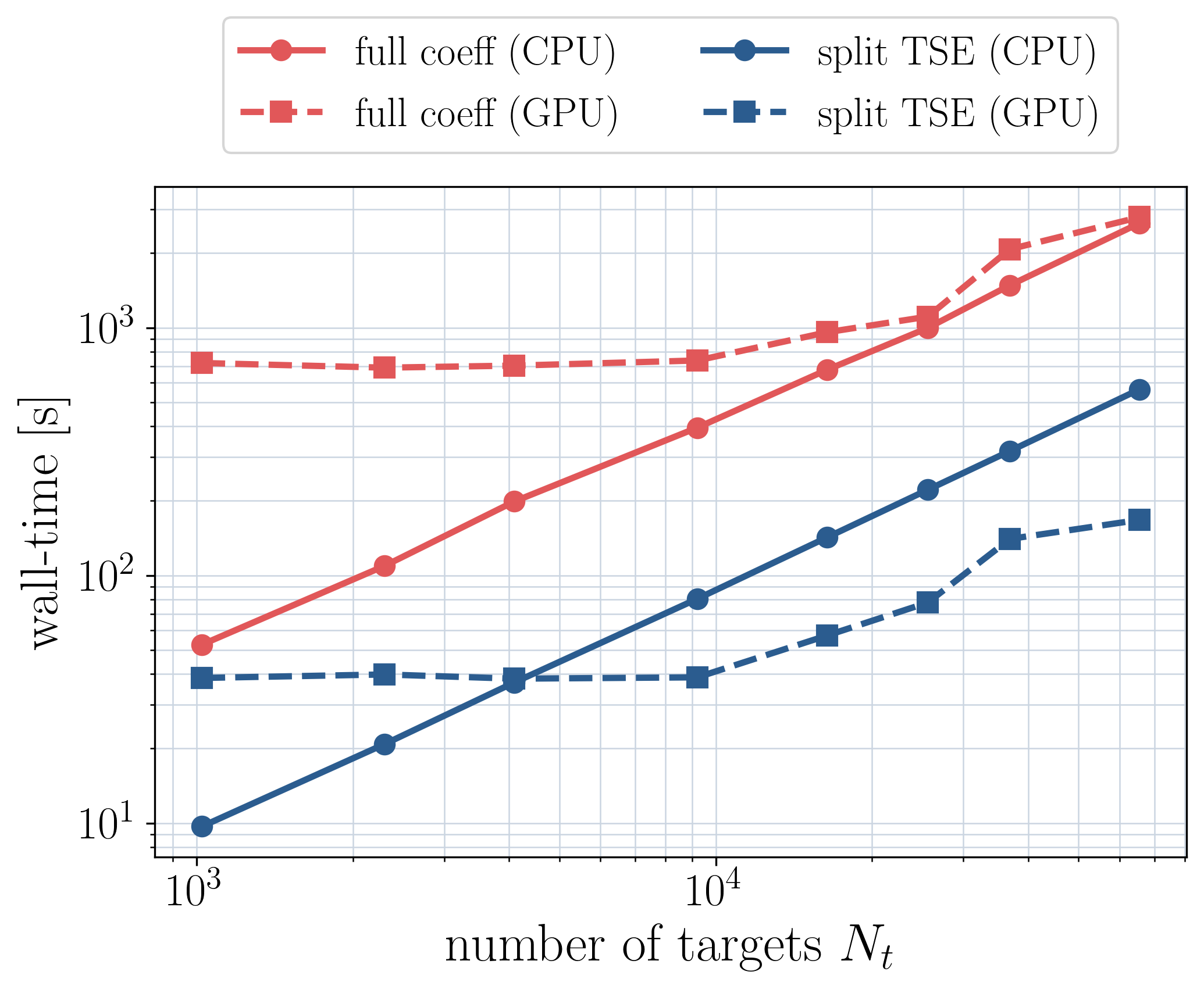}
  \hfill
  \includegraphics[width=0.495\linewidth]{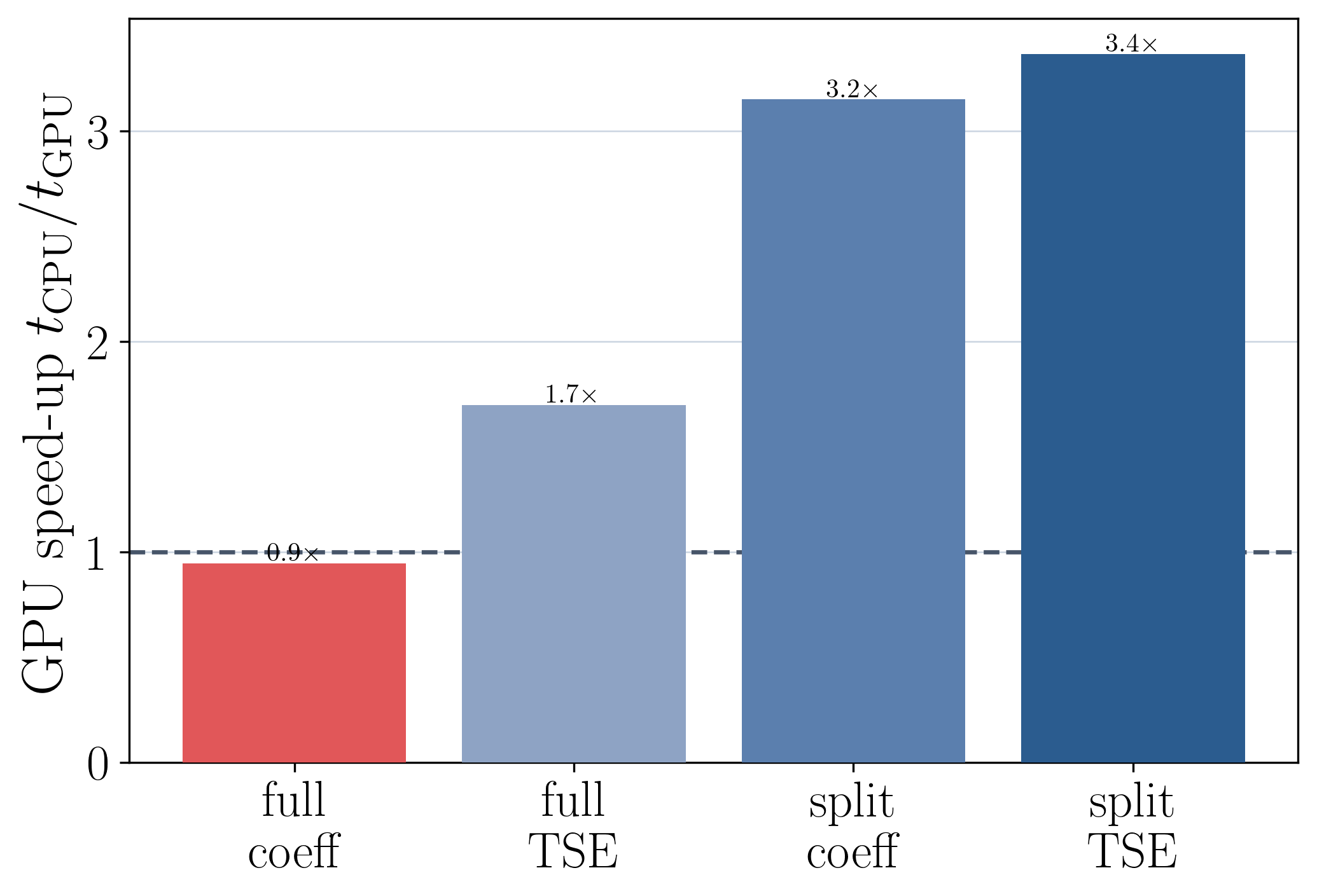}
  \caption{CPU/GPU performance portability on the W7-X conformal current sheet.
  \emph{Left:} wall-time versus target
count for full-coefficient and split-TSE paths on OpenMP (solid) and CUDA
(dashed); GPU acceleration appears beyond $N_t\approx4\times10^{3}$ for the
accelerated path. \emph{Right:} GPU speed-up at $N_t=65{,}536$: coefficient
evaluation is register-bound, whereas TSE and split paths achieve up to
$3.4\times$ acceleration.}
  \label{fig:w7x_gpu}
\end{figure}

\section{Software}
\label{sec:software}

\paragraph{Implementation and availability.}
AMPERe is implemented in C++17 as a header-only library built on the
Kokkos performance-portability framework~\cite{Kokkos_2014}, supporting CPU
(OpenMP, Serial) and GPU (CUDA, HIP, SYCL) backends through a unified execution
model. AMPERe is independent of equilibrium-file formats and consumes only
self-contained HDF5 inputs generated by CLEBSCH, which performs the required
coordinate transformations when interfacing with fusion-device equilibria.

The source code is released under the MIT license at
\url{https://github.com/<ORG>/<REPO>}, with the version used for this work
archived at \texttt{DOI:<Zenodo-DOI>}. Build instructions, backend selection,
and examples are provided in the repository.

\paragraph{Reproducibility.}
All numerical results in Section~\ref{sec:results} can be reproduced from the
released \texttt{examples/} directory. Each example contains the C++ driver
used to generate the data and the Python script
\texttt{make\_paper\_figs.py} used to regenerate the corresponding figures.
Table~\ref{tab:examples} summarizes the mapping between examples and numerical
experiments.

\begin{table}[t]
\centering
\small
\caption{Released examples reproducing the numerical experiments of
Section~\ref{sec:results}.}
\label{tab:examples}
\setlength{\tabcolsep}{5pt}
\begin{tabular}{p{4.8cm}p{5.8cm}}
\hline
Example directory & Reproduced results \\
\hline
\texttt{planar\_sheet\_direct\_vs\_gl} &
Smooth quadrature failure (Figs.~\ref{fig:ps_failure},\ref{fig:ps_geom})\\
\texttt{planar\_sheet\_local\_qbx} &
Panel QBX and orphan fallback (Figs.~\ref{fig:qbx_orphan},\ref{fig:qbx_handoff})\\
\texttt{toroid\_direct\_vs\_gl} &
Near-surface breakdown of smooth quadrature (Fig.~\ref{fig:coil_breakdown})\\
\texttt{vcp\_qbx\_hybrid\_dipole} &
Rotating-ellipse virtual casing (Figs.~\ref{fig:vcp_geom}--\ref{fig:vcp_anchors})\\
\texttt{malhotra\_protocol} &
QBX versus upsampled trapezoidal quadrature (Figs.~\ref{fig:mal_headtohead},\ref{fig:mal_conv})\\
\texttt{target\_specific\_qbx} &
Evaluation strategies and acceleration (Figs.~\ref{fig:paths_equiv},\ref{fig:paths_cost})\\
\texttt{regcoil}, \texttt{vcp\_qbx\_w7x} &
W7-X REGCOIL and virtual-casing results (Section~\ref{sec:res_w7x})\\
\hline
\end{tabular}
\end{table}

\paragraph{Hardware.}
Timing results in Sections~\ref{sec:res_paths}--\ref{sec:res_gpu} were obtained
on the systems described in Section~\ref{sec:res_gpu}: a dual-socket Intel
Xeon E5-2640~v4 node for OpenMP and an NVIDIA RTX~2080~Ti for CUDA. Absolute
timings are hardware-dependent; the reported speed-ups, error levels, and
cross-backend consistency are the reproducible quantities.

\section{Conclusions}
\label{sec:conclusion}
\textcolor{red}{We have presented AMPERe, an algorithm and code that employs QBX to
evaluate magnetic fields at singular and near-singular target points. Originally
developed in the boundary-integral-equation community, QBX is here adapted to
fusion applications, with particular emphasis on virtual-casing fields and
evaluations close to winding surfaces. The code, however, is
not restricted to fusion applications and can be readily used to other
problems involving near-singular field evaluation. We assessed the accuracy and
performance of AMPERe through a series of numerical experiments designed to
probe the near-surface regime, where conventional quadrature schemes become
inaccurate. By construction, the AMPERe preserves the
divergence-free property of the magnetic field at every expansion order. The
numerical results further demonstrate that AMPERe accurately recovers the
symmetries of the underlying configurations despite being entirely
symmetry-unaware. For vacuum magnetic fields, the curl-free property is
likewise recovered to high accuracy.}

\textcolor{red}{The evaluator incorporates several variants aimed at improving
computational efficiency---including target-driven QBX, target-specific
expansions, and local correction---while retaining the panel-driven
formulation. The latter provides a naturally separable structure that
facilitates future integration with fast multipole methods and
boundary-integral solvers, which we leave to future work.}

\textcolor{red}{Two companion components make the overall workflow
self-contained for fusion applications. CLEBSCH converts VMEC and Boozer
equilibria into format-independent HDF5 inputs compatible with AMPERe, while a
Kokkos implementation of REGCOIL, as a part of AMPERe, computes the winding-surface sheet
currents within the same computational pipeline. The latter reproduces the
reference Fortran implementation to within $0.001\%$ in RMS error.}

\section{Acknowledgments}
\label{sec:ack}
The author H. Betar thanks Mohamed Aziat (Renaissance Fusion), Elizabeth Paul and Alexey Knyazev (Columbia University),
Donald Spong and Yashika Ghai (Oak Ridge National Laboratory), Jacobo Varela
(Institute for Fusion Studies, The University of Texas at Austin), and Daniele
Del Sarto (Universit\'e de Lorraine) for fruitful discussions that allowed the
improvement of this work.

\appendix

\section{Analytical configurations and reference fields}
\label{app:configs}

This appendix gives the exact surface parametrizations and closed-form reference
fields of the test configurations of Section~\ref{sec:setup}, together with the
Stratton--Chu trace construction that turns a reference field into the surface
data AMPERe evaluates. These, with the parameter values quoted in
Section~\ref{sec:setup}, fix each test unambiguously; the released code and data
(Section~\ref{sec:software}) provide the reference implementation.

\subsection{Surface parametrizations}
\label{app:surfaces}

\paragraph{Planar sheet.}
A flat square patch in the plane $z=0$, occupying $(x,y)\in[-L,L]^2$ with $L=5$,
carrying a uniform surface current $\mathbf{K}=K_0\,\hat{\mathbf{x}}$.

\paragraph{Torus (toroidal coil).}
A torus of major radius $R_0$ and minor radius $a$ carrying a purely poloidal
surface current of total current $I$, parametrized by $(\theta,\varphi)\in
[0,2\pi)^2$ as
\begin{equation}\label{eq:torus_map}
\mathbf{x}(\theta,\varphi)=\big((R_0+a\cos\varphi)\cos\theta,\ (R_0+a\cos\varphi)\sin\theta,\ a\sin\varphi\big),
\end{equation}
with $R_0=1$, $a=0.3$, $I=1$.

\paragraph{Rotating ellipse.}
The rotating ellipse of Malhotra et al.~\cite{malhotra2019} (their \textcolor{red}{Eq.~(55)}):
an elliptical cross-section of semi-axes $a,b$ whose orientation rotates at rate
$\kappa$ as it is swept around the torus. With $(\theta,\varphi)\in[0,2\pi)^2$,
\begin{equation}\label{eq:ellipse_map}
\begin{aligned}
R(\theta,\varphi) &= R_0 + a\cos(\kappa\theta)\cos\varphi + b\sin(\kappa\theta)\sin\varphi,\\
Z(\theta,\varphi) &= -\,a\sin(\kappa\theta)\cos\varphi + b\cos(\kappa\theta)\sin\varphi,\\
\mathbf{x}(\theta,\varphi) &= \big(R\cos\theta,\ R\sin\theta,\ Z\big),
\end{aligned}
\end{equation}
with $\kappa=3/2$ (the half-integer rate closes the surface, which is
$\pi$-symmetric). Two parameter sets are used: $(R_0,a,b)=(3,\,0.6,\,0.3)$ for
the dipole configuration and $(R_0,a,b)=(2,\,0.7,\,1.0)$ for the Malhotra
benchmark. The tangents, normal, and surface Jacobian follow analytically from
\eqref{eq:ellipse_map}. The W7-X last closed flux surface is instead supplied as
a Fourier boundary from the equilibrium pipeline (Section~\ref{sec:software});
it has no closed-form reference field.

\subsection{Reference fields}
\label{app:fields}

\paragraph{Planar sheet.}
For $\mathbf{K}=K_0\hat{\mathbf{x}}$ the field is curl-free off the sheet and,
writing the corner abbreviations $u_i=X_i-x$, $v_j=Y_j-y$,
$\rho_{ij}=\sqrt{u_i^2+v_j^2+z^2}$ with $(X_1,X_2)=(Y_1,Y_2)=(-L,L)$, is given in
closed form by
\begin{equation}\label{eq:sheet_field}
B_x=0,\qquad
B_y=-\frac{\mu_0}{4\pi}K_0\,\Omega,\qquad
B_z=+\frac{\mu_0}{4\pi}K_0\!\!\sum_{i,j\in\{1,2\}}\!\! s_{ij}\,
\operatorname{arcsinh}\!\frac{u_i}{\sqrt{v_j^2+z^2}},
\end{equation}
where $\Omega=\sum_{i,j\in\{1,2\}} s_{ij}\,\operatorname{atan2}\!\big(u_i v_j,\,
z\,\rho_{ij}\big)$ is the solid angle subtended by the patch and the sign pattern
is $s_{11}=s_{22}=+1$, $s_{12}=s_{21}=-1$.

\paragraph{Toroidal coil.}
The poloidal surface current produces the ideal toroidal-solenoid field, purely
azimuthal and confined to the bore,
\begin{equation}\label{eq:torus_field}
\mathbf{B}=B_\varphi\,\hat{\mathbf{e}}_\varphi,\qquad
B_\varphi=
\begin{cases}
-\,\dfrac{\mu_0 I}{2\pi R}, & R_0-a<R<R_0+a,\\[4pt]
0, & \text{otherwise},
\end{cases}
\end{equation}
with $R=\sqrt{x^2+y^2}$ the cylindrical radius; in the normalized units used here
($\mu_0/4\pi=1$, $I=1$) the interior field is $|B_\varphi|=2/R$.

\paragraph{Magnetic dipole.}
A point dipole of moment $\mathbf{m}$ at $\mathbf{r}_d$ produces, with
$\mathbf{s}=\mathbf{r}-\mathbf{r}_d$,
\begin{equation}\label{eq:dipole_field}
\mathbf{B}(\mathbf{r})=\frac{\mu_0}{4\pi}
\left[\frac{3(\mathbf{m}\cdot\mathbf{s})\,\mathbf{s}}{|\mathbf{s}|^{5}}
-\frac{\mathbf{m}}{|\mathbf{s}|^{3}}\right],
\qquad \mathbf{r}_d=(3,0,0),\ \ \mathbf{m}=\hat{\mathbf{z}} .
\end{equation}

\paragraph{Current loop.}
A closed filamentary loop $C$ carrying current $I$ produces the vacuum
($\lambda=0$) Biot--Savart field
\begin{equation}\label{eq:loop_field}
\mathbf{B}_0(\mathbf{r})=\frac{\mu_0 I}{4\pi}\oint_C
\frac{\mathrm{d}\boldsymbol{\ell}'\times(\mathbf{r}-\mathbf{r}')}
{|\mathbf{r}-\mathbf{r}'|^{3}},
\qquad
\mathbf{c}(t)=\mathbf{c}_0+\rho\,(\cos t\,\mathbf{u}+\sin t\,\mathbf{v}),
\end{equation}
where the loop is the circle of radius $\rho=5$ in the tilted plane spanned by
the orthonormal pair $\{\mathbf{u},\mathbf{v}\}$, chosen so that the ellipse is
not a flux surface of $\mathbf{B}_0$. Since $C$ is well separated from the
surface, the loop integral is evaluated by the periodic trapezoidal rule to
$\sim16$ digits, providing a closed-form-quality ground truth (the protocol
of~\cite{malhotra2019}).

\subsection{Stratton--Chu trace construction}
\label{app:traces}

Each virtual-casing test converts a reference field $\mathbf{B}_0$, known on the
source surface $S$ with outward unit normal $\hat{\mathbf{n}}$, into the
equivalent surface densities of \eqref{eq:Ksigma},
\begin{equation}\label{eq:traces_app}
\mathbf{K}=\tfrac{1}{\mu_0}\,\hat{\mathbf{n}}\times\mathbf{B}_0,
\qquad
\sigma=\hat{\mathbf{n}}\cdot\mathbf{B}_0 .
\end{equation}
Writing $\mathbf{B}_K$ and $\mathbf{B}_\sigma$ for the surface-current and
surface-charge integrals of \eqref{eq:VCP}, the field is reconstructed in the
extinction sign convention
\begin{equation}\label{eq:extinction}
\mathbf{B}_{\rm VCP}(\mathbf{r})=-\mathbf{B}_K(\mathbf{r})-\mathbf{B}_\sigma(\mathbf{r})
=\begin{cases}\mathbf{B}_0(\mathbf{r}), & \mathbf{r}\in\Omega\ (\text{inside }S),\\
\mathbf{0}, & \mathbf{r}\notin\Omega,\end{cases}
\end{equation}
in the exact-quadrature limit. On the rotating ellipse neither trace vanishes
($S$ is not a flux surface of $\mathbf{B}_0$), so both the $\mathbf{K}$-term ---
the structure-preserving piece of Section~\ref{sec:structure} --- and the
$\sigma$-term are exercised; on a genuine flux surface $\sigma\equiv0$ and only
the $\mathbf{K}$-term remains.

\section{AMPERe-REGCOIL}
\label{app:REGCOIL_Kokkos}

The stellarator sheet-current source used in Section~\ref{sec:res_w7x} is
generated by a Kokkos reimplementation, internal to AMPERe, of
\textsc{regcoil}~\cite{landreman2017regcoil}.

\textsc{regcoil}~\cite{landreman2017regcoil} designs a stellarator coil set as
a smooth, divergence-free sheet current on a toroidal \emph{winding surface}
$S_c$ surrounding the plasma boundary $S_p$ (the LCFS), through a
single-valued current potential $\Phi$,
\begin{equation}\label{eq:reg_phi}
\mathbf{K}=\hat{\mathbf{n}}\times\nabla\Phi,
\qquad
\Phi(\theta,\zeta)=\frac{G\,\zeta}{2\pi}+\frac{I\,\theta}{2\pi}
+\sum_{j=1}^{m_{\rm nmax}}\Phi_j\,\sin(m_j\theta-n_j\zeta),
\end{equation}
with $(\theta,\zeta)$ poloidal/toroidal angles on $S_c$, $G$ the net poloidal
current, and $I$ the net toroidal current (zero for a stellarator). The
coefficients $\Phi_j$ minimize
\begin{equation}\label{eq:reg_obj}
\chi^2_B+\lambda\,\chi^2_K,
\qquad
\chi^2_B=\oint_{S_p} w_p\,(\mathbf{B}\!\cdot\!\hat{\mathbf{n}})^2 ,
\qquad
\chi^2_K=\oint_{S_c} |\mathbf{K}|^2\,dA ,
\end{equation}
where $\chi^2_B$ [T$^2$m$^2$] is the residual normal field the coils leave on
the plasma boundary and $\chi^2_K$ [A$^2$] penalizes coil complexity;
$\lambda\to0$ is the unregularized NESCOIL~\cite{merkel1987} limit, and the
knee of the $(\chi^2_K,\chi^2_B)$ L-curve is the usual compromise. Since
$\mathbf{B}\!\cdot\!\hat{\mathbf{n}}$ is linear in $\boldsymbol\Phi$ through a
Biot--Savart coupling between $S_c$ and $S_p$, both terms are quadratic and
$\partial(\chi^2_B+\lambda\chi^2_K)/\partial\boldsymbol\Phi=0$ reduces the
problem, at every $\lambda$, to the dense symmetric positive-definite normal
equations
\begin{equation}\label{eq:reg_normal}
\big(A+\lambda\,C\big)\,\boldsymbol\Phi(\lambda)=-\big(\mathbf{dB}+\lambda\,
\mathbf{dK}\big),
\end{equation}
of size $m_{\rm nmax}\times m_{\rm nmax}$ (full expressions
in~\cite{landreman2017regcoil}). AMPERe solves \eqref{eq:reg_normal} by a
dense Cholesky factorization implemented entirely in-house, without external
linear-algebra dependencies, for each $\lambda$ in a user-specified grid. The
resulting sheet-current field is then evaluated on the plasma boundary or at
arbitrary interior and exterior targets using the hybrid QBX evaluator of
Sections~\ref{sec:qbx}--\ref{sec:acceleration}.

REGCOIL is implemented directly in Kokkos rather than invoked as an external
Fortran executable, making the coil-design stage performance-portable across
the same CPU and GPU backends as the remainder of AMPERe
(Section~\ref{sec:accel_kokkos}). This design provides a unified workflow in
which sheet-current construction and field evaluation operate on the same
panelized surface and share the same QBX infrastructure, eliminating
dependencies on VMEC/netCDF during evaluation as well as intermediate file
formats or external executables.

\subsection{Benchmark against standalone REGCOIL}
\label{app:reg_bench}

The Kokkos solver is validated against the standalone Fortran implementation
on the W7-X vacuum equilibrium of Section~\ref{sec:res_w7x}, using identical
plasma and winding surfaces to isolate solver differences. Three cases are
compared: \emph{Fortran}, the reference implementation; \emph{Kokkos-A}, the
AMPERe solver on the Fortran winding surface; and \emph{Kokkos-B}, the AMPERe
solver on the uniformly offset surface used in production
(Section~\ref{sec:res_w7x}). Figure~\ref{fig:app_regcoil_chi2} shows the two
objective terms in \eqref{eq:reg_obj} versus $\lambda$ at $M=24$. The two
implementations produce indistinguishable L-curves, with Kokkos-A differing
from Fortran by median relative errors of $0.11\%$ in $\chi^2_B$ and
$0.003\%$ in $\chi^2_K$ over the full scan. Kokkos-B differs only in the
weakly regularized limit ($\lambda\rightarrow0$), where the objective reaches
its numerical floor and the effect of the modified winding geometry becomes
visible. The resulting fields reproduce the VMEC benchmark of
Section~\ref{sec:res_w7x}, while the direct Fortran--Kokkos comparison gives
component-wise differences of $0.001$--$0.003\%$.

The same solver has also been extended to axisymmetric
tokamak equilibria with net plasma current and validated on a JET equilibrium;
these results are reported elsewhere.

\begin{figure}[t]
  \centering
  \includegraphics[width=0.48\linewidth]{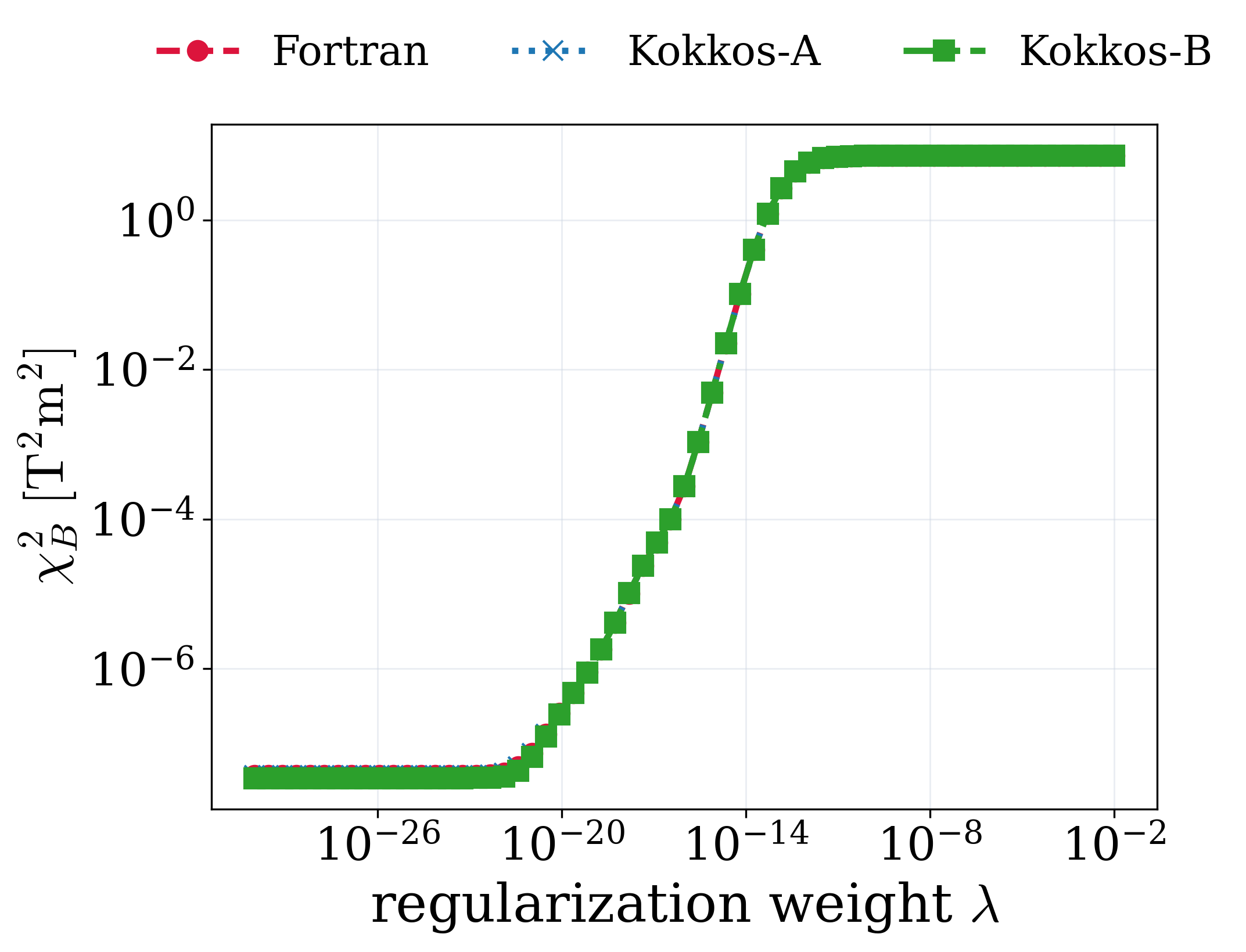}
  \hfill
  \includegraphics[width=0.48\linewidth]{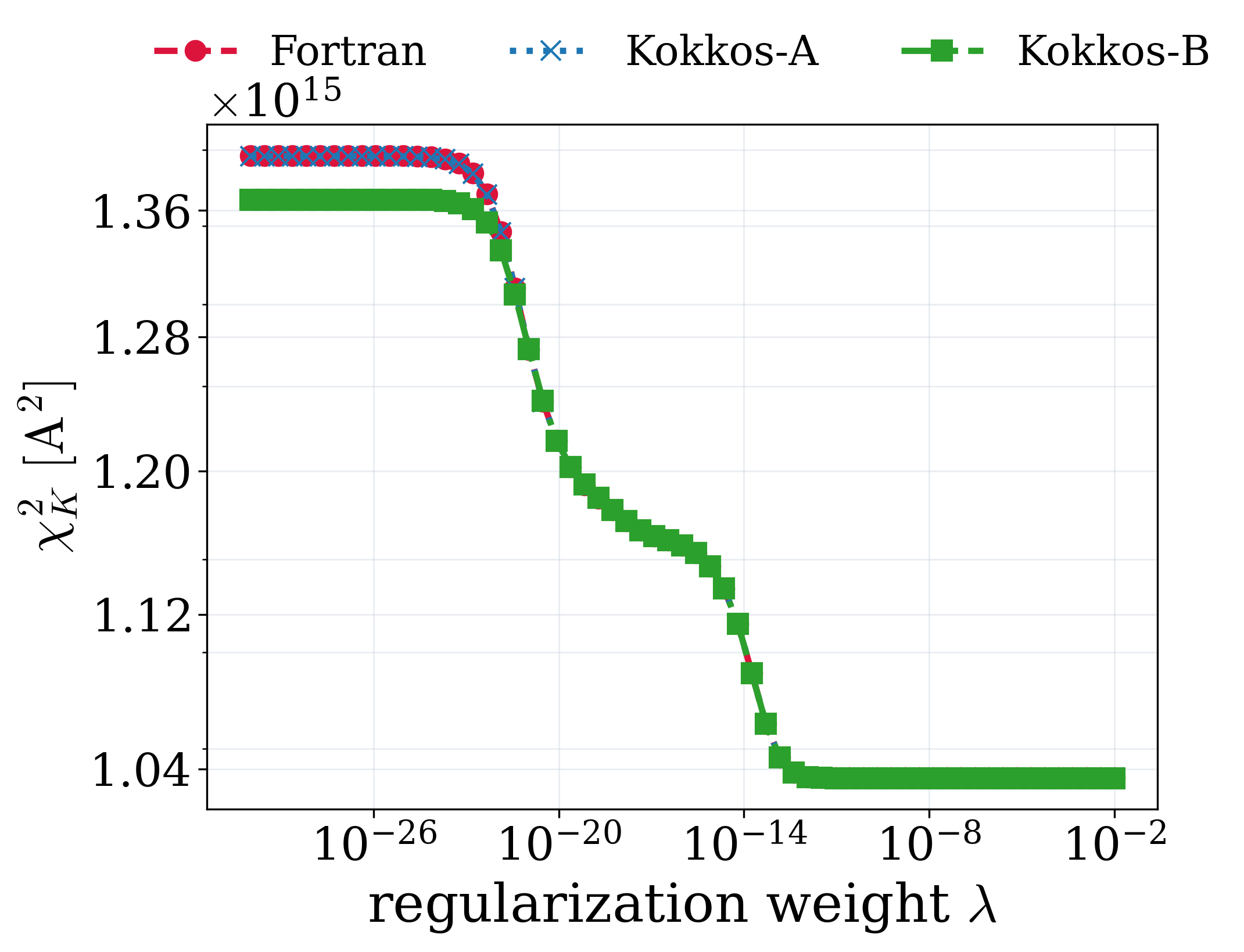}
\caption{REGCOIL objective terms from \eqref{eq:reg_obj} versus regularization
weight $\lambda$ at $M=24$ for the Fortran, Kokkos-A (Fortran winding
surface), and Kokkos-B (AMPERe offset surface) solutions. \emph{Left:}
normal-field residual $\chi^2_B$. \emph{Right:} coil-complexity penalty
$\chi^2_K$.}
  \label{fig:app_regcoil_chi2}
\end{figure}

\section{CLEBSCH Kokkos}
\label{app:CLEBSCH_Kokkos}

CLEBSCH is the companion code that supplies every equilibrium-derived input
required by AMPERe in this paper (Sections~\ref{sec:res_w7x}
and~\ref{app:REGCOIL_Kokkos}). It is named for the Clebsch representation of a
magnetic field that underlies straight-field-line flux coordinates,
\begin{equation}\label{eq:clebsch}
\mathbf{B}=\nabla\psi\times\nabla\alpha,
\qquad
\alpha=\theta^{*}-\iota(\psi)\,\zeta,
\end{equation}
with $\psi$ labeling the nested flux surfaces and $\alpha$ the field-line
label built from the straight-field-line poloidal angle $\theta^{*}$, the
rotational transform $\iota$, and the toroidal angle $\zeta$.

Starting from a VMEC~\cite{hirshman1983vmec} equilibrium (readers for DESC and
Boozer equilibria share the same interface), CLEBSCH reconstructs real-space
fields, geometry, and metric quantities from the stored Fourier spectrum on
arbitrary grids or single flux surfaces, and perform coordinate transformations.
Only one of its outputs matters for the present paper: the plasma-boundary geometry with
exact tangents, the net poloidal/toroidal currents.

Like AMPERe, CLEBSCH is a standalone, header-only, open-source Kokkos
library, usable independently of AMPERe for equilibrium reconstruction and
flux-coordinate transformation. It is released under the MIT license at
\url{https://github.com/<ORG>/<CLEBSCH-REPO>}.

\section*{Declaration of generative AI and AI-assisted technologies}
During the preparation of this work the author used (H. Betar) Google Gemini and
Anthropic Claude in order to refactor and improve an earlier version of the
AMPERe source code, and to check the manuscript for clarity, and
typographical errors. After using these tools, the author reviewed and edited
the content as needed and takes full responsibility for the content of the
publication.

\bibliographystyle{elsarticle-num}
\bibliography{AMPERe}

\end{document}